\documentclass[aps,pra,groupedaddress,tightenlines,longbibliography,notitlepage,nofootinbib,onecolumn,superscriptaddress,11pt]{revtex4-1}
\usepackage{amsmath,amssymb,graphicx,amsmath,amssymb,amsthm,physics,bm,bbm,color,mathtools,anyfontsize}
\usepackage[sans]{dsfont}
\usepackage[a4paper]{geometry}
\usepackage[colorlinks=true, allcolors=blue]{hyperref}

\usepackage{newpxtext,newpxmath}
\def\tr{\operatorname{tr}}
\def\id{\operatorname{id}}
\newcommand \inner[1]{ \left<#1\right>}
\renewcommand {\spadesuit} { \mathring{\mathbf \Lambda}}
\newcommand{\overbar}[1]{\mkern 2mu\overline{\mkern-2mu#1\mkern-2mu}\mkern 2mu}
\newcommand{\dt}{\mathbf{D}^{(2)}}
\newcommand{\rt}{\mathcal{R}^{(2)}}
\newcommand{\wrt}{\widetilde{\mathcal{R}}^{(2)}}

\DeclarePairedDelimiterX{\channels}[2]{[}{]}{%
  #1\;\delimsize\|\;#2%
}
\DeclarePairedDelimiterX{\states}[2]{(}{)}{%
  #1\;\delimsize\|\;#2%
}

\newcommand{\staD}{D\states}
\newcommand{\chD}{D\channels}
\newcommand{\schD}{D^{(2)}\channels}
\newcommand{\BstD}{\mathbf{D}\states}

\newcommand{\BchD}{\mathbf{D}\channels}
\newcommand{\BschD}{\mathbf{D}^{(2)}\channels}

\newcommand{\BnchD}{\mathbf{D}^{(n)}\channels}
\newcommand{\BnachD}{\mathbf{D}^{(n-1)}\channels}
\newcommand{\BnonechD}{\mathbf{D}^{(n+1)}\channels}
\newcommand{\Cwhat}{\widehat{\mathsf{C}}}
\newcommand{\Cstat}{\mathsf{C}}

\DeclareOldFontCommand{\rm}{\normalfont\rmfamily}{\mathrm}

\newtheorem{theorem}{Theorem}
\newtheorem{proposition}{Proposition}%

\newtheorem{lemma}{Lemma}
\newtheorem*{example*}{Example}%
\newtheorem{remark}{Remark}%
\newtheorem*{remark*}{Remark}%
\newtheorem{dfn}{Definition}%

\begin{document}

\title{Fundamental limitations on the recoverability of quantum processes}

\author{Sohail}
\email{sohail.sohail@ttu.edu}
\email{sohail.malda@gmail.com}
\affiliation{Department of Computer Science, Texas Tech University, Lubbock, TX 79409, United States}
\affiliation{Harish-Chandra Research Institute, A CI of Homi Bhabha National Institute, Chhatnag Road, Jhunsi, Prayagraj 211019, India}

\author{Vivek Pandey}
\email{vivekpandey3881@gmail.com}
\affiliation{Harish-Chandra Research Institute, A CI of Homi Bhabha National Institute, Chhatnag Road, Jhunsi, Prayagraj 211019, India}
\affiliation{Institute of Physics, Faculty of Physics, Astronomy and Informatics, Nicolaus Copernicus University, Grudziadzka 5/7, 87-100 Toru\'n, Poland
}

\author{Uttam Singh}
\email{uttam@iiit.ac.in}
\affiliation{Centre for Quantum Science and Technology (CQST), International Institute of Information Technology Hyderabad, Gachibowli 500032, Telangana, India}

\author{Siddhartha Das}
\email{das.seed@iiit.ac.in}
\affiliation{Center for Security, Theory and Algorithmic Research (CSTAR), International Institute of Information Technology Hyderabad, Gachibowli 500032, Telangana, India}
\affiliation{Centre for Quantum Science and Technology (CQST), International Institute of Information Technology Hyderabad, Gachibowli 500032, Telangana, India}

\begin{abstract}
Quantum information processing and computing tasks can be understood as quantum networks, comprising quantum states and channels and possible physical transformations on them. It is hence pertinent to estimate the change in informational content of quantum processes due to physical transformations they undergo. The physical transformations of quantum states are described by quantum channels, while the transformations of quantum channels are described by quantum superchannels. In this work, we determine fundamental limitations on how well the physical transformation on quantum channels can be undone or reversed, which are of crucial interest to design and benchmark quantum information and computation devices. In particular, we refine (strengthen) the quantum data processing inequality for quantum channels under the action of quantum superchannels. We identify a class of quantum superchannels, which appears to be the superchannel analogue of subunital  quantum channels, under the action of which the entropy of an arbitrary quantum channel is nondecreasing. We also provide a refined inequality for the entropy change of quantum channels under the action of an arbitrary quantum superchannel.
\end{abstract}

\maketitle

\section{Introduction}

\subsection{Motivation and background}
Quantum theory lies at the core of our understanding of the natural phenomena at the small scale and plays a fundamental role in the rapidly growing area of quantum information processing and computing. Traditionally, quantum states are used as the information carriers, and they need to be appropriately manipulated to perform the desired protocols or algorithms, e.g., quantum key distribution~\cite{XMZ+20,DBWH21} and protocols implying the advantage in quantum computation~\cite{shor94,grover96,BGK18}; these manipulations are quantum state transformations. The physical transformations of quantum states of systems are described by completely positive (CP), trace-preserving (TP) maps in quantum theory and are referred to as quantum channels. Achieving desired manipulations of quantum states requires precise application of specific quantum channels. A thorough understanding of quantum channels, therefore, is not only essential for the theoretical insights but also pivotal for the realization of quantum computers, sensing and information processing devices, and quantum Internet~\cite{NC00, DM03, H13book, SSDBG2015, wilde2016,SSHD23}.

 Mathematically, quantum channels can be thought to encompass the notions of physical evolution, measurement, discarding of a quantum system, appending of a quantum system, and density operators (states)\footnote{We direct the readers to~\cite[Section 4.6]{wilde2016} for detailed discussion.}. This motivates for a convenient and overarching framework to handle quantum states, measurements, and quantum channels succinctly and efficiently in a unified way. Quantum networks comprising quantum states, measurements, and quantum channels provide such a framework~(cf.~\cite{DBWH21}), where information can be encoded in quantum channels and/or quantum states depending on the desired tasks~\cite{SNB+06,Das19,DW19,PLY+23}; then the manipulation of information can be treated in a unified way by the higher order transformations of quantum networks. Considering quantum networks as elementary quantum objects, their transformations would then describe the most general kind of physical transformations of quantum information. Examples of transformations of quantum channels include among others the programmable quantum channels (taking quantum state as input and outputting a quantum channel) \cite{DP05,GBW21} and the cloning of unknown unitaries~\cite{CDGP08}. The transformations of quantum channels are achieved by the action of quantum superchannels--- linear transformations that map quantum channels into quantum channels~\cite{Giulio2008, Chiribella,Zyc08, Per17,Bisio_2019}. In this work, we shed light on the irreversibility of the transformation of quantum channels under the action of quantum superchannels. In particular, we derive entropic inequalities that provide non-trivial lower bounds on the entropy change and data processing inequality pertaining to the transformation of quantum channels under the action of quantum superchannels.

The entropic functionals attributed to quantum channels, together with their inherent properties, furnishes both qualitative and quantitative methodologies to analyse quantum processes from the perspective of quantum information theory. Entropic inequalities, like uncertainty relations~\cite{CBTW17}, second laws of thermodynamics~\cite{BHN+15}, data processing inequality~\cite{Rus02}, etc.~have been instrumental in devising quantum communication and computation protocols. The quantum relative entropy $\staD*{\cdot}{\cdot}$, a quantum generalization of Kullback-Leibler divergence, is one of the core entropic functions and is indispensable to the analysis of distinguishability problems in quantum information theory. $\staD*{\rho}{\sigma}$ between a state $\rho$ of an arbitrary dimensional quantum system, finite or infinite, and a positive semidefinite operator $\sigma$ is defined as~\cite{Umegaki1962}\footnote{The logarithm can be taken with respect to the base $2$ without loss of generality.}
 \begin{equation}
    \staD*{\rho}{\sigma} := 
       \begin{cases}
                \tr(\rho(\log{\rho}-\log{\sigma})) &  \text{if}\  \operatorname{supp}(\rho)\subseteq \operatorname{supp}(\sigma), \\ 
                +\infty & \text{otherwise}.
       \end{cases}
 \end{equation}
 Many other important entropic quantities, like the von Neumann entropy, quantum mutual information, and quantum conditional entropy, can be expressed in terms of the quantum relative entropy. For instance, the von Neumann entropy $S(\rho):= -\tr[\rho\log\rho]$ of a quantum state $\rho$ is $S(\rho)= - \staD*{\rho}{\mathbbm{1}}$, where $\mathbbm{1}$ is the identity operator. The monotonicity of quantum relative entropy between states under the action of quantum channels, also referred to as the quantum data processing inequality, states that the quantum relative entropy between arbitrary states $\rho$ and $\sigma$ is nonincreasing under the action of an arbitrary quantum channel $\mathcal{N}$~\cite{Petz:1986, Petz:1988}, i.e.,
 \begin{equation}\label{eq:data_proc}
     \staD*{\rho}{\sigma}-\staD*{\mathcal{N}(\rho)}{\mathcal{N}(\sigma)}\geq 0.
 \end{equation}
The quantum data processing inequality has several ramifications in quantum information theory and is one of the widely invoked inequalities to derive limitations in quantum processing and computing tasks~\cite{Rus02,H13book,BSFPW2017, DBWH21,SSC2023,KW20}. A consequence of the above inequality is a second law of thermodynamics like inequality which states that the von Neumann entropy~\footnote{We simply refer to the von Neumann entropy functional as the entropy. We also discuss generalized entropy functions. The use of the entropy for the von Neumann entropy will be clear from the context.} of a quantum state is nondecreasing under the action of a unital channel~\cite{Alberti:1977wc,AU82}, i.e., $S(\mathcal{N}(\rho))-S(\rho)\geq 0$ for any quantum channel $\mathcal{N}$ that preserves the identity operator, i.e., $\mathcal{N}(\mathbbm{1})=\mathbbm{1}$, and arbitrary quantum state $\rho$. It is known that the von Neumann entropy of any quantum state is also nondecreasing for subunital channels~\cite{Alberti:1977wc,AU82}, meaning that  $S(\mathcal{N}(\rho))-S(\rho)\geq 0$ for quantum channels $\mathcal{N}$ satisfying $\mathcal{N}(\mathbbm{1})\leq \mathbbm{1}$. The input and output dimensions of unital channels are the same, whereas the output dimension of a subunital channel is greater than or equal to the input dimension.

The quantum data processing inequality (Eq.~\eqref{eq:data_proc}) for states is saturated, i.e., $ \staD*{\rho}{\sigma}=\staD*{\mathcal{N}(\rho)}{\mathcal{N}(\sigma)}$, if and only if there exists a completely positive, trace nonincreasing map $\mathcal{P}$, called the Petz recovery map, such that $\mathcal{P} \circ \mathcal{N} (\rho)= \rho$ and $\mathcal{P} \circ \mathcal{N} (\sigma)= \sigma$~\cite{Petz:1986, Petz:1988}. This result has found crucial applications in quantum error correction~\cite{BK02,MN12} among others. The quantum data processing inequality has been recently improved to the following. For arbitrary quantum states $\rho$ and  $\sigma$, and an arbitrary quantum channel $\mathcal{N}$, there exists a universal recovery map $\mathcal{P}^{\rm R}_{\sigma,\mathcal{N}}$ (a completely positive and trace nonincreasing map that depends only on $\sigma$ and $\mathcal{N}$ with no dependence on $\rho$) such that~\cite{Junge_2018}
 \begin{equation}
     \staD*{\rho}{\sigma}-\staD*{\mathcal{N}(\rho)}{\mathcal{N}(\sigma)}\geq -\log F\left(\rho,\mathcal{P}^{\rm R}_{\sigma,\mathcal{N}}\circ\mathcal{N}(\rho)\right)\geq 0,
 \end{equation}
 where $F(\rho,\sigma):=\norm{\sqrt{\rho}\sqrt{\sigma}}_1^2$ is the fidelity between $\rho$ and $\sigma$. $\norm{\mathsf{X}}_{1}=\tr{\sqrt{\mathsf{X}^\ast \mathsf{X}}}$ is the trace-norm of $\mathsf{X}$. There is also a strengthened (refined) inequality~\cite{Buscemi_2016} for entropy change of an arbitrary state under the action of an arbitrary positive and trace-preserving map $\mathcal{N}$,
 \begin{equation}\label{eq:bdw_main}
     S\left(\mathcal{N}(\rho)\right)-S(\rho)\geq \staD*{\rho}{\mathcal{N}^{\ast}\circ\mathcal{N}(\rho)},
 \end{equation}
 where $\mathcal{N}^\ast$ denotes the adjoint map of the linear map $\mathcal{N}$. The above inequality is a strengthened version of entropy change under quantum channels as quantum channels are special cases of positive and trace-preserving maps. Notice that $\staD*{\rho}{\mathcal{N}^\ast\circ\mathcal{N}(\rho)}\geq 0$ for all subunital channels $\mathcal{N}$ and quantum states $\rho$, thus $\mathcal{N}^\ast$ can be deemed as a universal recovery map (channel) for unital channels $\mathcal{N}$ in the context of entropy change.

The relative entropy between quantum channels~\cite{Cooney_2016,Felix_2018,Yuan} is also defined in a way similar to the diamond norm~\cite{watrous2004} of two quantum channels where the trace-distance is replaced by the relative entropy. Such a notion of relative entropy between quantum channels is nondecreasing under the action of quantum superchannels~\cite{Felix_2018,Hirche_2023}. Motivated by the definition of the entropy for quantum states, \cite{Gour_2019} proposed axioms that the definition of the entropy of a quantum channel should satisfy. In \cite{Yuan,Gour_entropy}, the definition of quantum entropy was introduced which satisfies the physically motivated axioms mentioned in \cite{Gour_2019}; the entropy of a quantum channel is equal to its completely bounded entropy introduced in~\cite[Eq.~(1.1)]{DJKR06}~\cite[Proposition 6]{Gour_entropy}. The entropy of a quantum channel finds operational meaning in the information-theoretic tasks of quantum channel merging~\cite{Gour_entropy}, hypothesis testing of quantum channels~\cite{Yuan}, and distributed private randomness distillation~\cite{YHW19}.
\color{black}

\textit{Note}.--- Some ideas discussed here were instrumental in defining information-theoretic quantities like the  conditional entropy, mutual information, and conditional mutual information of multipartite quantum channels in \cite{DGP24}.

\subsection{Our contribution and main results}
\label{summary}
In this paper, our main objective is to determine the fundamental limitations on the reversibility of entropy change and data processing of quantum channels under the action of quantum superchannels. Theorems~\ref{theorem_3}~and~\ref{theorem_4} comprise our main results, which can be considered as the generalizations of inequalities~\eqref{eq:bdw_main} and \eqref{eq:data_proc}, respectively, to the case of quantum channels and their transformations under quantum superchannels. To arrive at these results, we prove several meta results that are inequalities for the entropic functionals of quantum channels, in particular, see for instance Propositions~\ref{equ:prop1} and \ref{equ:prop_new}, which may be of independent interest. We employ an alternate expression (Definitions~\ref{dfn:entropy_channel} and \ref{dfn:gen_ent_channel}) for the definition of the entropy of quantum channel that works for arbitrary dimensional quantum channels, discrete-variable or continuous-variable systems, and can also be extended to the definition of entropy for arbitrary higher order quantum processes. We note that our definition of the entropy for quantum channels coincides with the definition of the entropy for quantum channels proposed in \cite{Gour_entropy} for all practical purposes (cf.~\cite[Definitions 1 and 12]{Gour_entropy}). Our expression for the entropy of quantum channel is the generalization of the expression for the entropy of quantum state $\rho_A$ defined as the negative of the relative entropy $\staD*{\rho_A}{\mathbbm{1}_A}$ of $\rho_A$ with respect to $\mathbbm{1}_A$. Note that $\mathbbm{1}_A$ is not a trace-class operator but a bounded operator that is well-defined for both finite- and infinite-dimensional separable Hilbert spaces.

We introduce $\mathcal{R}$-subpreserving supermaps $\hat{\Theta}:\mathcal{L}(A,B)\to\mathcal{L}(C,D)$ that are defined as supermap such that $\hat{\Theta}(\mathcal{R}_{A\to B})\leq \mathcal{R}_{C\to D}$, where $\mathcal{R}_{A\to B}$ and $\mathcal{R}_{C\to D}$ are completely depolarizing maps (Definition~\ref{def_completely_depolarising_map}). The entropy $S[\mathcal{N}]$ of an arbitrary quantum channel $\mathcal{N}_{A\to B}$ is then defined as $S[\mathcal{N}]:= -\chD*{\mathcal{N}}{ \mathcal{R}}$, i.e., negative of the relative entropy between a quantum channel $\mathcal{N}_{A\to B}$ and the completely depolarizing map $\mathcal{R}_{A\to B}$ (see Definition~\ref{dfn:entropy_channel}). We prove that the entropy of a quantum channel is nondecreasing under the action of an $\mathcal{R}$-subpreserving superchannel $\hat{\Theta}:\mathcal{L}(A,B)\to\mathcal{L}(C,D)$ (Proposition~\ref{prop:entropy_gain}), i.e.,
\begin{equation}\label{eq:ent_gain_1}
    S[\hat{\Theta}({\mathcal{N}})]-S[\mathcal{N}]\geq 0.
\end{equation}
A direct consequence of the above inequality is that the entropy of a quantum channel is nondecreasing under the action of an $\mathcal{R}$-preserving superchannel $\hat{\Theta}$, i.e., $\hat{\Theta}(\mathcal{R}_{A\to B})=\mathcal{R}_{C\to D}$, where we only require $\abs{B}=\abs{D}$ (cf.~\cite[Proposition~3]{Gour_2019}). We advocate  $\mathcal{R}$-subpreserving and $\mathcal{R}$-preserving maps (channels) to be the supermaps analogue of subunital and unital maps (channels), respectively, in the sense that a subunital map $\mathcal{M}$ is $\mathbbm{1}$-subpreserving, i.e., $\mathcal{M}_{A\to B}(\mathbbm{1}_A)\leq \mathbbm{1}_B$. Furthermore, analogous to inequality~\eqref{eq:bdw_main}, we derive a lower bound on the entropy change of an arbitrary quantum channel under the action of a quantum superchannel in Theorem~\ref{theorem_3}, which also strengthens the inequality~\eqref{eq:ent_gain_1}.

We show that the relative entropy of quantum channels is nonincreasing under the action of a wide class of CP-preserving and TP-preserving supermaps~(see Proposition~\ref{equ:prop_new}).
We strengthen the data processing inequality for quantum channels by deriving nontrivial lower bound to the monotonicity of relative entropy of quantum channels under the action of quantum superchannels. That is, the following lower bound holds for the monotonicity of the relative entropy $\chD*{\mathcal{N}}{\mathcal{M}}$ between quantum channels $\mathcal{N}$ and $\mathcal{M}$ with a desirable property under the action of a quantum superchannel $\Theta$, 
 \begin{equation}
        \chD*{\mathcal{N}}{ \mathcal{M}} -   \chD*{\Theta\left(\mathcal{N}\right) }{ \Theta \left(\mathcal{M}\right)} \geq - \log F \left(\mathsf{C}^{\Psi}_{\mathcal{N}},\left(\mathcal{P}^{\rm R}\circ \mathfrak{T}'\right)\left(\mathsf{C}^{\Psi}_{\mathcal{N}}\right)\right) \geq 0,
    \end{equation}
where the fidelity $F$ is taken between the operators related to the quantum channel $\mathcal{N}$ and the operators related to the reversibility of the action of the quantum superchannel $\Theta$ on $\mathcal{N}$ and $\mathcal{M}$. For a precise statement and a condition on the pair of channels $\mathcal{N}$ and $\mathcal{M}$ for which the above bound holds, see Theorem~\ref{theorem_4}. An important consequence of this result is that for a specific class of quantum channels and a quantum superchannel that preserves this class, the saturation of the data processing inequality implies the existence of a completely CP-preserving supermap that undoes (or reverses) the action of the superchannel on given quantum channels, see Remark~\ref{rem:data-sat}.

Noticing the similarity between the definitions of entropy for quantum states and quantum channels, we define the generalized divergence and the entropy of quantum superchannels, and subsequently generalize these notions to higher-order quantum processes. By defining the notion of a completely depolarizing supermap, we show that the entropy functional for quantum superchannels is nondecreasing under a quantum super-superchannel that subpreserves the completely depolarizing supermap. Then we prove the subadditivity of the quantum superchannel entropy under tensor products. Further, we show that for a replacer superchannel $\Theta_{0}$, which maps every channel to a fixed channel $\mathcal{N}_{0}$, the entropy of $\Theta_{0}$ is upper bounded by the entropy of $\mathcal{N}_{0}$ to a constant dimension-dependent term.

\subsection{Structure of the paper}
The structure of this paper is as follows. In Section~\ref{prelims&notations}, we fix our notations and present the preliminaries and background necessary to derive the main results of this paper. Here, we discuss some properties of the space of CP maps, quantum channels, and superchannels; and the notion of representing maps for superchannels. In Section~\ref{Entropy_of_quantum_channels}, we define the entropy of a quantum channel via an axiomatic approach, and we also extend the definition of entropy for quantum channels with infinite-dimensional input or output space. We also prove several bounds on channel entropy gain for a wide class of channels including tele-covariant channels. In Section~\ref{Sufficiency_of_superchannels}, we introduce the notion of the sufficiency of superchannels and we prove a stronger version of data processing inequality and entropy gain of quantum channels, with a remainder term which depends on recovery supermap.  In Section~\ref{Generalized_divergence_and_entropy_of_superchannels}, we define the generalized divergence and entropy functional for superchannels and we discuss some properties of both the functionals. Finally, in Section~\ref{conclusion}, we give the concluding remarks of our findings. We also define a divergence functional for the higher order process, and provide a necessary condition for the functional to be a generalized divergence. Using this definition of generalized divergence, we define the entropy functional for higher-order processes. 
\tableofcontents
\section{Preliminaries: Standard notations, definitions, and facts} \label{prelims&notations}

In this section, we briefly mention some standard definitions and notations that are important for deriving and discussing our results.

 The separable Hilbert space associated with a quantum system $A$ is denoted by $A$ and its dimension is denoted by $|A|:={\rm dim}(A)$. The Hilbert space of a joint system $AB$ is given by the tensor product $A\otimes B$ of the separable Hilbert spaces $A$ and $B$. Let $A\simeq B$ denote that the two Hilbert spaces $A$ and $B$ are isomorphic. The space of all linear operators acting on $A$ is denoted by ${\mathcal{L}}(A)$ and the identity operator on $A$ is denoted by $\mathbbm{1}_{A}$. The space $\mathcal{L}\left(A\right)$ is equipped with the Hilbert-Schmidt inner product $\langle \mathsf{P}, \mathsf{Q}\rangle  \equiv \mathrm{tr}\left({\mathsf{P}}^{\dag} \mathsf{Q}\right)$ $\forall ~ \mathsf{P}, \mathsf{Q} \in \mathcal{L}\left(A\right)$. Let $\mathcal{L}\left(A\right)_{+} \subset \mathcal{L}\left(A\right)$ denote the set of all positive semidefinite operators; for simplicity, we will use the terms positive semidefinite and positive interchangeably throughout the paper. Let $\mathcal{L}_{h}\left(A\right)$ denote the subset containing all the  Hermitian operators in $\mathcal{L}\left(A\right)$. There is a natural partial ordering in $\mathcal{L}_{h}\left(A\right)$ as follows: for operators $\{\mathsf{P}, \mathsf{Q}\}\in \mathcal{L}_{h}\left(A\right) $, we say that $\mathsf{P}\geq \mathsf{Q}$ if and only if $(\mathsf{P}- \mathsf{Q})\in \mathcal{L}\left(A\right)_{+}$. The state of a quantum system is described by a positive semidefinite trace-class operator with unit trace, called density operator, and is denoted by $\rho$. Let $\mathcal{D}\left(A\right)$ and $\mathcal{D}(A\otimes B)$ denote the set of all density operators on $A$ and on $A\otimes B$, respectively. The rank-one elements of $\mathcal{D}\left(A\right)$ are called the pure states. With a slight abuse of notation, we may call a normalized vector $\ket{\Psi}\in A$ as a (pure) state whenever it is clear that we are referring to the corresponding rank-one density operator $\Psi_A:= \op{\Psi}_A\in\mathcal{D}(A)$. For the sake of simplicity, we may denote an operator $\rho\in\mathcal{L}(A\otimes B)$ as $\rho_{AB}$. For an operator $\rho_{AB}\in\mathcal{L}(A\otimes B)$, $\rho_A \in\mathcal{L}(A)$ is called the marginal operator for the system $A$ (or reduced operator of the system $A$) and is defined as $\rho_A= \tr_B(\rho_{AB})$. Similarly, one defines $\rho_B=\tr_A(\rho_{AB})$. If $\rho_{AB}\in\mathcal{D}(A\otimes B)$, then the marginals $\rho_A\in\mathcal{D}(A)$ and $\rho_B\in\mathcal{D}(B)$ are called reduced states of $\rho_{AB}$. Let $\mathsf{FRank}\left(A\otimes B\right) \subset \mathcal{L}\left(A\otimes B\right)$ be a set of rank one operators $\Psi_{AB}\in \mathcal{L}\left(A\otimes B\right)$ such that $\mathrm{rank}\left(\Psi_{A}\right) = |A|$, where $A \simeq B$. The maximally entangled operator $\op{\Psi^+}_{AB}:=\Psi^+_{AB} \in \mathcal{L}(A\otimes B)_+$ is defined via $\ket{\Psi^+}_{AB}:=\sum_{i=1}^{d}\ket{i,i}_{AB}$, where $d=\min\{|A|,|B|\}$, and we call the normalized maximally entangled operator $d^{-1}\Psi^+_{AB}\in\mathcal{D}(A\otimes B)$ as a maximally entangled state of $AB$.

 Let 
$\mathcal{L}\left(A,B\right)$
be the space of all linear maps from $\mathcal{L}\left({A}\right)$ to $\mathcal{L}\left({B}\right)$. We denote a linear map $\mathcal{N}\in \mathcal{L}\left(A,B\right)$ as $\mathcal{N}_{A \rightarrow B}$. We will use both notations $\mathcal{N}_{A\to B}$ and $\mathcal{N}:\mathcal{L}(A)\to\mathcal{L}(B)$, interchangeably for a linear map $\mathcal{N}$ whenever there is no ambiguity. A linear map $\mathcal{N}_{A \rightarrow B}$  is called positive if it maps elements of  $ \mathcal{L}\left({A}\right)_{+}$ to the elements of $ \mathcal{L}\left({B}\right)_{+}$, and completely positive if $\mathrm{id}_{R} \otimes \mathcal{N}_{A \rightarrow B}$ (where $\rm id_{R}$ denotes the identity map acting on a reference system $R$) is a positive map for arbitrary size of the reference system $R$. A map $\mathcal{N}_{A \rightarrow B}$ is called trace-preserving  if for all $\mathsf{X}\in{\mathcal{L}}({A})$, $\tr\left(\mathcal{N}_{A \rightarrow B}(\mathsf{X})\right)=\tr\left(\mathsf{X}\right)$, and a positive map $\mathcal{N}_{A \rightarrow B}$ is called trace nonincreasing if for all $\mathsf{X}\in{\mathcal{L}}(A)_{+}$,  $\tr\left(\mathcal{N}_{A \rightarrow B}(\mathsf{X})\right)\leq\tr(\mathsf{X})$. A quantum channel $\mathcal{N}:\mathcal{L}\left({A}\right) \rightarrow \mathcal{L}\left({B}\right)$ is a linear, completely positive and trace-preserving (CPTP) map which describes a physically meaningful evolution of a quantum system. The set of all quantum channels $\mathcal{N}_{A \rightarrow B}$  is denoted by $\mathcal{L}(A, B)_{+}$. A linear map $\mathcal{N}_{A \rightarrow B}$ is called subunital if $\mathcal{N}_{A \rightarrow B} \left( \mathbbm{1}_{{A}}\right) \leq \mathbbm{1}_{{B}}$, unital if $\mathcal{N}_{A \rightarrow B} \left( \mathbbm{1}_{{A}}\right) = \mathbbm{1}_{{B}}$, and superunital if $\mathcal{N}_{A \rightarrow B} \left( \mathbbm{1}_{{A}}\right) > \mathbbm{1}_{{B}}$. The adjoint $\mathcal{N}^{*}_{B \rightarrow A}$ of a linear map $\mathcal{N}_{A \rightarrow B}$ is the unique linear map that satisfies the following equation:
 \begin{equation}
     \langle \mathsf{P}, \mathcal{N}_{A \rightarrow B}\left(\mathsf{Q}\right) \rangle = \langle \mathcal{N}_{B \rightarrow A}^{*}\left(\mathsf{P}\right), \mathsf{Q} \rangle, \label{map_adjoint}
 \end{equation}
 for all $\mathsf{Q} \in \mathcal{L}\left({A}\right)$ and $\mathsf{P} \in \mathcal{L}\left({B}\right)$. The adjoint map of a quantum channel is a completely positive and unital map. For a trace nonincreasing map, the adjoint map is subunital while for a trace nondecreasing map, the adjoint map is superunital. The trace norm, Hilbert-Schmidt norm, and the operator norm for an operator $\mathsf{X}\in\mathcal{L}(A)$ are defined as $\norm{\mathsf{X}}_{1}:= \mathrm{tr}\left(\sqrt{\mathsf{X}^{\dagger} \mathsf{X}}\right)$, $\norm{\mathsf{X}}_{2}:= \sqrt{\mathrm{tr}\left(\mathsf{X}^{\dagger} \mathsf{X}\right)}$, and $\norm{\mathsf{X}}_{\infty}:= \sup_{x \in A, x\neq 0} \frac{\norm{\mathsf{X}(x)}}{\norm{x}}$, where $\norm{\cdot}$ denotes the norm on $A$ induced by the inner product on $A$.

\begin{dfn}[Tele-covariant channels~\cite{holevo2002}](see also \cite[Sec.~V.C.]{DBWH21})\label{dfn:telecov}
    Let $G$ be a finite group and let $ \{U_{g}\}_{g}$ and $\{V_{g}\}_{g}$ be unitary representations of $G$ in the Hilbert spaces $A$ and $B$, respectively. A quantum channel $\mathcal{N}_{A \rightarrow B}$ is called covariant with respect to $G$ if for all $\rho\in \mathcal{D}(A)$,
    \begin{equation}
    \mathcal{N}_{A \rightarrow B} \circ \mathcal{U}_{g} = \mathcal{V}_{g } \circ \mathcal{N}_{A \rightarrow B} \hspace{0.2cm} \forall g \in G,
\end{equation}
where $\mathcal{U}_g (\cdot):= U_g(\cdot)U_g^\ast$ and $\mathcal{V}_g (\cdot):= V_g(\cdot)V_g^\ast$.
The quantum channel $\mathcal{N}_{A \rightarrow B}$ is called tele-covariant if it is covariant with respect to $G$ with unitary representation that satisfies $|G|^{-1}\sum_{g} U_{g} \left(\mathsf{X}\right) U^{*}_{g} = \frac{\tr(\mathsf{X})}{|A|}\mathbbm{1}_A$ for all $\mathsf{X}\in\mathcal{L}(A)$.
\end{dfn}

\noindent
{\it von Neumann entropy}.---The von Neumann entropy (or the entropy) of a quantum state $\rho_A \in {\mathcal{D}}(A)$ is a non-negative quantity defined as
 \begin{equation}
    S(\rho_A)= -\tr({\rho_{A}\log{\rho_{A}}}),
 \end{equation}
 where the logarithm is taken with respect to base two.

\medskip
 \noindent
{\it Generalized divergences}.--- A map $\mathbf{D}:\mathcal{L}\left(A\right)_{+} \times \mathcal{L}\left(A\right)_{+} \rightarrow \mathbb{R}$ is called generalized state divergence if it satisfies data processing inequality, i.e.,
\begin{equation}
    \BstD{\rho}{\sigma} \geq \BstD{\mathcal{N} \left(\rho\right) }{ \mathcal{N} \left(\sigma \right)},
\end{equation}
where $\rho, \sigma \in \mathcal{L}\left(A\right)_{+}$ and $\mathcal{N}:\mathcal{L}(A)_{+}\rightarrow \mathcal{L}(A)_{+}$ is any quantum channel. The quantum relative entropy is one example of generalized divergence. Given a density operator $\rho \in {\mathcal{D}}(A)$ and a positive semidefinite operator $\sigma\in{\mathcal{L}}(A)_{+}$, the relative entropy is defined as~\cite{Umegaki1962}:
 \begin{equation}
    \staD{\rho}{\sigma} := 
       \begin{cases}
                \tr(\rho(\log{\rho}-\log{\sigma})) &  \text{if}\  \operatorname{supp}(\rho)\subseteq \operatorname{supp}(\sigma), \\ 
                +\infty & \text{otherwise}.
       \end{cases}\label{equ:relative entropy}
 \end{equation}

\noindent

\noindent    
{\it The CJKS isomorphism}.--- Let $A$ and $B$ be two finite-dimensional Hilbert spaces, and let $\mathcal{N}_{A \rightarrow B}$ be a linear map. Let $ \{e_{ij}\}_{ij}$ be a complete set of matrix units for $\mathcal{L}\left({A}\right)$ and $ \{f_{ij}\}_{ij}$ be a complete set of matrix units for $\mathcal{L}\left({B}\right)$. Then the CJKS  matrix~\cite{book,article,CHOI1975285,kraus1983states,Sudarshan1985}, which we shall refer to as the Choi operator, for $\mathcal{N}_{A \rightarrow B}$ is defined as the operator $ \Cwhat_{\mathcal{N}} =  \sum_{{i,j}=1}^{|A|} e_{ij} \otimes \mathcal{N}_{A \rightarrow B}\left(e_{ij}\right)$,
which belongs to $ \mathcal{L}\left({A}\otimes {B}\right)$. Let us further define $\mathsf{C}_{\mathcal{N}}=\Cwhat_{\mathcal{N}}/|A|$. If $\mathcal{N}_{A\rightarrow B}$ is a quantum channel, then $\mathsf{C}_{\mathcal{N}}$ is a quantum state, and we will call it the Choi state of the channel $\mathcal{N}_{A\rightarrow B}$. Consider a  map $J_{A,B} : \mathcal{L}(A,B)\rightarrow \mathcal{L}(A\otimes B)$ defined by $J_{A,B}\left(\mathcal{N}_{A \rightarrow B}\right)= \Cwhat_{\mathcal{N}}$.
One can check that this map is linear and injective. To show that it is also surjective, let us consider an element  $  \mathsf{X} \in \mathcal{L}(A\otimes B) $ which can be written as $\mathsf{X} =\sum_{i,j,k,l} x_{ijkl}  e_{ij}\otimes f_{kl} 
     = \sum_{i,j} e_{ij}\otimes \sum_{k,l} x_{ijkl}f_{kl}.$ Now, $\Gamma_{\mathsf{X}}(e_{ij}) = \sum_{k,l} x_{ijkl}f_{kl}$, defines a map $\Gamma_{\mathsf{X}} :\mathcal{L}({A})\rightarrow \mathcal{L}(B)$ by linear extension. So, we can write
 \begin{eqnarray}
     \mathsf{X}=\sum_{{i,j}=1}^{|A|} e_{ij} \otimes \Gamma_{\mathsf{X}}(e_{ij}), \label{equ:CJ_isomorphism}
 \end{eqnarray}
  which implies that the map $J_{A,B}$ is surjective as well. This map is known as Choi-Jamio{\l}kowski isomorphism. The action of $\Gamma_{\mathsf{X}}$ on an element $\mathsf{Q}\in \mathcal{L}(A)$ can be written as, $ \Gamma_{\mathsf{X}}\left(\mathsf{Q}\right)=\mathrm{tr}_{A}\left[\left(\mathsf{Q}^t \otimes \mathbbm{1}_B\right)\mathsf{X}\right]$,
  where $t$ represents transposition with respect to the usual canonical basis of $A$ and ${\rm tr}_{A}$ is partial trace with respect to the system $A$. The Choi-Jamio{\l}kowski isomorphism leads to the concept of a channel-state duality. To understand the channel-state duality we need to have a look at the CJKS theorem on completely positive maps.
  
\medskip
\noindent
{\it CJKS theorem on completely positive maps}.--- The 
 Choi operator $\Cwhat_{\mathcal{N}} = \sum_{{i,j}=1}^{|A|} $ $e_{ij} \otimes \mathcal{N}_{A \rightarrow B}(e_{ij})$ $\in $ $\mathcal{L}(A\otimes B) $ is positive semidefinite if and only if the linear map $\mathcal{N}_{A \rightarrow B}$ is completely positive.

\medskip
\noindent
{\it Quantum superchannels}.--- A linear map $\Theta:\mathcal{L}\left(A,B\right) \rightarrow \mathcal{L}\left(C,D\right)$ is called a supermap which maps elements of $\mathcal{L}\left(A,B\right)$ to elements of $\mathcal{L}\left(C,D\right)$. A supermap $\Theta:\mathcal{L}\left(A,B\right) \rightarrow \mathcal{L}\left(C,D\right)$ is called completely CP-preserving if the map $ \mathrm{id} \otimes \Theta : \mathcal{L}\left(A',B'\right) \otimes \mathcal{L}\left(A,B\right)\rightarrow \mathcal{L}\left(A',B'\right) \otimes \mathcal{L}\left(C,D\right)$ is CP-preserving for for all choices of the Hilbert spaces $A'$ and $B'$, where $\mathrm{id}$ denotes the identity supermap. A supermap is called TP-preserving if it maps a trace-preserving map into a trace-preserving map. A quantum superchannel is defined as a completely CP-preserving and TP-preserving supermap. It is clear from the definition that a quantum superchannel maps a quantum channel to a quantum channel.

Let $\mathcal{N}_{A \rightarrow B}$ and $\mathcal{M}_{C \rightarrow D}$ be any two quantum channels related via a superchannel $\Theta$ as $ \Theta\left(\mathcal{N}_{A \rightarrow B}\right) = \mathcal{M}_{C \rightarrow D}$.
The action of $\Theta$ on $\mathcal{N}_{A \rightarrow B}$ can be physically realized as follows~\cite{Chiribella,Gour_2019}:
\begin{equation}
    \Theta\left(\mathcal{N}_{A \rightarrow B}\right)  = \mathcal{V}_{BR \rightarrow D}\circ\left(\mathcal{N}_{A\rightarrow B} \otimes \mathrm{id}_{R}\right) \circ \mathcal{U}_{C \rightarrow AR}, \label{equ:superchannel_action}
\end{equation}
where $\mathcal{V} \in \mathcal{L}\left({BR},D\right)$ is a postprocessing channel, $\mathcal{U} \in \mathcal{L}\left({C},AR\right)$ is a preprocessing channel, and $R$ is some reference system.

\medskip
\noindent
{\it Space of quantum channels and superchannels}.--- Here we discuss various properties of the space of quantum channels and superchannels. We start by introducing a notion of inner product on $\mathcal{L}(A,B)$ as the follows.
\begin{lemma}[\cite{Gour_2019}]\label{equ:lemma1}
  Let $\{b^{A}_{j}\}_{j}$  be an orthonormal basis for $\mathcal{L}\left({A}\right)$. Then, the space $\mathcal{L}(A,B)$ of linear maps is equipped with the following inner product:
\begin{eqnarray}
    \inner{\mathcal{N}_{A \rightarrow B},\mathcal{M}_{A \rightarrow B}}=\sum_{j} \inner{\mathcal{N}(b^{A}_{j}),\mathcal{M}(b^{A}_{j})}, \forall \hspace{0.2cm} \mathcal{N}_{A \rightarrow B}, \mathcal{M}_{A \rightarrow B} \in \mathcal{L}(A,B), \label{equ:inner_product_for_channels}
\end{eqnarray}
where the inner product on the right-hand side is the usual Hilbert-Schmidt inner product. This inner product is independent of the choice of orthonormal basis.
\end{lemma} 
The above inner product was defined in \cite{Gour_2019}, and it was asserted that the inner product is independent of choice of orthonormal basis. In Appendix \ref{detailed_proof_lemma_1}, we provide a simple proof of this statement. For a given pair of bases $\{b^{A}_{i}\}_{i}$ and $\{b^{B}_{j}\}_{j}$, for $\mathcal{L}(A)$ and $\mathcal{L}(B)$, respectively, we can define an orthonormal basis $\{\mathcal{E}^{A \rightarrow B}_{ij}\}$ for the space $\mathcal{L}(A,B)$ as follows
\begin{align}
    \mathcal{E}^{A \rightarrow B}_{ij}\left(\mathsf{K}\right)=\inner{b^{A}_{i},\mathsf{K}}b^{B}_{j}= \mathrm{tr}\left(({b^{A}_{i})}^{\dagger}\mathsf{K}\right)b^{B}_{j},
\end{align}
where $\mathsf{K} \in \mathcal{L}\left(A\right)$. It can be verified that the above-defined basis is orthonormal, i.e, $\inner{\mathcal{E}^{A \rightarrow B}_{ij},\mathcal{E}^{A \rightarrow B}_{kl}}=\delta_{ik}\delta_{jl}$. 
The adjoint $\Theta^{*} :\mathcal{L}\left({C},D\right) \rightarrow \mathcal{L}\left({A},B\right)$  of a linear map $\Theta :\mathcal{L}\left({A},B\right) \rightarrow \mathcal{L}\left({C},D\right)$ is a linear map that for all $\mathcal{N}_{A \rightarrow B}$ and $\mathcal{N'}_{C \rightarrow D}$ satisfies
\begin{align}
 \inner{\Theta (\mathcal{N}), \mathcal{N'}}=  \inner{\mathcal{N}, \Theta^{*}(\mathcal{N'})},
\end{align}
where the inner products on the left-hand side and right-hand side are in $\mathcal{L}\left({C},D\right)$ and $\mathcal{L}\left({A},B\right)$, respectively.

 In~\cite{Gour_2019}, it was shown that the Choi-type theorem~(CJKS isomorphism) also holds for supermaps. Here we provide an alternative proof that is suitable for our context.

\medskip
\noindent
{\it Choi type theorem for supermaps}.---
Let $\{e^{Q}_{ij}\}_{ij}$ be a canonical basis for the space $\mathcal{L}\left({Q}\right)$, where $ Q \in \{A,B,C,D\}$. Using the bra-ket notation, we can write $ e^{Q}_{ij} = \ketbra{i}{j}_Q$, where $\{\ket{i}_Q\}_{i}$ is an orthonormal basis for $Q$. An orthonormal basis for $\mathcal{L}\left({A},B\right)$ is the following,
\begin{eqnarray}
    \mathcal{E}^{A \rightarrow B}_{ijkl}\left(\mathsf{K}\right)=\inner{e^{A}_{ij},\mathsf{K}}e^{B}_{kl}= \mathrm{tr}\left({e_{ij}^{A}}^{\dagger} \mathsf{K}\right)e^{B}_{kl}.\label{equ:channel_basis}
\end{eqnarray}
Let $\mathscr{L}\left[\mathcal{L}\left({A},B\right),\mathcal{L}\left({C},D\right)\right]$ denote the set of all the linear maps which map elements of $\mathcal{L}\left({A},B\right)$ to $\mathcal{L}\left(C,D\right)$.
  The spaces $\mathscr{L}\left[\mathcal{L}\left({A},B\right),\mathcal{L}\left({C},D\right)\right]$ and $\mathcal{L}\left({A},B\right) \otimes \mathcal{L}\left({C},D\right)$ are isomorphic as vector spaces under the following identification
\begin{eqnarray}
    \Theta \leftrightarrow \Lambda_{\Theta} := \sum_{ijkl}  \mathcal{E}^{A \rightarrow B}_{ijkl} \otimes \Theta\left( \mathcal{E}^{A \rightarrow B}_{ijkl}\right) , \label{Choi-type-operator}
\end{eqnarray}
where $\Theta \in \mathscr{L}\left[\mathcal{L}\left({A},B\right),\mathcal{L}\left({C},D\right)\right]$ and $\Lambda_{\Theta} \in \mathcal{L}\left({A},B\right) \otimes \mathcal{L}\left({C},D\right)$.
It should be noted that the object $\Lambda_{\Theta}$ has a mathematical structure similar to Choi operators. The following lemma supports this similarity further.

\begin{lemma}[\cite{Gour_2019}]\label{theorem_1} 
   A supermap $\Theta \in \mathscr{L}\left[\mathcal{L}\left({A},B\right),\mathcal{L}\left({C},D\right)\right]$ is completely CP-preserving if and only if $\Lambda_{\Theta}=\sum_{ijkl}  \mathcal{E}^{A \rightarrow B}_{ijkl} \otimes \Theta \left( \mathcal{E}^{A \rightarrow B}_{ijkl}\right) \in \mathcal{L}\left({A},B\right) \otimes \mathcal{L}\left({C},D\right) $ is CP. 
\end{lemma} 
The above lemma in its present form is not mentioned in~\cite{Gour_2019}, but we can arrive in the present form by using \cite[Theorem 1]{Gour_2019}. See Appendix~\ref{detailed_proof_thm_1} for an alternate, detailed proof of the above lemma. The CP vs completely CP-preserving correspondence in the above lemma depends on the choice of basis. The necessary and the sufficient condition on the choice of basis such that the above lemma holds can be found in~\cite{sohail_2023_June}.

\medskip
\noindent
{\it Representing Map}.--- Let $\Theta\in \mathscr{L}\left[\mathcal{L}\left({A},B\right),\mathcal{L}\left({C},D\right)\right]$ be a supermap.
If this map is completely CP-preserving and TP-preserving then it is called a quantum superchannel. The supermap $\Theta\in \mathscr{L}\left[\mathcal{L}\left({A},B\right),\mathcal{L}\left({C},D\right)\right]$ induces a map $\mathsf{T}: \mathcal{L}\left({A \otimes B}\right) \rightarrow \mathcal{L}\left({C \otimes D}\right) $, at the level of Choi operators, and is given by
\begin{equation}
            \mathsf{T}(\mathsf{X})= \Cwhat_{\Theta\left(\Gamma_{\mathsf{X}}\right)}, \label{equ:rep_map2}
\end{equation}
where $\mathsf{X} \in \mathcal{L}\left({A \otimes B}\right)$ is the Choi operator of  $\Gamma_{\mathsf{X}}$ and is defined in Eq.~\eqref{equ:CJ_isomorphism}, and $\Cwhat_{\Theta\left(\Gamma_{\mathsf{X}}\right)}$ is Choi operator for map $\Theta\left(\Gamma_{\mathsf{X}}\right)$. Following~\cite{Chiribella}, we will call this map the representing map of the supermap $ \Theta$. The spaces $\mathscr{L}\left[\mathcal{L}\left({A},B\right),\mathcal{L}\left({A},B\right)\right]$ and $\mathcal{L}\left(A \otimes B,A\otimes B\right)$ are algebras with respect to composition and are isomorphic under the identification: $\Theta \leftrightarrow \mathsf{T}$~\cite{Chiribella, Gour_2019, sohail_2023_June}. The following lemma provides the necessary and sufficient conditions for the representing map $\mathsf{T}$ to be completely positive.

\begin{lemma}[\cite{Gour_2019, Chiribella}]\label{lemma_3}
    A supermap $\Theta$ is completely CP-preserving if and only if its representing map $\mathsf{T}$ is CP.
\end{lemma}

 We have already shown that a supermap $\Theta \in \mathscr{L}\left[\mathcal{L}\left({A},B\right),\mathcal{L}\left({C},D\right)\right]$ is completely CP-preserving if and only if $\Lambda_{\Theta}=\sum_{ijkl}  \mathcal{E}^{A \rightarrow B}_{ijkl} \otimes \Theta \left( \mathcal{E}^{A \rightarrow B}_{ijkl}\right) \in \mathcal{L}\left(A,B\right)\otimes \mathcal{L}\left(C, D\right) $ is CP, where $\Lambda_{\Theta}$ can be viewed as a map from $\mathcal{L}(A \otimes C)$ to $\mathcal{L}(B\otimes D)$, and $\{\mathcal{E}^{{A} \rightarrow B}_{\delta s \beta q}\}_{\delta s \beta q}$ is the orthonormal basis for $\mathcal{L}(A ,B)$ defined in Eq.~(\ref{equ:channel_basis}). Hence, we can determine its complete positivity by checking the positivity of its Choi operator $ \Cwhat_{\Lambda_{\Theta}} $. For a detailed proof, see Appendix~\ref{sec:detailed_proof_of_lemma3}. 

 Note that the representing map $\mathsf{T}:\mathcal{L}\left(A \otimes {B}\right) \rightarrow \mathcal{L}\left({C}\otimes {D}\right)$ of a TP-preserving supermap preserves the trace of those $\mathsf{X} \in \mathcal{L}\left(A \otimes {B}\right)$ for which $\Gamma_{\mathsf{X}}$ is trace-preserving provided $|A|=|C|$, but in general $\mathsf{T}$ could neither be a trace nonincreasing map nor a trace increasing map. However, we can transform the map $\mathsf{T}$ into a map $\mathsf{T}{'}: \mathcal{L}\left(A \otimes {B}\right) \rightarrow \mathcal{L}\left(C \otimes {D}\right)$ in such a way that the map $\mathsf{T}{'}$ is trace-preserving, and whenever $|A|=|C|$ and $\Gamma_{\mathsf{X}}$ is trace-preserving, $\mathsf{T}(\mathsf{X})=\mathsf{T}{'}(\mathsf{X})$ holds. The explicit form of the map $\mathsf{T}{'}$ is given by the following expression:
\begin{equation}
    \mathsf{T}{'}(\mathsf{X})=\mathsf{T}(\mathsf{X})+\left[\mathrm{tr}(\mathsf{X})-\mathrm{tr}(\mathsf{T}(\mathsf{X}))\right] \sigma_{0}, \label{equ:trace_preserving_T}
\end{equation}
where $\sigma_{0}\in \mathcal{L}(C\otimes D)$ with $\mathrm{tr}(\sigma_{0})=1$. In particular, $\mathsf{T}{'}\left(\Cwhat_{\mathcal{N}}\right)=\mathsf{T}\left(\Cwhat_{\mathcal{N}}\right)$ holds for all quantum channels $\mathcal{N}$, if $\mathsf{T}$ does not change the space of operators, or $|A| = |C|$. Choosing $\sigma_0$ to be a fixed quantum state, we observe that $\mathsf{T}'$ is CPTP if $\mathsf{T}$ is CP and trace nonincreasing map. For the case when $\mathsf{T}$ is  neither trace nonincreasing map nor trace increasing map, we have derived necessary and sufficient conditions for $\mathsf{T}'$  to be a CP map in terms of Choi operator of $\mathsf{T}$ in Appendix \ref{quantum_channel_from_TP_map}.

\begin{dfn}\label{dfn_sct}
     We define a set of superchannel  $\rm{SCT}\subset\mathscr{L}\left[\mathcal{L}\left({A},B\right),\mathcal{L}\left({C},D\right)\right] $ as 
     \begin{align}
         \mathrm{SCT}:= \{\Theta:\mathcal{L}\left(A,B\right) \rightarrow \mathcal{L}\left(C,D\right) \text{such that that $\mathsf{T}'$ is CPTP (see Appendix \ref{quantum_channel_from_TP_map})}\}.
     \end{align}    
\end{dfn}
The map $\mathsf{T}'$ is defined via Eq.~\eqref{equ:trace_preserving_T}. In section \ref{Sufficiency_of_superchannels}, we will use this class of superchannels to establish several results for the channel recovery problem.

A given representing map $\mathsf{T}:\mathcal{L}\left(A \otimes {B}\right) \rightarrow \mathcal{L}\left(C \otimes {D}\right)$ induces a supermap $\Lambda : \mathcal{L}(A,B)\rightarrow \mathcal{L}(C,D) $ as follows, 
\begin{equation}
    \Lambda(\mathcal{N})(\mathsf{X})=\mathrm{tr}_{C}\left[(\mathsf{X}^{t} \otimes \mathbbm{1}) \mathsf{T}\left(\Cwhat_{\mathcal{N}}\right)\right],
\end{equation}
 where $\mathsf{X}^t$ denote the transpose of operator $\mathsf{X}$, $\mathcal{N} \in \mathcal{L}(A ,B)$, and $\Cwhat_{\mathcal{N}}$ is the Choi operator for $\mathcal{N}$. Before proceeding further, we first find the adjoints of the maps $\mathsf{T}$ and $\mathsf{T}{'}$.
\begin{lemma}[\cite{Gour_2019}]\label{lemma_4}
    The adjoint $\mathsf{T}^{*}:\mathcal{L}(C\otimes D) \rightarrow \mathcal{L}(A\otimes B)$ of the representing map $\mathsf{T}$ of a supermap $\Theta:\mathcal{L}(A, B)\rightarrow \mathcal{L}(C, D)$ is given by
\begin{eqnarray}
    \mathsf{T}^{*}\left(\mathsf{Y}\right)= \Cwhat_{\Theta^{*}(\Gamma_{\mathsf{Y}})},
\end{eqnarray}
where $\Theta^{*}$ is the adjoint of $\Theta$, $\mathsf{Y} \in \mathcal{L}\left({C \otimes D}\right)$ is the Choi operator of  $\Gamma_{\mathsf{Y}}$, and $\Cwhat_{\Theta^{*}(\Gamma_{\mathsf{Y}})}$ is the Choi operator for ${\Theta^{*}(\Gamma_{\mathsf{Y}})}$.
\end{lemma}
This result, which tells that the identification $\mathsf{T} \leftrightarrow \Theta$ is a ${}^\ast$-isomorphism, was proved in~\cite{Gour_2019}. We provide an alternate proof of this result in Appendix~\ref{sec:detailed_proof_lemma_4}.
\begin{remark}
    With the help of the representing map, we notice that a supermap $\Theta$ is completely CP-preserving if and only if its adjoint $\Theta^{*}$ is completely CP-preserving.
\end{remark}

\begin{lemma}
The adjoint $\mathsf{T}^{'*}:\mathcal{L}\left({C} \otimes {D}\right) \rightarrow \mathcal{L}\left({A} \otimes {B}\right)$ of the map $\mathsf{T}{'}$, which is defined via Eq.~\eqref{equ:trace_preserving_T}, is given by 
\begin{equation}
    \mathsf{T}^{'*}(\mathsf{Y})=\mathsf{T}^{*}(\mathsf{Y})+ \mathrm{tr}\left(\mathsf{Y}^{\dagger} \sigma_{0}\right)\left(\mathbbm{1}_{AB}-\mathsf{T}^{*}(\mathbbm{1}_{CD})\right),
\end{equation} 
where  $\sigma_{0}\in \mathcal{D}(C\otimes D)$ is a fixed state.
\end{lemma}
\begin{proof}
Let us define a map $\mathsf{T}_{0}:\mathcal{L}\left({A} \otimes {B}\right) \rightarrow \mathcal{L}\left({C} \otimes {D}\right)$ such that its action on any $\mathsf{X} \in \mathcal{L}(A \otimes B) $ is given by $\mathsf{T}_{0}(\mathsf{X})=\left[\mathrm{tr}(\mathsf{X})-\mathrm{tr}(\mathsf{T}(\mathsf{X}))\right]\sigma_{0}$. Now, from the definition of the adjoint, for all $Y\in \mathcal{L}\left({C} \otimes {D}\right)$, we have
\begin{align}
\nonumber
    \inner{\mathsf{T}^{*}_{0}(\mathsf{Y}),\mathsf{X}} = \inner{\mathsf{Y},\mathsf{T}_{0}(\mathsf{X})}
   & = \inner{\mathsf{Y},\left[\mathrm{tr}(\mathsf{X})-\mathrm{tr}(\mathsf{T}(\mathsf{X}))\right]\sigma_{0}}= \inner{\mathsf{Y},\sigma_{0}} \mathrm{tr}(\mathsf{X}) -\inner{\mathsf{Y},\sigma_{0}} \mathrm{tr}\left(\mathsf{T}(\mathsf{X})\right)\nonumber\\
    &= \inner{\mathsf{Y},\sigma_{0}} \inner{\mathbbm{1}_{AB},\mathsf{X}}  -\inner{\mathsf{Y},\sigma_{0}} \inner{\mathsf{T}^{*}(\mathbbm{1}_{CD}),\mathsf{X}}\nonumber\\
    &= \inner{\mathrm{tr}\left(\mathsf{Y}^{\dagger}\sigma_{0}\right) \left[\mathbbm{1}_{AB}-\mathsf{T}^{*}(\mathbbm{1}_{CD})\right],\mathsf{X}}.
\end{align}
 So, we have $\mathsf{T}^{*}_{0}\left(\mathsf{Y}\right)=\mathrm{tr}\left(\mathsf{Y}^{\dagger}\sigma_{0}\right) [\mathbbm{1}_{AB}-\mathsf{T}^{*}(\mathbbm{1}_{CD})]$. As the adjoint is linear, we have the desired result.
\end{proof}

We now introduce the completely depolarizing map, which will play an important role in expressing our main results.
 \begin{dfn} \label{def_completely_depolarising_map}
The action of the completely depolarizing map $\mathcal{R}_{A \rightarrow B}$ is defined as
    \begin{equation}\label{equ:completely_depolarizing_map}
        \mathcal{R}_{A \rightarrow B}(\mathsf{X})= \mathrm{tr}(\mathsf{X})\mathbbm{1}_{B} \hspace{0.2cm} \forall~~ \mathsf{X}\in\mathcal{L}(A).
    \end{equation}
\end{dfn}

\begin{lemma}\label{lemma5} The representing map $\mathsf{T}: \mathcal{L}\left({A} \otimes {B}\right) \rightarrow \mathcal{L}\left({C} \otimes {D}\right)$ of a supermap $\Theta$ is trace-preserving if and only if $\Theta^{*}:\mathcal{L}(C,D) \rightarrow \mathcal{L}(A ,B)$ is completely depolarizing preserving ($\mathcal{R}$-preserving), i.e., $\Theta^{*}(\mathcal{R}_{C \rightarrow D})=\mathcal{R}_{A \rightarrow B}$. 
 \end{lemma} 
\begin{proof}
   Let $\mathcal{R}_{A \rightarrow B}$ and $\mathcal{R}_{C \rightarrow D}$ be the completely depolarizing maps of $\mathcal{L}\left(A,B\right)$ and $\mathcal{L}\left({C,D}\right)$, respectively. Then we have the following relations. 
   \begin{align*}
       \mathsf{T}~ \text{is TP} \Leftrightarrow \mathsf{T}^{*}(\mathbbm{1}_{CD})=\mathbbm{1}_{AB}\Leftrightarrow \mathsf{T}^{*}\left(\Cwhat_{\mathcal{R}_{C \rightarrow D}}\right)= \Cwhat_{{\mathcal{R}_{A \rightarrow B}}}\Leftrightarrow \Cwhat_{\Theta^{*}(\mathcal{R}_{C \rightarrow D})} = \Cwhat_{\mathcal{R}_{A \rightarrow B}}\Leftrightarrow \Theta^{*}(\mathcal{R}_{C \rightarrow D})=\mathcal{R}_{A \rightarrow B}.
   \end{align*}
   This completes the proof of the lemma. 
\end{proof}  

 \begin{remark}
     A similar proof technique can be applied to prove that the adjoint $\mathsf{T}^{*}$ of the representing map $T$ of a supermap $\Theta$ is trace-preserving if and only if $\Theta$ is $\mathcal{R}$-preserving.
 \end{remark}

\medskip
\noindent
{\it Representing map with respect to rank one operators with full rank reduced states}.---  Consider a rank one operator $\Psi_{RA} \in \mathsf{FRank}(R\otimes A)$ and an orthonormal basis $\{\ket{j}\}_j$ of $R\simeq A$. Such an operator can be written as $\Psi_{RA} = \ketbra{\Psi}{\Psi}_{RA}$ with $ \ket{\Psi}_{RA}=(\mathsf{A}_{\Psi} \otimes \mathbbm{1}_A)\sum_{i=1}^{|A|}\ket{i,i}_{RA}$, 
where $\mathsf{A}_{\Psi}\in\mathcal{L}(R)$ is defined via $\bra{i}\mathsf{A}_{\Psi} \ket{j}:=\langle i,j\ket{\Psi}_{RA}$. As $\mathrm{rank}(\mathsf{A}_{\Psi}) = \mathrm{rank}(\mathsf{A}_{\Psi} \mathsf{A}_{\Psi}^{*})$, we have the reduced state $\Psi_R=\mathsf{A}_{\Psi} \mathsf{A}_{\Psi}^{*}$ of $\Psi_{RA}$ invertible if and only if $\mathsf{A}_{\Psi}$ is invertible. For such an operator, we define the Choi operator for a linear map $\mathcal{N}: \mathcal{L}(A) \rightarrow \mathcal{L}(B)$ as the following:
\begin{eqnarray} \Cwhat_{\mathcal{N}}^{\Psi}:=\left(\mathrm{id}_{R\rightarrow R} \otimes \mathcal{N}_{A\rightarrow B}\right) \left(\Psi_{RA}\right).
\end{eqnarray}
Now the Choi operator $\Cwhat_{\mathcal{N}}^{\Psi} \in \mathcal{L}(A\otimes B)$ and the Choi operator $\Cwhat_{\mathcal{N}}\in \mathcal{L}(A\otimes B)$ defined with respect to a maximally entangled operator $\Psi^+_{AB}$ are related as follows. 

\begin{align}
\label{eq:choi-wit-without-max-ent}
\Cwhat_{\mathcal{N}}^{\Psi}&:=\left(\mathrm{id}_{R\rightarrow R} \otimes \mathcal{N}_{A\rightarrow B}\right) \left(\Psi_{RA}\right)\nonumber\\
&= (\mathsf{A}_{\Psi} \otimes \mathbbm{1}_B) \Cwhat_{\mathcal{N}} (A^{*}_{\Psi} \otimes \mathbbm{1}_B)\nonumber\\   
&= T_{\Psi}\left(\Cwhat_{\mathcal{N}}\right),
\end{align} 
where $T_{\Psi} : \mathcal{L}\left({A} \otimes {B}\right) \rightarrow \mathcal{L}\left({A}\otimes {B}\right)$ is given by $ T_{\Psi}(\mathsf{X}_{AB})=  (\mathsf{A}_{\Psi} \otimes \mathbbm{1}_B) \mathsf{X}_{AB} (\mathsf{A}^{*}_{\Psi} \otimes \mathbbm{1}_B),$
for all $\mathsf{X}_{AB} \in \mathcal{L}({A} \otimes {B})$. As $\mathsf{A}_{\Psi}$ is invertible, we have $T_{\Psi}$ invertible, and looking at the above form it is clear that $T_{\Psi}$ is completely positive, hence the Choi-Jamio{\l}kowski isomorphism holds if we use $\Psi_{RA}\in \mathsf{FRank}(R\otimes A)$ instead of maximally entangled operator $\Psi^+_{RA} \in \mathsf{FRank}(R\otimes A)$. 

\begin{remark}
\label{rem:choi-state-choi-operator-relation}
We notice that $\Psi_{RA}$ belongs to $\mathsf{FRank}(R\otimes A)$, and it satisfies $\tr\left(\Psi_{RA}\right)=\tr\left(\mathsf{A}^{*}_{\Psi} \mathsf{A}_{\Psi} \right)$ implying that the normalization of $\Psi_{RA}$ is determined by the operator $\mathsf{A}_{\Psi}$. Hence the normalized Choi operator $\Cstat^{\Psi}_{\mathcal{N}}$ obeys Eq.~\eqref{eq:choi-wit-without-max-ent}. Further, absorbing the normalization factor $|A|$ of $\Psi^+_{RA}$ in the operator $\mathsf{A}_{\Psi}$, we see that Eq.~\eqref{eq:choi-wit-without-max-ent} is satisfied by $\Cstat^{\Psi}_{\mathcal{N}}$ and $\Cstat_{\mathcal{N}}$, i.e., $\Cstat_{\mathcal{N}}^{\Psi} = T_{\Psi}\left(\Cstat_{\mathcal{N}}\right)$. If $\mathcal{N}$ is a quantum channel, then both $\Cstat_{\mathcal{N}}^{\Psi}$ and $\Cstat_{\mathcal{N}}$ become the Choi states.

\end{remark}

Note that $T_{\Psi}$ induces a supermap 
$\mathrm{\Theta}_{\Psi} :\mathcal{L}\left(A,B\right) \rightarrow \mathcal{L}\left(A,B\right)$ given by
\begin{eqnarray}
     \mathrm{\Theta}_{\Psi}(\mathcal{N})(\mathsf{X})=\tr_{A}\left((\mathsf{X}^{t} \otimes \mathbbm{1}_B) T_{\Psi}\left(\Cwhat_{\mathcal{N}}\right)\right), \label{equ: HRI}
\end{eqnarray}
for all $\mathcal{N} \in \mathcal{L}\left({A},{B}\right)$ and $\mathsf{X} \in \mathcal{L}({A})$, where $\mathsf{X}^{t}$ represents the transpose of $\mathsf{X}$ with respect to the usual canonical basis. Similarly, we can define $\Cwhat_{\mathcal{M}}^{\Phi}$, $T_{\Phi}$, and $\mathrm{\Theta}_{\Phi}$ for another rank one operator $\Phi_{RC}\in \mathsf{FRank}(R\otimes C)$, and a map $\mathcal{M}: \mathcal{L}(C) \rightarrow \mathcal{L}(D)$.

Now, we define a representing map $\mathfrak{T}: \mathcal{L}({A} \otimes {B}) \rightarrow \mathcal{L}({C} \otimes {D})$ for a supermap $\mathrm{\Theta}:\mathcal{L}\left(A,B\right) \rightarrow \mathcal{L}\left(C,D\right)$ with respect to the Choi operators $\Cwhat_{\mathcal{N}}^{\Psi}$ and $\Cwhat_{\mathcal{M}}^{\Phi}$ as follows:
\begin{eqnarray}
    \mathfrak{T}\left(\Cwhat_{\mathcal{N}}^{\Psi}\right)= \Cwhat_{\mathrm{\Theta}(\mathcal{N})}^{\Phi} \label{equ:new_rep_1}.
\end{eqnarray}
We establish the relation between the representing map $\mathsf{T}$ of the supermap $\mathrm{\Theta}$ defined with respect to the Choi operators $\Cwhat_{\mathcal{N}}$ and $\Cwhat_{\mathcal{M}}$, and the representing map $\mathfrak{T}$ of the supermap $\mathrm{\Theta}$ defined with respect to the Choi operators $\Cwhat_{\mathcal{N}}^{\Psi}$ and $\Cwhat_{\mathcal{M}}^{\Phi}$ as follows. We have $\mathsf{T}\left(\Cwhat_{\mathcal{N}}\right)= \Cwhat_{\mathrm{\Theta}(\mathcal{N})}\Leftrightarrow \mathsf{T} \circ T_{\Psi}^{-1} \left(\Cwhat_{\mathcal{N}}^{\Psi}\right)= T_{\Phi}^{-1} \left(\Cwhat_{\mathrm{\Theta}(\mathcal{N})}^{\Phi}\right)$, which implies $T_{\Phi} \circ \mathsf{T} \circ T_{\Psi}^{-1} \left(\Cwhat_{\mathcal{N}}^{\Psi}\right)=  \Cwhat_{\mathrm{\Theta}(\mathcal{N})}^{\Phi}$. By comparing this with
 Eq.~(\ref{equ:new_rep_1}), we have 
\begin{eqnarray}
    \mathfrak{T}= T_{\Phi} \circ \mathsf{T} \circ T_{\Psi}^{-1}. \label{new_rep_3}
\end{eqnarray}
Noticing that $T^{-1}_{\Psi}$ and $T_{\Phi}$ are both completely positive, we have that $\mathfrak{T}$ is completely positive if and only if $\mathsf{T}$ is completely positive. From Lemma \ref{lemma_3}, $\mathsf{T}$ is completely positive if and only if $\mathrm{\Theta}$ is completely CP-preserving. Thus, we conclude that $\mathfrak{T}$ is completely positive if and only if $\mathrm{\Theta}$ is completely CP-preserving.

\medskip
\noindent
{\it The adjoint map $\mathfrak{T}^{*}$}.---
We already know that $\mathsf{T}^{*}\left(\Cwhat_{\mathcal{M}}\right)=\Cwhat_{\mathrm{\Theta}^{*}(\mathcal{M})}$. As the map $\mathsf{T} \leftrightarrow \mathrm{\Theta}$ is product preserving, we have $ T^{-1}\left(\Cwhat_{\mathcal{M}}\right)= \Cwhat_{\mathrm{\Theta}^{-1}(\mathcal{M})}$.
From this relation we derive the action of $\mathfrak{T}^{*} : \mathcal{L}({C} \otimes {D}) \rightarrow \mathcal{L}({A} \otimes {B})$ on $\Cwhat_{\mathcal{M}}^{\Phi}$ as
\begin{eqnarray}
    \mathfrak{T}^{*} \left(\Cwhat_{\mathcal{M}}^{\Phi}\right)= (T_{\Psi}^{-1})^{*} \circ \mathsf{T}^{*} \circ (T_{\Phi})^{*} \circ T_{\Phi} \left(\Cwhat_{\mathcal{M}}\right)= \Cwhat^{\Psi}_{\mathrm{\Theta}_{\Psi}^{-1} \circ \mathrm{\Theta}^{-1*}_{\Psi} \circ \mathrm{\Theta}^{*} \circ \mathrm{\Theta}^{*}_{\Phi} \circ \mathrm{\Theta}_{{\Phi}(\mathcal{M})}} .\label{T_theta_phi}
\end{eqnarray}
By defining $\mathrm{\Theta}^{*}_{\Phi} \circ \mathrm{\Theta}_{\Phi}:= \Hat{\mathrm{\Theta}}_{\Phi}$ and $\mathrm{\Theta}^{*}_{\Psi} \circ \mathrm{\Theta}_{\Psi}:= \Hat{\mathrm{\Theta}}_{\Psi} $ we can rewrite the above equation as
\begin{eqnarray}
    \mathfrak{T}^{*} \left(\Cwhat_{\mathcal{M}}^{\Phi}\right)&&= \Cwhat^{\Psi}_{{\Hat{\mathrm{\Theta}}_{\Psi}}^{-1} \circ \mathrm{\Theta}^{*} \circ \Hat{\mathrm{\Theta}}_{{\Phi}(\mathcal{M})}}. \label{equ:theta_caps}
\end{eqnarray}
\begin{remark} \label{rem:choi-state-choi-operator-relation_1}
    We can also define the representing map $\mathfrak{T}$ of a supermap $\Theta$ with respect to the normalized Choi operators $\Cstat_{\mathcal{N}}^{\Psi}$ and $\Cstat_{\mathcal{N}}^{\Phi}$ as $\mathfrak{T}(\Cstat_{\mathcal{N}}^{\Psi})=\Cstat_{\Theta(\mathcal{N})}^{\Phi}$. Due to Remark~\ref{rem:choi-state-choi-operator-relation}, the representing map defined in this way satisfies Eq.~(\ref{new_rep_3}) and Eq.~(\ref{equ:theta_caps}). From $\mathfrak{T}$, a trace-preserving map $\mathfrak{T}'$ can be defined as in Eq.~(\ref{equ:trace_preserving_T}). For a trace-preserving map $\mathcal{N}$ and a TP-preserving supermap $\Theta$, we have $\tr\left(\Cstat_{\mathcal{N}}^{\Psi} \right)=\tr(\Psi_{RA})$ and $\tr \left( \mathfrak{T}(\Cstat_{\mathcal{N}}^{\Psi})\right)=\tr\left(\Cstat_{\Theta(\mathcal{N})}^{\Phi} \right)=\tr(\Phi_{RC})$. Due to the fact that $\Psi_{RA}$ and $\Phi_{RC}$ are normalized, i.e., $\tr(\Psi_{RA})=\tr(\Phi_{RC})=1$, we have $\mathfrak{T}'(\Cstat_{\mathcal{N}}^{\Psi})=\mathfrak{T}(\Cstat_{\mathcal{N}}^{\Psi})$. Notice that the validity of this relation does not require $|A|=|C|$ unlike the case where the representing map is defined with respect to the unnormalized $\Psi_{RA}$ and $\Phi_{RC}$. This observation is essential in deriving Theorems~\ref{theorem_3} and \ref{theorem_4}.
\end{remark}
\sloppy{
\noindent
{\it Representing map for super-supermaps}.--- Let us denote by $\spadesuit: \mathscr{L}\left[\mathcal{L}\left(A,B\right),\mathcal{L}\left(C,D\right)\right] \rightarrow \mathscr{L}\left[\mathcal{L}\left(A',B'\right),\mathcal{L}\left(C',D'\right)\right] $ a higher order process transforming supermaps to supermaps.} We call it a super-supermap. Now, we can define its representing supermap $\mathbb{T}_{\spadesuit}: \mathcal{L}\left(A,B\right)\otimes \mathcal{L}\left(C,D\right) \rightarrow \mathcal{L}\left(A',B'\right) \otimes \mathcal{L}\left(C',D'\right)$ by 
\begin{align}
    \mathbb{T}_{\spadesuit}(\Lambda_{\Theta})= \Lambda_{\spadesuit(\Theta)}.
\end{align}
where the supermap $\Theta \in \mathscr{L}\left[\mathcal{L}\left(A,B\right),\mathcal{L}\left(C,D\right)\right]$ and $\Lambda_{\Theta}= \sum_{ijkl}  \mathcal{E}^{A \rightarrow B}_{ijkl} \otimes \Theta\left( \mathcal{E}^{A \rightarrow B}_{ijkl}\right)$. 
For a given supermap $\mathbb{T}_{\spadesuit}$, one can find its corresponding super-supermap $\spadesuit$ as
\begin{align}
    \spadesuit[\Theta](\mathcal{N})=\tr_{A'}[(\overbar{\mathcal{N}}^{*} \otimes \id)  \mathbb{T}_{\spadesuit}(\Lambda_{\Theta})], \label{YYY_1}
\end{align}
where $\overbar{\mathcal{N}}$ represents the complex conjugate of the map $\mathcal{N}$ defined as $\overbar{\mathcal{N}}:= \sum_{ijkl} \bar{\alpha}_{ijkl}\mathcal{E}^{A' \rightarrow B'}_{ijkl}$ when $\mathcal{N}$ is written as $\mathcal{N}:= \sum_{ijkl}\alpha_{ijkl}\mathcal{E}^{A' \rightarrow B'}_{ijkl}$ and ${}^\ast$ represents the adjoint as defined in Eq.~\eqref{map_adjoint}. The trace of a map $\mathcal{N}: \mathcal{L}(A') \rightarrow \mathcal{L}(A')$ is defined by $\tr[\mathcal{N}]= \sum_{ij} \inner{e^{A'}_{ij}, \mathcal{N} \left(e^{A'}_{ij}\right)}$ with $e^{A'}_{ij}$ being the canonical basis of the space $A'$. Using the form of Eq.~(\ref{equ:superchannel_action}) for $\mathbb{T}_{\spadesuit}$ in Eq.~(\ref{YYY_1}), we expect a representation similar to Eq.~\eqref{equ:superchannel_action} to hold for a super-superchannel $\spadesuit$.

  \section{Entropy change of quantum channels}\label{Entropy_of_quantum_channels}    

 Distinguishing between two different quantum states is one of the most basic tasks in quantum information theory as a multitude of problems in quantum information theory can be recast as state distinguishability problems \cite{Hel76}. In the asymptotic limit of quantum state discrimination, the minimum error in distinguishing two quantum states is expressed neatly by the quantum Chernoff bound \cite{ACMT+07}, which is very closely related to the relative entropy  between the states \cite{Vedral_2002}. 
  This provides an operational meaning to the relative entropy between two quantum states.  Apart from these, the relative entropy plays a pivotal role in proving a plethora of results such as the strong subadditivity of the von Neumann entropy \cite{NC00} and the recoverability of quantum states affected by quantum noise \cite{OP93}. In fact, the quantum relative entropy can be treated as a parent quantity for the analysis of various information theoretic problems, as many other information theoretic quantities such as the von Neumann entropy, the quantum mutual information \cite{NC00}, and the free energy \cite{BHO+13} can be re-expressed in terms of relative entropy. Our interest here is in the von Neumann entropy. The von Neumann entropy is a measure of randomness or uncertainty present in a quantum system and has a variety of uses in quantum information theory, e.g., it is crucial in the theory of entanglement~\cite{Schumacher1996}, quantum communications~\cite{Holevo1998,Yuen1993}, and quantum cryptography~\cite{CBTW17}. It is a natural generalization of Shannon entropy~\cite{SHANNON1948}. We can write the von Neumann entropy of a quantum state $\rho$ in terms of the quantum relative entropy as follows
\begin{equation}
     S(A)_{\rho} = -\staD*{\rho_{A}}{\mathbbm{1}_A}.\label{equ:entropy}
\end{equation}
If the quantum system is finite-dimensional, the above formula can also be written as 
\begin{equation}
     S(A)_{\rho} =\log|A| -\staD*{\rho_{A}}{\mathbbm{\widetilde{1}}_A}, \label{equ:entropy 1}
\end{equation}
where $ \mathbbm{\widetilde{1}}_{A}:=\mathbbm{1}_A/|A|$ is the maximally mixed state of the system  $A$.

While quantum states describe a quantum system at any instant completely, its dynamics is described by the quantum channels. It is insightful to note that quantum states themselves can be thought of as quantum channels with one dimensional inputs, and quantum measurements can be thought of as quantum channels with one dimensional outputs \cite{Chiribella}. When viewed from an operational perspective, a quantum state is a preparation process while a measurement describes post-processing. In this sense, quantum channels can be thought of as a universal paradigm  describing both static and dynamic processes of quantum systems.

To this end, one can naturally ask: Is there a similar notion of entropy~(or relative entropy) that can be applied to quantum channels? One way to answer this question is to first find a meaningful notion of relative entropy between quantum channels, then, similar to quantum states, relative entropy between quantum channels can be used to define the entropy of quantum channels. This procedure was followed in~\cite{Gour_entropy} to define the notion of entropy for a quantum channel. The notion of relative entropy between two quantum channels also has an operational meaning analogous to quantum states, namely, it quantifies the difficulty in discriminating  two quantum channels~\cite{Gour_entropy}. It was advocated in~\cite{Gour_entropy} that the entropy functional for a quantum channel should satisfy some physically motivated axioms~(see Appendix~\ref{axioms}). In~\cite{Gour_entropy}, Eq.~(\ref{equ:entropy 1}) is considered as the starting point for defining the entropy of quantum channels and for a quantum channel $\mathcal{N}:\mathcal{L}\left({A}\right)\rightarrow\mathcal{L}\left(B\right)$, the authors have defined its entropy $S(\mathcal{N})$ as 
\begin{equation}
    S(\mathcal{N}):=\log|B| -\chD*{\mathcal{N}}{ \widetilde{\mathcal{R}}}, \label{equ:definition_entropy_gour}
\end{equation}
where $\widetilde{\mathcal{R}}:\mathcal{L}\left({A}\right)\rightarrow\mathcal{L}\left(B\right)$ is the completely depolarizing channel defined by $\widetilde{\mathcal{R}}(\mathsf{X})=\mathrm{tr}(\mathsf{X})\widetilde{\mathbbm{1}}_B$ and $\chD*{\mathcal{N}}{ \mathcal{M}} $ is the relative entropy between two quantum channels $\mathcal{N}$ and $\mathcal{M}$ defined by~\cite{Gour_entropy}
\begin{equation}
    \chD*{\mathcal{N}}{ \mathcal{M}} :=  \underset{\rho_{RA}\in\mathcal{D}(R\otimes A)}{\mathrm{sup}}~ \staD*{\mathrm{id}_{R}\otimes\mathcal{N} (\rho_{RA}) }{\mathrm{id}_{R}\otimes\mathcal{M}(\rho_{RA})}.
\end{equation}
The entropy of a quantum channel was postulated to be non-decreasing under the action of random unitary superchannels. A stronger statement was also proved: if a superchannel $\Theta:\mathcal{L}\left(A,B\right)\rightarrow \mathcal{L}\left(C,D\right)$ maps the completely depolarizing channel to the completely depolarizing channel, then the entropy of a quantum channel as defined in Eq.~(\ref{equ:definition_entropy_gour})  never decreases under the action of such quantum superchannels, provided $|B| = |D|$ and $|A|=|C|$. The fact that the entropy defined in  Eq.~(\ref{equ:definition_entropy_gour}), never decreases under the action of a random unitary superchannel follows as a consequence of this stronger statement.

Recall that under a positive trace-preserving and subunital map $\mathcal{N}$, the entropy of a quantum state never decreases i.e., $S(\mathcal{N}(\rho)) \geq S(\rho)$. We want to investigate if a similar statement holds for the entropy of quantum channels as well. We note that the generalization of the identity matrix is the completely depolarizing map~(see Definition \ref{def_completely_depolarising_map}) and hence the generalization of subunitality for a supermap $\mathrm{\Theta}$ is captured by $\mathrm{\Theta}(\mathcal{R}) \leq \mathcal{R}$, where the meaning of the ordering between the hermiticity-preserving maps can be found in Definition~\ref{Rsub_preserving_map}. Here, we ask whether the entropy of a quantum channel is nondecreasing under those quantum superchannels $\mathrm{\Theta}$ for which  $\mathrm{\Theta}(\mathcal{R}_{A \rightarrow B}) \leq \mathcal{R}_{C \rightarrow D}$ holds. To answer this question, first, we rewrite the definition of the entropy of a quantum channel which is in a similar spirit to Eq.~\eqref{equ:entropy}~(see Definition \ref{dfn:entropy_channel} ). However, to do so, we first need the notion of the relative entropy between a quantum channel and a completely positive~(CP) map. This is because, in the settings of quantum channels, the generalization of identity matrix is the completely depolarising map. In fact, the completely depolarising map is the identity element of the convolution algebra of maps acting on finite-dimensional matrix algebra~\cite{Sohail_2022}. In Proposition~\ref{prop:entropy_gain}, we have shown that for a quantum superchannel $\Theta$, if $\Theta(\mathcal{R}_{A \rightarrow B}) \leq  \mathcal{R}_{C \rightarrow D}$ holds, then the entropy of a quantum channel is non-decreasing under the action of $\Theta$, i.e., $S\left[\Theta(\mathcal{N}_{A \rightarrow B})\right] \geq S\left[\mathcal{N}_{A \rightarrow B}\right]$. We also observe that a necessary condition for the relation $S\left[\Theta(\mathcal{N}_{A \rightarrow B})\right] \geq S\left[\mathcal{N}_{A \rightarrow B}\right]$ to hold is $|B| \leq |D|$ (See Remark~(\ref{rem:r-pres})).

\subsection{Data processing inequality}
The generalized channel divergence between two given quantum channels $\mathcal{N}_{A\rightarrow B}$ and $\mathcal{M}_{A \rightarrow B}$ was defined in \cite{Cooney_2016,Felix_2018,Yuan}. In the following, we analogously define the relative entropy between a quantum channel and a CP map.
 \begin{dfn}[\cite{Cooney_2016,Felix_2018}]
    The relative entropy between a quantum channel  $\mathcal{N}_{A\to B}$ and a completely positive map $\mathcal{M}_{A\to B} $ is defined as
\begin{equation}
    \chD*{\mathcal{N}}{ \mathcal{M}} =  \underset{\rho_{RA}\in\mathcal{D}(R\otimes A)}{\mathrm{sup}}~\staD{\mathrm{id}_{R}\otimes\mathcal{N} (\rho_{RA})}{\mathrm{id}_{R}\otimes\mathcal{M}(\rho_{RA})},
\end{equation}
where $R$ is a reference system of arbitrary dimension. Due to purification, data processing inequality, and Schmidt decomposition, it suffices to optimize over the set of all pure bipartite states $\Psi_{RA}$
 \begin{equation}
      \chD*{\mathcal{N}}{ \mathcal{M}} =  \underset{\Psi_{RA}\in\mathcal{D}(R\otimes A)}{\mathrm{sup}}~\staD{\mathrm{id}_{R}\otimes\mathcal{N} (\Psi_{RA})}{ \mathrm{id}_{R}\otimes\mathcal{M}(\Psi_{RA})}.
 \end{equation}
\end{dfn}
It immediately follows that the relative entropy $ \chD*{\mathcal{N}}{ \mathcal{M}}$ is non-negative if the state that realizes the supremum, say $\widetilde{\Psi}_{RA}$, is such that $\mathrm{tr}\left(\mathrm{id}_{R}\otimes\mathcal{M}\left(\widetilde{\Psi}_{RA}\right)\right)\leq 1$. In the following proposition, we prove that the relative entropy $\chD*{\mathcal{N}}{ \mathcal{M}}$ is monotonic under the action of a quantum superchannel, i.e., it follows the data processing type inequality.

\begin{proposition}\label{equ:prop1}
    The relative entropy between a quantum channel $\mathcal{N}:\mathcal{L}\left({A}\right)\rightarrow\mathcal{L}\left({A'}\right)$ and a completely positive map $\mathcal{M}:\mathcal{L}\left({A}\right)\rightarrow\mathcal{L}\left({A'}\right)$  is nonincreasing under the action of a quantum superchannel $\Theta:\mathcal{L}\left({A},{A'}\right)\rightarrow \mathcal{L}\left({C},{C '}\right)$, i.e.,
    \begin{equation}
       \chD*{\Theta\left(\mathcal{N}\right) }{ \Theta \left(\mathcal{M}\right)} \leq \chD*{\mathcal{N}}{ \mathcal{M}}.
    \end{equation}
\end{proposition}
We have provided a detailed proof of the above proposition in Appendix~\ref{monotonicity_of_channeldivergence_proof}. Recall that the relative entropy between a state and a positive operator is monotonic under the action of a quantum channel. The above proposition tells us that the same properties also hold if we replace the state with a channel, the channel with a superchannel, and the positive operator with a CP map. Note that the relative entropy between two quantum states (or two positive semidefinite operators) is monotonic under a larger class of operations, namely, under all positive and trace-preserving maps (compared to a restricted class of CPTP maps) \cite{Hermes2017}. What happens to the monotonicity of relative entropy between two quantum channels if we restrict to process the quantum channels under the larger class of CP-preserving and TP-preserving supermaps (cf. Proposition \ref{equ:prop1})? We answer this question in the affirmative, and thereby generalize the monotonicity result of \cite{Hermes2017} to the higher order transformations. More specifically, if we restrict ourselves to a particular class of channels such that for any pair of channels $\mathcal{N}$ and $\mathcal{M}$ in this class, the supremum in the relative entropy functional occurs at some state $\Psi $ whose reduced state is full rank, then the monotonicity of relative entropy also holds for CP-preserving and TP-preserving supermaps, which we prove in our next proposition.

{\sloppy
\begin{proposition}\label{equ:prop_new}
    The relative entropy between quantum channels $\mathcal{N}:\mathcal{L}\left({A}\right)\rightarrow\mathcal{L}\left({A'}\right)$ and $\mathcal{M}:\mathcal{L}\left({A}\right)\rightarrow\mathcal{L}\left({A'}\right)$  is nonincreasing under the action of a CP-preserving and TP-preserving supermap $\Theta:\mathcal{L}\left({A},{A'}\right)\rightarrow \mathcal{L}\left({C},{C '}\right)$, i.e.,
    \begin{equation}
       \chD*{\Theta\left(\mathcal{N}\right) }{ \Theta \left(\mathcal{M}\right)} \leq \chD*{\mathcal{N}}{ \mathcal{M}},
    \end{equation}   
    if the pure states that realize the supremum in the relative entropy functional $\chD*{\mathcal{N}}{ \mathcal{M}}$ and $\chD*{\Theta \left(\mathcal{N}\right)}{ \Theta \left(\mathcal{M}\right) }$ have the full rank marginals, and the representing map of $\Theta$ is a trace nonincreasing map.
\end{proposition}}

The detailed proof of the above proposition, which is one of our main results, is provided in Appendix~\ref{equ:proof_of_prop2}. Further, we provide an explicit condition on supermap $\Theta$ in Eq.~(\ref{equ:tp_condition}) such that the monotonicity holds for the above-mentioned class of channels even if we take the supermap $\Theta$ to be CP-preserving but not TP-preserving. It would be interesting to see if the result holds for all quantum channels. We also extend some properties that hold for the relative entropy between a state and a positive operator to the relative entropy between a channel and a CP map, for example, if we have $\mathsf{P} \leq \mathsf{Q}$, where $\mathsf{P}$ and $\mathsf{Q}$ are positive operators, then for any state $\rho$, we also have $\staD*{\rho }{ \mathsf{P}} \geq \staD*{\rho}{\mathsf{Q}}$. We prove the channel version of this inequality as the following proposition.
\begin{proposition}\label{prop2}
    Consider a quantum channel $\mathcal{N}:\mathcal{L}\left({A}\right)\rightarrow\mathcal{L}\left({A'}\right)$ and two completely positive maps $\mathcal{M}:\mathcal{L}\left({A}\right)\rightarrow\mathcal{L}\left({A'}\right)$ and  $\widetilde{\mathcal{M}}:\mathcal{L}\left({A}\right)\rightarrow\mathcal{L}\left({A'}\right)$ with the ordering $\widetilde{\mathcal{M}}\geq \mathcal{M}$, then we have 
    \begin{equation}
        \chD*{\mathcal{N}}{\widetilde{\mathcal{M}}} \leq \chD*{\mathcal{N}}{\mathcal{M}}.
    \end{equation}
    \end{proposition}

 \begin{proof}
     Since $\widetilde{\mathcal{M}}$ and $\mathcal{M}$ are completely positive maps with the ordering $\widetilde{\mathcal{M}}\geq \mathcal{M}$, then the following ordering will also hold: 
         $\Delta_{\mathcal{M}, \widetilde{\mathcal{M}}}(\epsilon):=\log\left\{\left(\mathrm{id}_{R}\otimes\widetilde{\mathcal{M}}\right)\left(\Psi_{RA}\right)+ \epsilon \mathbbm{1}_{RA'}\right\} -\log \left\{ \left(\mathrm{id}_{R}\otimes \mathcal{M}\right)\left(\Psi_{RA}\right)+ \epsilon \mathbbm{1}_{RA'}\right\} \geq 0$ for any $\epsilon >0$.
     Now, we have 
     \begin{align*}
        0 & \leq \lim_{\epsilon\rightarrow 0+} \tr\left(\left(\mathrm{id}_{R}\otimes \mathcal{N}\right)\left(\Psi_{RA}\right)~ \Delta_{\mathcal{M}, \widetilde{\mathcal{M}}}(\epsilon)\right)\\
        &=\lim_{\epsilon\rightarrow 0+} 
        \staD*{\left(\mathrm{id}_{R}\otimes \mathcal{N}\right)\left(\Psi_{RA}\right) }{ \left(\mathrm{id}_{R}\otimes \mathcal{M}\right)\left(\Psi_{RA}\right) + \epsilon \mathbbm{1}_{RA'}} \\
        &- \staD*{\left(\mathrm{id}_{R}\otimes \mathcal{N}\right)\left(\Psi_{RA}\right)}{ \left(\mathrm{id}_{R}\otimes \widetilde{\mathcal{M}}\right)\left(\Psi_{RA}\right) + \epsilon \mathbbm{1}_{RA'}}\\
         &=\staD*{\left(\mathrm{id}_{R}\otimes \mathcal{N}\right)\left(\Psi_{RA}\right) }{ \left(\mathrm{id}_{R}\otimes \mathcal{M}\right)\left(\Psi_{RA}\right)} - \staD*{\left(\mathrm{id}_{R}\otimes \mathcal{N}\right)\left(\Psi_{RA}\right) }{ \left(\mathrm{id}_{R}\otimes \widetilde{\mathcal{M}}\right)\left(\Psi_{RA}\right) }.
     \end{align*}
     After taking the supremum over the set of pure states $\Psi_{RA}$, we obtain our desired inequality, i.e., $\chD*{\mathcal{N}}{ \widetilde{\mathcal{M}}} \leq \chD*{\mathcal{N}}{\mathcal{M}}$. 
 \end{proof}

\begin{proposition}\label{prop3}
    Consider a pair of quantum channels $\mathcal{N}:\mathcal{L}\left({A}\right)\rightarrow\mathcal{L}\left({A'}\right)$ and $\widetilde{\mathcal{N}}:\mathcal{L}\left({B}\right)\rightarrow\mathcal{L}\left({B'}\right)$, and another pair of completely positive maps $\mathcal{M}:\mathcal{L}\left({A}\right)\rightarrow\mathcal{L}\left({A'}\right)$ and  $\widetilde{\mathcal{M}}:\mathcal{L}\left({B}\right)\rightarrow\mathcal{L}\left({B'}\right)$, then we have 
    \begin{equation}
       \chD*{\mathcal{N}\otimes \widetilde{\mathcal{N}}}{\mathcal{M}\otimes \widetilde{\mathcal{M}}} \geq \chD*{\mathcal{N}}{\mathcal{M}}
        +\chD*{\widetilde{\mathcal{N}}}{ \widetilde{\mathcal{M}}}.
    \end{equation}
    \end{proposition}
The proof of this Proposition is provided in Appendix~\ref{Detailed_proof_of_Proposition4}.

\subsection{Entropy of quantum channels and its properties}
In this section, we study the entropy of a quantum channel and its behavior under the action of quantum superchannels. We begin by introducing a class of supermaps called $\mathcal{R}$-subpreserving supermaps, as we will later show that the entropy of a quantum channel is nondecreasing under the action of $\mathcal{R}$-subpreserving superchannels.
\begin{dfn}\label{Rsub_preserving_map}
    Let $\mathcal{R}_{A \rightarrow B}$ and $\mathcal{R}_{C \rightarrow D}$ be two completely depolarizing maps, i.e., $\mathcal{R}_{A \rightarrow B}(\mathsf{X})=\tr(\mathsf{X})\mathbbm{1}_B$ and $\mathcal{R}_{C \rightarrow D}(\mathsf{Y})=\tr(\mathsf{Y})\mathbbm{1}_D$ for all $\mathsf{X}\in \mathcal{L}(A)$ and $\mathsf{Y}\in \mathcal{L}(C)$. Then, a supermap $\Theta:\mathcal{L}\left(A,B\right)\rightarrow \mathcal{L}\left(C,D\right)$ is called $\mathcal{R}$-subpreserving if
\begin{eqnarray}
  \Theta(\mathcal{R}_{A \rightarrow B}) \leq \mathcal{R}_{C \rightarrow D},
\end{eqnarray}
where the above ordering means that the map $(\mathcal{R}_{C \rightarrow D} - \Theta(\mathcal{R}_{A \rightarrow B}))$ is completely positive.
\end{dfn}

The entropy of a quantum channel is formally defined in terms of the relative entropy between the quantum channel and the completely depolarizing map as follows:
\begin{dfn}\label{dfn:entropy_channel}
Consider a quantum channel $\mathcal{N}_{A \rightarrow B}$, and the completely depolarizing map $\mathcal{R}_{A \rightarrow B}$. The entropy of $\mathcal{N}_{A \rightarrow B}$  is defined as (cf. Eq.~\eqref{equ:definition_entropy_gour})
\begin{equation}
    S\left[\mathcal{N}_{A \rightarrow B}\right] := - \chD*{\mathcal{N}_{A \rightarrow B} }{ \mathcal{R}_{A \rightarrow B}}, \label{equ:channel_entropy_definition}
\end{equation}
where $\chD*{\mathcal{N}_{A \rightarrow B} }{ \mathcal{R}_{A \rightarrow B}}$ is the relative entropy between the quantum channel $\mathcal{N}_{A \rightarrow B}$ and the completely depolarizing map $\mathcal{R}_{A \rightarrow B}$.
\end{dfn}

Note that since the relative entropy between two quantum channels quantifies the minimum error in discriminating between them~\cite{Gour_entropy}, the definition of entropy (Definition \ref{dfn:entropy_channel}) is naturally endowed with the same operational meaning. That is, the entropy of a quantum channel quantifies how well it can be distinguished from the completely randomizing channel. We now show that the definition of the entropy functional has some nice properties: (1) it is nondecreasing under
random unitary superchannels, (2) it is additive under tensor product, and (3) it satisfies the normalization condition. First, we prove that the entropy functional of a quantum channel is additive.

\begin{proposition}\label{proposition_2}
    The entropy functional of quantum channels is additive under tensor product of channels, i.e,
   \begin{equation}
       S\left[{\mathcal{N}}_{A_1 \rightarrow B_1} \otimes {\mathcal{M}}_{A_2 \rightarrow B_2}\right] = S\left[{\mathcal{N}}_{A_1 \rightarrow B_1}\right] + S\left[{\mathcal{M}}_{A_2 \rightarrow B_2}\right],
   \end{equation}
where $\mathcal{N}_{A_{1} \rightarrow B_1}$ and $\mathcal{M}_{A_{2} \rightarrow B_2}$ are quantum channels. 
\end{proposition}

 To show that the entropy of quantum channels is additive, it is sufficient to show that the following equality holds:
   \begin{align}
        \chD*{\mathcal{N}_{A_1\rightarrow B_1}\otimes {\mathcal{M}}_{A_2 \rightarrow B_2} }{ \mathcal{R}_{A_1\rightarrow B_1}\otimes {\mathcal{R}}_{A_2 \rightarrow B_2}}  = \chD*{\mathcal{N}_{A_1 \rightarrow B_1}}{ \mathcal{R}_{A_1\rightarrow B_1}} + \chD*{\mathcal{M}_{A_2 \rightarrow B_2}}{\mathcal{R}_{A_2\rightarrow B_2}}.
   \end{align}

    The complete proof is presented in Appendix~\ref{prop2_channel_entropy_addition_proof}. Now we prove that the entropy functional attains its maximum for the completely depolarizing channel and is zero for a replacer channel that outputs a pure state.

\begin{proposition}
    Let $\mathcal{N}_{A \rightarrow B}$ be a replacer channel, i.e., for all $\rho \in \mathcal{D}(A)$, it satisfies $\mathcal{N}_{A \rightarrow B}(\rho)=\sigma_{0}$ for a fixed $\sigma_{0}\in \mathcal{D}(B)$. Then $S(\mathcal{N}_{A \rightarrow B})=0$ when $\sigma_{0}$ is a pure state, and $S(\mathcal{N}_{A \rightarrow B})$ is maximum, i.e., $S(\mathcal{N}_{A \rightarrow B})= \log|B|$ when $\sigma_{0}=\frac{1}{|B|}\mathbbm{1}_B$, i.e., when $\mathcal{N}_{A \rightarrow B}$ is the completely depolarizing channel.
\end{proposition}

\begin{proof}
    Let $\Psi_{RA}$ be an arbitrary pure state. It can be shown that $\mathrm{id}_{R} \otimes \mathcal{N}_{A \rightarrow B} (\Psi_{RA})=\Psi_{R} \otimes \sigma_{0}$ and $\mathrm{id}_{R} \otimes \mathcal{R}_{A \rightarrow B} (\Psi_{RA})=\Psi_{R} \otimes \mathbbm{1}_{B}$. Now, we have
    \begin{align*}
        \staD*{\mathrm{id}_{R} \otimes \mathcal{N}_{A \rightarrow B} (\Psi_{RA}) }{  \mathrm{id}_{R} \otimes \mathcal{R}_{A \rightarrow B} (\Psi_{RA})}= \staD*{\Psi_{R} \otimes \sigma_{0}}{\Psi_{R} \otimes \mathbbm{1}_{B}}= \staD*{\sigma_{0}}{ \mathbbm{1}_{B}}.
    \end{align*}
    Hence, the entropy of the replacer channel is given by $
    S\left[\mathcal{N}_{A \rightarrow B}\right]= -\staD*{\sigma_{0}}{  \mathbbm{1}_{B}}=S\left(\sigma_{0}\right)$. Noticing that $S\left(\sigma_{0}\right)=0$ when $\sigma_{0}$ is a pure state and $S\left(\sigma_{0}\right)=\log|B|$ when $\sigma_{0}=\mathbbm{1}_{B}/|B|$, completes the proof.
\end{proof}

Another crucial property that the entropy functional of a quantum channel has to satisfy is that it should be nondecreasing under the action of a random unitary superchannel. The definition used in \cite{Gour} satisfies this property as a corollary of the fact that the channel entropy always increases under those quantum superchannels which preserve the completely depolarizing channel. Here, we generalize this result and show that the entropy functional is nondecreasing under completely depolarizing map subpreserving superchannels. Random isometry superchannels are one of the examples of such class of superchannels, and we prove this in the following  proposition.

\begin{dfn}
    A random isometry superchannel $\Theta:\mathcal{L}\left(A ,B\right)\rightarrow \mathcal{L}\left(C ,D\right)$ is defined as:
\begin{eqnarray}
    \mathrm{\Theta}\left(\mathcal{N}_{A \rightarrow B}\right)= \sum_{i} p_{i} \mathcal{V}_{B \rightarrow D}^{i}\circ \mathcal{N}_{A \rightarrow B} \circ \mathcal{U}_{C \rightarrow A}^{i}, \label{Random_isometry_super_channel}
\end{eqnarray}
where $\{p_{i}\}_i$ is a probability distribution, and $\mathcal{V}_{B \rightarrow D}^{i}$, and $\mathcal{U}_{C \rightarrow A}^{i}$ are isometry channels defined as $\mathcal{V}_{B \rightarrow D}^{i}\left(\mathsf{X}\right)=v^{i}\mathsf{X}v^{i\dagger}$ and $\mathcal{U}_{C \rightarrow A}^{i}\left(\mathsf{Y}\right)=u^{i}\mathsf{Y}u^{i\dagger}$ respectively. Here, $u^i:C \rightarrow A$ and $v^i:B \rightarrow D$ are isometries.
\end{dfn}
\begin{proposition}\label{Rtheta_inequality}
    A random isometry superchannel $\Theta:\mathcal{L}\left(A ,B\right)\rightarrow \mathcal{L}\left(C ,D\right)$ is an $\mathcal{R}$-subpreserving superchannel.
\end{proposition}

\begin{proof}
Consider a random isometry superchannel $\Theta$ given by Eq.~\eqref{Random_isometry_super_channel} and a positive semidefinite operator $\mathsf{X} \in \mathcal{L}\left({C}\right)$. Then, we have $\left[\Theta(\mathcal{R}_{A \rightarrow B})\right]\left(\mathsf{X}\right)=\sum_{i} p_{i} \mathcal{V}_{B \rightarrow D}^{i}\circ \mathcal{R}_{A \rightarrow B} \circ \mathcal{U}_{C \rightarrow A}^{i}\left(\mathsf{X}\right)$. Using the definitions of $\mathcal{U}^{i}_{C \rightarrow A}$, and $\mathcal{R}_{A \rightarrow B}$, and $ \mathcal{V}^{i}_{B \rightarrow D}$, we get 
\begin{align}
\left[\Theta(\mathcal{R}_{A \rightarrow B})\right]\left(\mathsf{X}\right) =\sum_{i} p_{i} v^{i}v^{i \dagger} \mathbbm{1}_{D}  \mathrm{tr}\left(\mathsf{X}\right)
     =\sum_{i} p_{i} v^{i}v^{i \dagger} \mathcal{R}_{C \rightarrow D} \left(\mathsf{X}\right).
\end{align}
Defining a map $V : D \rightarrow D$ by $V=\sum_{i} p_{i} v^{i}v^{i \dagger}$, we have the following relation 
  \begin{eqnarray}
      \left[\Theta(\mathcal{R}_{A \rightarrow B})\right](\mathsf{X})= V \mathcal{R}_{C \rightarrow D} (\mathsf{X})= \mathcal{R}_{C \rightarrow D} \left(\mathsf{X}\right) V. \label{equ:defV}
  \end{eqnarray}
Notice that $v^{i}v^{i \dagger}$ are projectors from $D$ onto the range of $v^{i}$, which implies that $(\mathbbm{1}_D-V) \geq 0$. Whenever $\mathsf{X}$ is positive semidefinite, $\mathcal{R}_{C \rightarrow D} (\mathsf{X})$ is positive semidefinite. As $(\mathbbm{1}_D-V)$ and  $\mathcal{R}_{C \rightarrow D} (\mathsf{X})$ commutes, their product is also positive semidefinite, i.e. $(\mathbbm{1}_D - V)\mathcal{R}_{C \rightarrow D} \left(\mathsf{X}\right) \geq 0$, which implies that the map $(\mathbbm{1}_D - V)\mathcal{R}_{C \rightarrow D}$ is positive. It can be shown that $(\mathbbm{1}_D - V)\mathcal{R}_{C \rightarrow D}$ is completely positive as well. Hence, we have the ordering $ \mathbbm{1}_D \mathcal{R}_{C \rightarrow D}\geq V \mathcal{R}_{C \rightarrow D}$, which implies 
\begin{eqnarray}
    \Theta(\mathcal{R}_{A \rightarrow B}) \leq  \mathcal{R}_{C \rightarrow D}. \label{equ:sub_R_preserving}
\end{eqnarray}
\end{proof}

 \begin{remark}\label{rem:r-pres}
     Recall that for an $\mathcal{R}$-subpreserving supermap $\Theta$, we have  $\Theta(\mathcal{R}_{A \rightarrow B}) \leq  \mathcal{R}_{C \rightarrow D}$. If we take $\Theta$ to be a random unitary superchannel, then we get $\Theta(\mathcal{R}_{A \rightarrow B}) =  \mathcal{R}_{C \rightarrow D}$ i.e., $\Theta$ becomes $\mathcal{R}$-preserving superchannel. Using Eq.~(\ref{equ:superchannel_action}), we observe that a necessary condition for a superchannel $\Theta:\mathcal{L}\left(A ,B\right)\rightarrow \mathcal{L}\left(C ,D\right)$ to be an $\mathcal{R}$-subpreserving superchannel is $|B| \leq |D|$. Whereas, a necessary condition for a superchannel $\Theta:\mathcal{L}\left(A ,B\right)\rightarrow \mathcal{L}\left(C ,D\right)$ to be an $\mathcal{R}$-preserving superchannel is $|B| = |D|$.
     
     If we consider a superchannel $\Theta:\mathcal{L}\left(A ,B\right)\rightarrow \mathcal{L}\left(C ,D\right)$ with the following action on an arbitrary quantum channel $\mathcal{N}_{A\to B}$, $\mathrm{\Theta}\left(\mathcal{N}\right) = \mathcal{V}\circ \mathcal{N}\circ \mathcal{P}$, where $\mathcal{V}_{B\to D}$ and $\mathcal{P}_{C\to A}$ are arbitrary quantum channels. Then, $\mathrm{\Theta}$ is $\mathcal{R}$-subpreserving if $\mathcal{V}$ is subunital and $\mathcal{R}$-preserving if $\mathcal{V}$ is unital. We note that there exists unital channels which cannot be expressed as a mixture of unitary operations~\cite{MM09}.
 \end{remark}

\begin{proposition}\label{prop:entropy_gain}
    Let $\Theta:\mathcal{L}\left(A ,B\right)\rightarrow \mathcal{L}\left(C,D\right)$ be an $\mathcal{R}$-subpreserving superchannel, then the entropy of a quantum channel $\mathcal{N}_{A\to B}$ is nondecreasing under the action of $\Theta$, i.e., \begin{equation}
     S\left[\Theta(\mathcal{N}_{A \rightarrow B})\right] \geq S\left[\mathcal{N}_{A \rightarrow B}\right].   
    \end{equation}
\end{proposition}
    \begin{proof}
    Let us assume that for a superchannel $\Theta:\mathcal{L}\left(A ,B\right)\rightarrow \mathcal{L}\left(C ,D\right)$ the inequality $\Theta(\mathcal{R}_{A \rightarrow B}) \leq  \mathcal{R}_{C \rightarrow D}$ holds, where the maps $\mathcal{R}_{A \rightarrow B}\in \mathcal{L}\left(A,B\right)$ and $\mathcal{R}_{C \rightarrow D}\in \mathcal{L}\left(C,D\right)$ are completely depolarising maps.  Then, for any positive operator $\rho_{RC}\in \mathcal{L}\left({R\otimes C}\right)$, where $R$ is a reference system, we have $( \mathrm{id}_{R} \otimes \Theta(\mathcal{R}_{A \rightarrow B}))\rho_{RC}\leq  (\mathrm{id}_{R} \otimes \mathcal{R}_{C \rightarrow D})\rho_{RC}$. Now, for any $\epsilon >0$, adding $\epsilon \mathbbm{1}$ on both sides so as to make both sides strictly positive operators, we also have the following inequality:
    $\log \{( \mathrm{id}_{R} \otimes \mathrm{\Theta}(\mathcal{R}_{A \rightarrow B}))\rho_{RC} + \epsilon \mathbbm{1}\} \leq  \log \{( \mathrm{id}_{R} \otimes \mathcal{R}_{C \rightarrow D})\rho_{RC} + \epsilon \mathbbm{1}\}$,
which further implies 
\begin{align}
    &-\tr \left[ \left\{( \mathrm{id}_{R} \otimes \mathrm{\Theta}(\mathcal{N}_{A \rightarrow B}))\rho_{RC}\right\}\log \left\{\left( \mathrm{id}_{R} \otimes \mathrm{\Theta}(\mathcal{R}_{A \rightarrow B})\right)\rho_{RC} + \epsilon \mathbbm{1}\right\} \right] \nonumber\\
   & \geq -\tr \left[ \{(\mathrm{id}_{R} \otimes \mathrm{\Theta}(\mathcal{N}_{A \rightarrow B}))\rho_{RC}\}  \log \left\{(\mathrm{id}_{R} \otimes \mathcal{R}_{C \rightarrow D})\rho_{RC} + \epsilon \mathbbm{1}\right\} \right].
\end{align}
Now we add $\mathrm{tr}\left[\left\{\left(\mathrm{id}_{R} \otimes \Theta\left(\mathcal{N}_{A \rightarrow B}\right)\right)\rho_{RA}\right\} \log \left\{\left( \mathrm{id} \otimes \Theta\left(\mathcal{N}_{A \rightarrow B}\right)\right)\rho_{RA}\right\}\right]$ on both sides of the above inequality and take the limit $\epsilon \rightarrow 0$ according to Eq.~(\ref{wilde 1}). After that we take the supremum over $\rho_{RC}$ to arrive at the following inequality.
\begin{align}
 & \sup_{\rho_{RC}} \chD*{\left(\mathrm{id}_{R} \otimes \Theta\left(\mathcal{N}_{A \rightarrow B}\right)\right)\rho_{RC}}{\left(\mathrm{id}_{R} \otimes \Theta\left(\mathcal{R}_{A \rightarrow B}\right)\right)\rho_{RC}} \\
  &\geq  \sup_{\rho_{RC}} \chD*{\left(\mathrm{id}_{R} \otimes \Theta\left(\mathcal{N}_{A \rightarrow B}\right)\right)\rho_{RC}}{\left(\mathrm{id}_{R} \otimes \mathcal{R}_{C \rightarrow D}\right)\rho_{RC}}.\nonumber
\end{align}
Observe that the left-hand side of the above inequality is nothing but the relative entropy between the quantum channel $\Theta\left(\mathcal{N}_{A \rightarrow B}\right)$ and the CP map $\Theta\left(\mathcal{R}_{A \rightarrow B}\right)$, and the right-hand side is the relative entropy between the quantum channel $\Theta\left(\mathcal{N}_{A \rightarrow B}\right)$ and the CP map $\mathcal{R}_{C \rightarrow D}$. Using the monotonicity of relative entropy between a quantum channel and a CP map under the action of quantum superchannels, we get $ \chD*{\mathcal{N}_{A \rightarrow B}}{\mathcal{R}_{A \rightarrow B}} \geq \chD*{\Theta\left(\mathcal{N}_{A \rightarrow B}\right)}{\mathcal{R}_{C \rightarrow D}}$, which implies that
\begin{eqnarray}
    S\left[\Theta\left(\mathcal{N}_{A \rightarrow B}\right)\right] \geq S\left[\mathcal{N}_{A \rightarrow B}\right].
\end{eqnarray}
This completes the proof of the proposition.
\end{proof}

The Proposition \ref{prop:entropy_gain} is an expression for the entropy gain of quantum channels under the action of $\mathcal{R}$-subpreserving superchannels. As a special case, this proposition yields an expression for the entropy gain of quantum channels under the action of random isometry superchannels, generalizing the result of the entropy gain of quantum states under noisy operations. The proof relies on the fact that the random isometry superchannels are sub-completely randomizing map-preserving superchannels~(see Proposition \ref{Rtheta_inequality}). Although it is clear that the entropy of quantum channels is nondecreasing under the action of random isometry superchannels, it does not provide any estimate of the change of entropy in terms of some remainder terms. In the following, we study the entropy gain of quantum channels in detail and estimate a remainder term for entropy gain for a class of quantum channels where we assume that the state that realizes the supremum has full rank reduced state. Note that the tele-covariant channels are an example of such class of channels where the supremum occurs at a maximally entangled state for which the reduced state is proportional to the identity operator. More specifically, we have the following theorem, which is one of our main results.

\begin{theorem}\label{theorem_3}

    Let $\Theta:\mathcal{L}\left(A, A'\right)\rightarrow \mathcal{L}\left(C ,C'\right)$ be a superchannel and $\mathcal{N} \in \mathcal{L}\left(A, A'\right)$ 
    be a quantum channel. Let  $\Psi_{\widetilde{A}A} \in \mathsf{FRank}\left(\widetilde{A}\otimes{A}\right)\cap \mathcal{D}\left(\widetilde{A}\otimes{A}\right)$ and $\Phi_{\widetilde{C}C} \in \mathsf{FRank}\left(\widetilde{C}\otimes{C}\right)\cap \mathcal{D}\left(\widetilde{C}\otimes{C}\right)$ be the pure states that realize the supremum in the entropy functional $S\left[\mathcal{N}\right]$ and $S\left[\Theta\left(\mathcal{N}\right)\right]$, respectively. Then
    \begin{equation}
        S\left[\Theta\left(\mathcal{N}\right)\right] - S\left[\mathcal{N}\right] \geq 
\staD*{\mathsf{C}_{\mathcal{N}}^{\Psi}}{\left(\mathsf{C}_{\mathcal{N}}^{\Psi}\right)_{\alpha}} + \Delta' \left[\Psi_{\widetilde{A}},\Phi_{\widetilde{C}}\right],
    \end{equation}
   where $\mathsf{C}_{\mathcal{N}}^{\Psi}\in\mathcal{D}(\widetilde{A}\otimes A)$ is the Choi state for the channel $\mathcal{N}$, $\left(\mathsf{C}_{\mathcal{N}}^{\Psi}\right)_{\alpha}:= \frac{1}{\alpha^{\alpha}} \left(\mathfrak{T}^{*} \circ \mathfrak{T}\left( \mathsf{C}_{\mathcal{N}}^{\Psi} \right)\right)^{\alpha}$
   and $\alpha:= \norm{\mathfrak{T}^{*}(\mathbbm{1})}_{\infty}$ with $\mathfrak{T}$ being the representing map of $\Theta$ and $\mathfrak{T}^{*}$ being its adjoint map, defined respectively in Eqs.~\eqref{T_theta_phi}~and~\eqref{equ:theta_caps}. Here,  $\Delta' \left[\Psi_{\widetilde{A}},\Phi_{\widetilde{C}}\right]:=S\left(\Psi_{\widetilde{A}}\right) - S\left(\Phi_{\widetilde{{C}}}\right)$, $\widetilde{A}\simeq A$, and $\widetilde{C}\simeq C$.
\end{theorem}
 We note that the term $\Delta' \left[\Psi_{\widetilde{A}},\Phi_{\widetilde{C}}\right]$ is the difference between the entanglement entropy of the states $\Psi_{\widetilde{A}A}$ and $\Phi_{\widetilde{C}C}$. The proof of the theorem is presented in Appendix~\ref{sec_proof_of_channel_entropy_gain}. In the theorem above, we can lower bound $\Delta' \left[\Psi_{\widetilde{A}},\Phi_{\widetilde{C}}\right]$ more succinctly if it can be proven that both the marginal states $\Psi_{\widetilde{A}}$ and $\Phi_{\widetilde{C}}$ are connected by a  positive  map. In Appendix~\ref{sec_entropy_gain_under_posive_maps}, we prove several entropy inequalities for the change in state under CP maps, unital CP maps, and positive maps. In the following we demonstrate that there always exist a CP map connecting $\Phi_{\widetilde{C}}$ to $\Psi_{\widetilde{A}}$. Let us define a map $\mathcal{F}: \mathcal{L}\left(\widetilde{C}\right) \rightarrow \mathcal{L}\left(\widetilde{A}\right)$ by $\mathcal{F} (\mathsf{Y})=\tr(\mathsf{Y}) \Psi_{\widetilde{A}} $. This map is CP and it satisfies $\mathcal{F}(\Phi_{\widetilde{C}})=\Psi_{\widetilde{A}}$. In the case when $\widetilde{A}=\widetilde{C}$, we have another example of a CP map that maps $\Phi_{\widetilde{C}}$ to $\Psi_{\widetilde{A}}$. As in our context we have $\Psi_{\widetilde{A}}> 0 $ and $\Phi_{\widetilde{C}}>0$, we construct a completely positive map $\mathcal{F}: \mathcal{L}\left(\widetilde{C}\right) \rightarrow \mathcal{L}\left(\widetilde{A}\right)$, as follows 
  \begin{align}
      \mathcal{F}(Y)= \left(\Psi_{\widetilde{A}}^{1/2}\Phi_{\widetilde{C}}^{-1/2} \right) Y \left(\Psi_{\widetilde{A}}^{1/2}\Phi_{\widetilde{C}}^{-1/2}\right)^\dag~~\forall Y\in \mathcal{L}\left(\widetilde{C}\right),
  \end{align}
  which satisfies $\mathcal{F}\left(\Phi_{\widetilde{C}}\right)=\Psi_{\widetilde{A}}$. Since a completely positive map is also a positive map, we can use Theorem \ref{result_entropy_gain_1} of Appendix \ref{sec_entropy_gain_under_posive_maps} to arrive at the following inequality, $\Delta' \left[\Psi_{\widetilde{A}},\Phi_{\widetilde{C}}\right]\geq \staD*{\Phi_{\widetilde{C}}}{\left(\Phi_{\widetilde{C}}\right)_{\gamma}}$, where $\gamma:= \norm{\mathcal{F}^{*}(\mathbbm{1}_{\widetilde{A}})}_{\infty}$ and  $\left(\Phi_{\widetilde{C}}\right)_{\gamma}:=\left(\frac{1}{\gamma}\right)^{\gamma}\left(\mathcal{F}^{*}\circ \mathcal{F}\left(\Phi_{\widetilde{C}}\right)\right)^\gamma$. Now, we can modify the lower bound on channel entropy gain as
  \begin{equation}
       S\left[\Theta\left(\mathcal{N}\right)\right] - S\left[\mathcal{N}\right] \geq \staD*{\mathsf{C}_{\mathcal{N}}^{\Psi}}{\left(\mathsf{C}_{\mathcal{N}}^{\Psi}\right)_{\alpha}} + \staD*{\Phi_{\widetilde{C}}}{\left(\Phi_{\widetilde{C}}\right)_{\gamma}}. \label{equ:channel_enropy_gain_bound}
  \end{equation}

\begin{remark}
    As there may exist more than one positive map $\mathcal{F}$ that maps $\Phi_{\widetilde{C}}$ to $\Psi_{\widetilde{A}}$, we can further tighten the bound on $\Delta' \left[\Psi_{\widetilde{A}},\Phi_{\widetilde{C}}\right]$ by taking supremum over all such $\mathcal{F}$ as follows:
    \begin{equation}
       S\left[\Theta\left(\mathcal{N}\right)\right] - S\left[\mathcal{N}\right] \geq \staD*{\mathsf{C}_{\mathcal{N}}^{\Psi}}{\left(\mathsf{C}_{\mathcal{N}}^{\Psi}\right)_{\alpha}} + \sup_{\substack{\mathcal{F}:\\ \mathcal{F}(\Phi_{\widetilde{C}})=\Psi_{\widetilde{A}}}}\staD*{\Phi_{\widetilde{C}}}{\left(\Phi_{\widetilde{C}}\right)_{\gamma}}. \nonumber
  \end{equation}
\end{remark}  
 
Deriving a non-negative lower bound on the entropy difference under the action of a superchannel as above holds significant implications for identifying a set of superchannels that are channel entropy nondecreasing. A trivial example of a channel entropy increasing superchannel is a superchannel that maps any input channel to a completely depolarising channel. If we consider superchannel $\Theta$ and channel $\mathcal{N}$ such that both the terms in the RHS of Eq.~\eqref{equ:channel_enropy_gain_bound} are greater than or equal to $0$, i.e., $ \staD*{\mathsf{C}_{\mathcal{N}}^{\Psi}}{\left(\mathsf{C}_{\mathcal{N}}^{\Psi}\right)_{\alpha}}\geq 0$, and $  \staD*{\Phi_{\widetilde{C}}}{\left(\Phi_{\widetilde{C}}\right)_{\gamma}}\geq 0$, we can easily conclude that superchannel $\Theta$ is entropy nondecreasing superchannel. As an example, in the following, we compute the lower bounds in Eq.~\eqref{equ:channel_enropy_gain_bound} for superchannels that map tele-covariant channels to tele-covariant channels.

\begin{example*}
For a tele-covariant quantum channel (see Definition~\ref{dfn:telecov}), the supremum in the entropy functional occurs at a maximally entangled state~\cite{Gour_2019}. The entropy of a tele-covariant channel $\mathcal{N}_{A\to B}$ is then given as
\begin{align}
S\left[\mathcal{N}\right]&=S\left(\mathrm{id}_{\widetilde{A}} \otimes \mathcal{N}\left(|A|^{-1}\Psi^+_{\widetilde{A}A}\right)\right)-S\left(|A|^{-1}\mathbbm{1}_A\right)\nonumber\\
    &=S\left(\mathsf{C}_{\mathcal{N}}\right)- \log|A|, \label{equ:pappu1}
\end{align}
where $|A|^{-1}\Psi^+_{\widetilde{A}A}$ is a maximally entangled state and $\mathsf{C}_{\mathcal{N}}$ is the Choi state of the channel $\mathcal{N}$. Let $\mathrm{cov}\left(\mathcal{L}(A ,B)\right)$  and $\mathrm{cov} \left(\mathcal{L}(C,D)\right)$ denote the set of all tele-covariant channels in $\mathcal{L}(A,B)$ and $\mathcal{L}\left(C,D\right)$, respectively. Now, consider a tele-covariance preserving superchannel $\Theta:\mathrm{cov}\left(\mathcal{L}(A,B)\right) \rightarrow \mathrm{cov} \left(\mathcal{L}(C,D)\right)$, i.e., a superchannel that maps a tele-covariant channel to a tele-covariant channel. The entropy of the tele-covariant channel $\Theta(\mathcal{N})$ is the following,
\begin{align}
S\left[\Theta\left(\mathcal{N}\right)\right]&=S\left(\left(\mathrm{id}_{\widetilde{C}} \otimes \Theta\left(\mathcal{N}\right)\right)\left(|C|^{-1}\Psi^+_{\widetilde{C}C}\right)\right)-S\left(|C|^{-1}\mathbbm{1}_C\right)\nonumber\\
\end{align}
where $|C|^{-1}\Psi^+_{\widetilde{C}C}$ is a maximally entangled state. Using Theorem \ref{theorem_3}, we can write the gain in the channel entropy as
\begin{align}
    S\left[\Theta\left(\mathcal{N}\right)\right] - S\left[\mathcal{N}\right] \geq \staD*{\mathsf{C}_{\mathcal{N}}}{\left(\mathsf{C}_{\mathcal{N}}\right)_{\alpha}} + \log{\frac{|A|}{|C|}}.
\end{align}
Observe that $\Delta' \left[|A|^{-1}\Psi^+_{\widetilde{A}},|C|^{-1}\Psi^+_{\widetilde{C}}\right]$ is already a constant in this case and equal to $\log\widetilde{d}$, where $\widetilde{d}= |A|/|C|$. Further, we have $\left(\mathsf{C}_{\mathcal{N}}\right)_{\alpha} =\left(\frac{\widetilde{d}}{\alpha |C|}\right)^{\alpha} \left(\mathsf{C}_{\mathrm{\Theta}^{*} \circ \mathrm{\Theta}(\mathcal{N})}\right)^\alpha $, where $\alpha=\widetilde{d} \norm{\mathsf{C}_{\Theta^{*}(\mathcal{R})}}_{\infty}$. Notice that if superchannel $\Theta$ is such that $|A| \geq |C| $, then we always have $\log \widetilde{d} \geq 0$.  Note that if $\tr\left( \left(\mathsf{C}_{\mathcal{N}}\right)_{\alpha}\right)\leq 1$, then we have $\staD*{\mathsf{C}_{\mathcal{N}}}{\left(\mathsf{C}_{\mathcal{N}}\right)_{\alpha}}\geq 0$.
\end{example*}

\subsection{Entropy of continuous-variable quantum channels}
In this section, we show that our definition of the entropy of discrete-variable quantum channels is easily extended to continuous-variable quantum channels, which is analogous to the definition of the entropy for quantum states in terms of relative entropy.

Let $A$ be a finite- or infinite-dimensional, separable Hilbert space and $B$ be an infinite-dimensional, separable Hilbert space. Let $\widehat{H}_B$ be a Hermitian operator representing the Hamiltonian associated with the quantum system $B$; the lowest eigenvalue of $\widehat{H}_B$ is bounded from below by $0$. Then, a thermal operator $\widehat{\tau}_{B}^{\beta}$ for $B$ with an inverse temperature $\beta\in[0,\infty]$ is given as
\begin{equation}\label{eq:thermalstate}
   \widehat{\tau}_B^{\beta}=\operatorname{e}^{-\beta \widehat{H}_B}.
\end{equation}
A thermal operator $ \widehat{\tau}_B^{\beta}$ is bounded and trace-class for all $\beta\in (0,\infty)$. We note that for $\dim(B)=\infty$ and as $\beta\to 0^+$, we get 
\begin{equation}
    \lim_{\beta\to 0^+} \widehat{\tau}_B^{\beta}=\mathbbm{1}_B.
\end{equation}
Let us consider a completely thermalizing map $\mathcal{R}^{\beta}_{A\to B}:\mathcal{L}({A})_+\to\mathcal{L}({B})_+$, whose action on all trace-class operator $\mathsf{X}_A \in\mathcal{L}({A})_+$ is given by
\begin{equation}
 \mathcal{R}^{\beta}_{A\to B} (\mathsf{X}_A)=\tr[\mathsf{X}_A] \widehat{\tau}_B^{\beta}.
\end{equation}

 We can now define $\beta$-parameterized generalized entropy functional $\mathbf{S}_{\beta}$ of an arbitrary quantum channel $\mathcal{N}_{A\to B}$, where $\dim(B)=\infty$, for $\beta\in(0,\infty)$ as
\begin{equation}
    \mathbf{S}_{\beta}[\mathcal{N}]:= - \BchD*{\mathcal{N}_{A\to B}}{\mathcal{R}^{\beta}_{A\to B}} =- \sup_{\rho_{RA}\in \mathcal{D}(R\otimes A)} \BstD*{\id_R\otimes\mathcal{N}_{A\to B}(\rho_{RA})}{\id_R\otimes\mathcal{R}^\beta_{ A\to B}(\rho_{RA})},
\end{equation}
where $R$ is of arbitrary dimension and $\mathbf{D}(\cdot\Vert\cdot)$ is a generalized divergence functional. It easily follows that 
\begin{align*}
    \BstD*{\id_R\otimes\mathcal{N}_{A\to B}(\rho_{RA})}{\id_R\otimes\mathcal{R}^\beta_{ A\to B}(\rho_{RA})} =  \BstD*{\mathcal{N}_{A\to B}(\rho_{RA})}{\rho_R\otimes \widehat{\tau}^\beta_B}.
\end{align*}

\begin{dfn}\label{dfn:gen_ent_channel}
    The generalized entropy functional $\mathbf{S}[\mathcal{N}]$ of an arbitrary quantum channel $\mathcal{N}_{A\to B}$, where $|B|=\infty$, is defined as
\begin{equation}
  \mathbf{S}[\mathcal{N}]:= \lim_{\beta\to 0^+}  \mathbf{S}_{\beta}[\mathcal{N}] = - \lim_{\beta\to 0^+} \BchD*{\mathcal{N}}{\mathcal{R}^{\beta}}. 
\end{equation}
Similarly, if the input states to the channel are accessible only from a set $\mathsf{Access}(R\otimes A)\subset \mathcal{D}(R\otimes A)$ of states because of some physical constraints, then we define the resource-constrained generalized entropy functional $\mathbf{S}_{\mathsf{Access}}[\mathcal{N}]$ of an arbitrary quantum channel $\mathcal{N}_{A\to B}$, where inputs to $\mathbf{D}[\mathcal{N}\Vert \mathcal{R}^{\beta}]$ are only from $\mathsf{Access}(R\otimes A)$ (cf.~\cite{Win16,Gour_entropy}), as
\begin{equation}
  \mathbf{S}_{\mathsf{Access}}[\mathcal{N}]:= \lim_{\beta\to 0^+}  \mathbf{S}_{_{\mathsf{Access}},\beta}[\mathcal{N}] = - \lim_{\beta\to 0^+} \mathbf{D}_{\mathsf{Access}}[\mathcal{N}~\Vert~ \mathcal{R}^{\beta}], 
\end{equation}
where
\begin{equation}
    \mathbf{D}_{\mathsf{Access}}[\mathcal{N}_{A\to B}~\Vert~ \mathcal{R}^{\beta}_{A\to B}] =\sup_{\rho_{RA}\in\mathsf{Access}(R\otimes A)} \BstD*{\id_R\otimes\mathcal{N}_{A\to B}(\rho_{RA})}{\id_R\otimes\mathcal{R}^\beta_{ A\to B}(\rho_{RA})}.
\end{equation}
\end{dfn}

\section{Sufficiency for the recoverability of quantum channels} \label{Sufficiency_of_superchannels}
 The data processing inequality states that the relative entropy between two states  $\rho$ and $  \sigma$ is nonincreasing under CPTP maps, i.e., $ \staD*{\rho}{\sigma} \geq 
 \staD*{\mathcal{N} \left(\rho\right)}{\mathcal{N} \left(\sigma \right)},$
where $\mathcal{N}\in \mathcal{L}(A,B)$ is some quantum channel. The necessary and sufficient condition for the saturation of this inequality is related to the sufficiency~(or reversibility) of the quantum channel $\mathcal{N}$~\cite{Petz:1986,Mosonyi_2004,Petz:1988}. A quantum channel is called sufficient for a set of states \(\{\rho_{i}\}_{i}\) if there exists another CPTP map $\mathcal{P}$ called the recovery map, such that for each $\rho_{i}$, we have  $\left(\mathcal{P} \circ \mathcal{N}\right)(\rho_{i}) = \rho_{i}$. Now, it can be immediately checked that a quantum channel $\mathcal{N}$ preserves the relative entropy between $\rho$ and $\sigma$ if there exists a recovery map  $\mathcal{P}$ for the set $\{\rho, \sigma\}$. Conversely, if a quantum channel $\mathcal{N}$ does not change the relative entropy between $\rho$ and $\sigma$, then its action on the set $\{\rho, \sigma\}$ is reversible and the recovery map, called the Petz recovery map $\mathcal{P}_{\sigma,\mathcal{N}}$, is known in terms of $\sigma$ and $\mathcal{N}$~\cite{Petz:1986,Petz:1988}. The Petz recovery map $\mathcal{P}_{\sigma,\mathcal{N}}:\mathcal{L}(B) \rightarrow \mathcal{L}(A)$ is a linear, completely positive, and trace nonincreasing map defined as

\begin{align}
    \mathcal{P}_{\sigma,\mathcal{N}}\left(\mathsf{X}\right) \equiv \sigma^{\frac{1}{2}}\mathcal{N}^{*}\left(\left(\mathcal{N}\left(\sigma\right)\right)^{-\frac{1}{2}} \mathsf{X}\left(\mathcal{N}\left(\sigma\right)\right)^{-\frac{1}{2}}\right)\sigma^{\frac{1}{2}},
\end{align}
where $\mathcal{N}^{*}$ is the adjoint map of $\mathcal{N}$. In recent years, some refinements of the data processing inequality have also been proposed in terms of a recovery map~\cite{Wilde_2015,Junge_2018,berta2015,Sutter_2016}. In ~\cite{Junge_2018}, it was shown that there exists a universal recovery map (channel) $\mathcal{P}^{\rm R}$ such that the monotonicity of the quantum relative entropy between two arbitrary quantum states $\rho$ and $\sigma$ under the action of an arbitrary quantum channel $\mathcal{N}\in \mathcal{L}(A,B)$ can be refined as

\begin{align}
  \staD*{\rho}{\sigma} - \staD{\mathcal{N} \left(\rho\right)}{\mathcal{N} \left(\sigma \right)}\geq-\log F\left( \rho, \left(\mathcal{P}^{\rm R} \circ \mathcal{N}\right)(\rho) \right),
\end{align}
where $F(\rho, \sigma):=\norm{\sqrt{\rho}\sqrt{\sigma}}^2_1$ is the fidelity between $\rho,\sigma\in\mathcal{D}(A)$. Here, the recovery channel $\mathcal{P}^{\rm R}$ is given by \cite{Junge_2018} (see also \cite{wilde2016})
\begin{equation}
    \mathcal{P}^{\rm R}\left(\mathsf{X}\right)= \int dt \alpha(t) \mathcal{P}^{t/2}_{\sigma, \mathcal{N}}(\mathsf{X})+\tr\left(\left(\mathbbm{1}_{A}-\Pi_{\mathcal{N}(\sigma)}\right)\mathsf{X}\right)\xi, \label{junge_recovery_channel}
\end{equation}
where $\xi \in \mathcal{D}\left(A\right)$, $\mathsf{X} \in \mathcal{L}(B)$, $\Pi_{\mathcal{N}(\sigma)}$ is the projection onto the support of $\mathcal{N}(\sigma)$, and $\alpha_{t} \equiv \pi\left(\cosh(\pi t)+1\right)^{-1}/2$ can be interpreted as probability density function on $t \in \mathbb{R}$. And $ \mathcal{P}^{t}_{\sigma,\mathcal{N}} \equiv \mathcal{U}_{\sigma, -t} \circ \mathcal{P}_{\sigma, \mathcal{N}}\circ \mathcal{U}_{\mathcal{N}(\sigma),t}$, where $\mathcal{U}_{\sigma,t}$ is a partial isometric map defined as  $\mathcal{U}_{\sigma,t}(\mathsf{X}) = \sigma^{it}\mathsf{X} \sigma^{-it}$. The fidelity between any two arbitrary states $\rho, \sigma$ is always bounded as $F(\rho,\sigma)\in[0,1]$, with $F(\rho,\sigma)=1$ if and only if $\rho=\sigma$ and $F(\rho,\sigma)=0$ if and only if $\rho\perp \sigma$. The term $\tr\left(\left(\mathbbm{1}_{A}-\Pi_{\mathcal{N}(\sigma)}\right)\mathsf{X}\right)\xi$ present in the above equation is there to ensure $\mathcal{P}^{\rm R}$ is trace-preserving. It can be easily checked that the recovery channel $\mathcal{P}^{\rm R}$ perfectly recovers $\sigma$ from $\mathcal{N}(\sigma)$, i.e., $\left(\mathcal{P}^{\rm R}\circ \mathcal{N}\right)(\sigma) = \sigma$, and also if $\rho$ and $(\mathcal{P}^{\rm R}\circ\mathcal{N})(\rho)$ are close then $-\log F\left(\rho, (\mathcal{P}^{\rm R}\circ \mathcal{N})(\rho)\right)$ will be a small number. So, the above refinement says that it is possible to perform a recovery operation in which one can perfectly recover one state and approximately recover the other states if the change in the quantum relative entropy between the two quantum states before and after the action of the quantum channel is small~\cite{Wilde_2015,Junge_2018}.

 Motivated by these, one can ask the following question: is it possible to reverse the action of a quantum superchannel $\Theta$ on some given pair of channels $\{\mathcal{N}, \mathcal{M}\}$? To answer this question, we first define the notion of the sufficiency of superchannels, in the similar spirit of the sufficiency of the quantum channels, as follows. A superchannel $\Theta$ is called sufficient for a set of channels \(\{\mathcal{N}_{i}\}_{i}\) if there exists another superchannel $\Theta^{\rm R}$, such that for each $\mathcal{N}_{i}$, we have  $\left(\Theta^{\rm R} \circ \Theta\right)\left(\mathcal{N}_{i}\right) = \mathcal{N}_{i}$. Note that the relative entropy between two quantum channels $\mathcal{N}, \mathcal{M} \in \mathcal{L}\left(A, A'\right)$ is also nonincreasing under the action of a quantum superchannel $\Theta:\mathcal{L}\left(A,A'\right)\rightarrow \mathcal{L}\left(C,C'\right)$, i.e., 
  $\chD*{\Theta\left(\mathcal{N}\right) }{\Theta \left(\mathcal{M}\right)} \leq \chD*{\mathcal{N}}{ \mathcal{M}}$ (see Proposition \ref{equ:prop1}, and also \cite{Cooney_2016,Felix_2018,Yuan}). If the action of the superchannel $\Theta$ on the set of channels $\{\mathcal{N}, \mathcal{M}\}$ can be reversed, i.e., if there exists another superchannel $\Theta^{\rm R}:\mathcal{L}\left(C,C'\right)\rightarrow \mathcal{L}\left(A,A'\right)$  such that
$\left(\Theta^{\rm R}\circ \Theta\right)\left(\mathcal{N}\right) = \mathcal{N}$, and $\left(\Theta^{\rm R}\circ \Theta\right)\left(\mathcal{M}\right) = \mathcal{M}$, then one can immediately show that the superchannel $\Theta$ does not change the relative entropy between channels $\mathcal{N}$ and $\mathcal{M}$. Here we prove that the converse is also true, i.e., if $\chD*{\Theta\left(\mathcal{N}\right) }{\Theta \left(\mathcal{M}\right)} = \chD*{\mathcal{N}}{ \mathcal{M}}$, then there exists a supermap $\Theta^{\rm R}:\mathcal{L}\left(C,C'\right)\rightarrow \mathcal{L}\left(A,A'\right)$ such that
$\left(\Theta^{\rm R}\circ \Theta\right)\left(\mathcal{N}\right) = \mathcal{N}$, and $\left(\Theta^{\rm R}\circ \Theta\right)\left(\mathcal{M}\right) = \mathcal{M}$, for channels and superchannels that obey certain symmetry. First, we prove a refinement of the data processing inequality for the relative entropy of quantum channels as the following theorem.

\begin{theorem}\label{theorem_4}
    Let $\Theta:\mathcal{L}\left(A,A'\right)\rightarrow \mathcal{L}\left(C,C'\right)$ be a superchannel in the class $\mathsf{SCT}$ (see Definition \ref{dfn_sct}) and  $\mathcal{N}_{A \rightarrow A'} $, and $\mathcal{M}_{A \rightarrow A'}$ be two quantum channels. Let  $\Psi_{\widetilde{A} A}\in \mathsf{FRank}\left(\widetilde{A}\otimes{A}\right)\cap \mathcal{D}\left(\widetilde{A}\otimes{A}\right)$ and $\Phi_{\widetilde{C} C} \in \mathsf{FRank}\left(\widetilde{C}\otimes{C}\right)\cap \mathcal{D}\left(\widetilde{C}\otimes{C}\right)$ be the pure states that realize the supremum in the relative entropy functional $\chD*{\mathcal{N}}{ \mathcal{M}}$ and $\chD*{\Theta\left(\mathcal{N}\right) }{\Theta \left(\mathcal{M}\right)}$, respectively. Then we have
    \begin{equation}
        \chD*{\mathcal{N}}{ \mathcal{M}} - \chD*{\Theta\left(\mathcal{N}\right) }{\Theta \left(\mathcal{M}\right)} \geq - \log F \left(\mathsf{C}^{\Psi}_{\mathcal{N}},\left(\mathcal{P}^{\rm R}\circ \mathfrak{T}'\right)\left(\mathsf{C}^{\Psi}_{\mathcal{N}}\right)\right),
    \end{equation}
   where $F(\rho, \sigma):=\norm{\sqrt{\rho}\sqrt{\sigma}}^2_1$ is the fidelity of states $\rho$ and $\sigma$,   $\mathsf{C}_{\mathcal{N}}^{\Psi}$ is the Choi state for the channel $\mathcal{N}_{A \rightarrow A'}$ with respect to $\Psi$ and  $\mathcal{P}^{\rm R}$ is the recovery map defined by Eq.~\eqref{junge_recovery_channel}. The map $\mathfrak{T}'$ is defined as $\mathfrak{T}'(\mathsf{X}):= \mathfrak{T}(\mathsf{X})+\left[\mathrm{tr}(\mathsf{X})-\mathrm{tr}(\mathfrak{T}(\mathsf{X}))\right] \sigma_{0}$, where $\sigma_{0}\in\mathcal{L}(C\otimes C')$ is chosen in such a way that $\mathfrak{T}'$ is a CPTP map (see Appendix \ref{quantum_channel_from_TP_map}) and $\mathfrak{T}:\mathcal{L}\left(A \otimes A'\right) \rightarrow \mathcal{L}\left(C \otimes C'\right)$ is the representing map of the superchannel $\Theta$~(see Remark \ref{rem:choi-state-choi-operator-relation_1} and Eq.~\eqref{equ:new_rep_1}).
\end{theorem}

 The detailed proof of this theorem is provided in Appendix~ \ref{detailde_proof_of_recovery}, while we sketch the proof of this theorem in the following. The basic ingredient behind the proof is the fact that a supermap induces a linear map, called the representing map, at the level of Choi matrices. The representing map $\mathfrak{T}$ of a superchannel $\Theta$ is in general a CP map that becomes trace-preserving if the adjoint of $\Theta$ satisfies a simple condition (see Eq.~\eqref{equ:tp_condition} of Appendix \ref{Entropy_of_quantum_channels}). This, in turn, simplifies to the fact that $\Theta^{*}$ is a completely depolarising map preserving~(see Lemma \ref{lemma5}) when $\Psi$ and $\Phi$ are maximally entangled states. Even if this trace-preserving condition on the representing map $\mathfrak{T}$ is not satisfied, since we have $\Theta \in \mathsf{SCT}$, we can always replace the representing map $\mathfrak{T}$ with the following CPTP map: $\mathfrak{T}'(\mathsf{X}):= \mathfrak{T}(\mathsf{X})+\left[\mathrm{tr}(\mathsf{X})-\mathrm{tr}(\mathfrak{T}(\mathsf{X}))\right] \sigma_{0}$ for all $\mathsf{X}\in \mathcal{L}(A\otimes A')$ with an appropriate choice of $\sigma_{0}\in\mathcal{L}(C\otimes C')$ (see Appendix \ref{quantum_channel_from_TP_map}). The map $\mathfrak{T}'$ satisfies $\mathfrak{T}'(\mathsf{C}^{\Psi}_{\mathcal{N}})=\mathfrak{T}(\mathsf{C}^{\Psi}_{\mathcal{N}})$ and $\mathfrak{T}'(\mathsf{C}^{\Psi}_{\mathcal{M}})=\mathfrak{T}(\mathsf{C}^{\Psi}_{\mathcal{M}})$. The remaining task is to find the recovery map $\mathcal{P}^{\rm R}:\mathcal{L}\left(C \otimes C'\right)\rightarrow \mathcal{L}\left(A \otimes A'\right)$ for the CPTP map $\mathfrak{T}'$. The recovery map $\mathcal{P}^{\rm R}$ is constructed in such a way that it always recovers $\mathsf{C}^{\Psi}_{\mathcal{M}}$ from $\mathfrak{T}'(\mathsf{C}^{\Psi}_{\mathcal{M}})$.

 \begin{remark}\label{rem:data-sat}
  Theorem \ref{theorem_4} provides the necessary and sufficient condition for the sufficiency of quantum superchannels. We can define a supermap $\Theta^{\rm R}:\mathcal{L}\left(C,C'\right)\rightarrow \mathcal{L}\left(A,A'\right)$, which we call the recovery supermap, corresponding to the recovery map $\mathcal{P}^{\rm R}$, such that it always recovers $\mathcal{M}$ from $\Theta\left(\mathcal{M}\right)$. The action of such a recovery supermap on $\widetilde{\mathcal{N}}_{C \rightarrow C'}$ is given as follows:
\begin{align}
   \Theta^{\rm R} \left(\widetilde{\mathcal{N}}_{C \rightarrow C'}\right) (\mathsf{X})= \mathrm{tr}_{A}\left((\mathsf{X}^{t} \otimes \mathbbm{1}) T^{-1}_{\Psi} \mathcal{P}^{\rm R}\left(\mathsf{C}^{\Phi}_{\widetilde{\mathcal{N}}}\right)\right). \label{equ:recovery_supermap}
\end{align}
It can be verified that if $\left(\mathcal{P}^{\rm R}\circ \mathfrak{T}\right)\left(\mathsf{C}^{\Psi}_{\mathcal{N}}\right) = \mathsf{C}^{\Psi}_{\mathcal{N}}$ then $\Theta^{\rm R}$ also recovers $\mathcal{N}$ from $\Theta\left(\mathcal{N}\right)$. So, the above refinement says that if the data processing inequality saturates, i.e.,  $  \chD*{\mathcal{N}}{ \mathcal{M}} = \chD*{\Theta\left(\mathcal{N}\right) }{\Theta \left(\mathcal{M}\right)} $, then there exists a recovery supermap $\Theta^{\rm R}$ such that 
\begin{align}
  \Theta^{\rm R}\circ \Theta \left(\mathcal{N}\right)   = \mathcal{N}~~\text{and}~~
   \Theta^{\rm R}\circ \Theta \left(\mathcal{M}\right)  = \mathcal{M}.
\end{align}
 \end{remark}
 
An immediate consequence of the above theorem is the following proposition (also see \cite[Observation 2]{DGP24}). 
\begin{proposition}\label{prop:rem-thm}
    Let $\mathcal{N}, \mathcal{M} \in \mathcal{L}\left(A,A'\right)$  be two tele-covariant quantum channels (with respect to the same unitary representation of a group)~\cite[Definition 4]{Das_2019}, and let $\Theta:\mathcal{L}\left(A,A'\right)\rightarrow \mathcal{L}\left(C,C'\right)$ be a tele-covariance preserving superchannel and $\Theta \in \mathsf{SCT}$. Then we have 
    \begin{equation}
         \chD*{\mathcal{N}}{ \mathcal{M}} - \chD*{\Theta\left(\mathcal{N}\right) }{\Theta \left(\mathcal{M}\right)} \geq - \log F \left(\mathsf{C}_{\mathcal{N}},\left(\mathcal{P}^{\rm R}\circ \mathfrak{T}'\right)\left(\mathsf{C}_{\mathcal{N}}\right)\right),
    \end{equation}
    where $\mathsf{C}_{\mathcal{N}}$ is the Choi state for channel $\mathcal{N}_{A \rightarrow A'}$, and $\mathcal{P}^{\rm R}$ and $\mathfrak{T}'$ are defined in Theorem~\ref{theorem_4}.
\end{proposition}
The proof of the proposition follows directly from the Theorem \ref{theorem_4} by noticing the following: when we take channels $\mathcal{N}$ and $\mathcal{M}$ to be covariant with respect to the same unitary representation of a group  and the superchannel $\Theta$ to be a tele-covariant channel preserving superchannel, $\mathsf{C}^{\Psi}_{\mathcal{N}} $ becomes $ \mathsf{C}_{\mathcal{N}}$ and the recovery supermap $\Theta^{\rm R}$ changes accordingly.

We now briefly mention an application of Theorem~\ref{theorem_4}, Remark~\ref{rem:data-sat}, and Proposition~\ref{prop:rem-thm} in quantum error correction protocol. Assume that we have a covariant error correction scheme~\cite{FNA+18,ZLJ21} where encoding is performed by a tele-covariant channel $\mathcal{M}$. If the encoding channel is corrupted into $\Theta(\mathcal{M})$, where $\Theta$ is a tele-covariant preserving superchannel, then the covariant error correction scheme may loose its efficacy. However, based on our results, since there exists a recovery supermap $\Theta^{\rm R}$ corresponding to $\Theta$ and $\mathcal{M}$, one can recover $\mathcal{M}$ by using $\Theta^{\rm R}\circ\Theta(\mathcal{M})=\mathcal{M}$, preserving the covariant error correction scheme. Further, our results show that the same recovery supermap $\Theta^{\rm R}$ can be used to exactly correct the error in another tele-covariant encoding channel $\mathcal{N}$ if ever $\mathcal{N}$ is transformed by $\Theta$, i.e., $\Theta^{\rm R}\circ\Theta(\mathcal{N})=\mathcal{N}$ provided that $\chD*{\mathcal{N}}{ \mathcal{M}} = \chD*{\Theta\left(\mathcal{N}\right) }{\Theta \left(\mathcal{M}\right)}$.

Next, we estimate the entropy gain of a tele-covariant channel under the action of tele-covariant channel preserving superchannels as an application of Theorem~\ref{theorem_4}. In particular, we have the following lemma.
\begin{lemma}
 \label{lem:tele-cov-ent-diff}
   Let $\mathcal{N}\in \mathcal{L}\left(A,A'\right)$ be a tele-covariant channel,  $\Theta:\mathcal{L}\left(A,A'\right)\rightarrow \mathcal{L}\left(C,C'\right)$ be a tele-covariant channel preserving superchannel and $\Theta \in \mathsf{SCT}$ such that $\Theta^{*}(\mathcal{R}_{C \rightarrow C'})=\frac{|C|}{|A|}\mathcal{R}_{A \rightarrow A'}$ and $\Theta(\mathcal{R}_{A \rightarrow A'}) \leq \frac{|C|}{|A|}\mathcal{R}_{C \rightarrow C'}$ . Then the change in the entropy of the quantum channel $\mathcal{N}$ is bounded as
       \begin{equation}
        S\left[\Theta\left(\mathcal{N}\right)\right] - S\left[\mathcal{N}\right] \geq \staD*{\mathsf{C}_{\mathcal{N}}}{\left(\widetilde{\mathcal{P}}^{\rm R}\circ \mathfrak{T}\right)\mathsf{C}_{\mathcal{N}}}+\log\frac{|A|}{|C|},
    \end{equation}
   where $\mathsf{C}_{\mathcal{N}}$ is the Choi state for channel $\mathcal{N}$ and $\mathfrak{T}\in \mathcal{L}\left(A\otimes A', C\otimes C'\right)$ is the representing map for superchannel $\Theta$ (see Remark \ref{rem:choi-state-choi-operator-relation_1} and Eq.~\eqref{equ:new_rep_1}). The map $\widetilde{\mathcal{P}}^{\rm R}\left(\mathsf{X}\right)$~(see \cite[Eq.~(36)]{Buscemi_2016}) is given by $\widetilde{\mathcal{P}}^{\rm R}\left(\mathsf{X}\right)  = {\mathfrak{T}}^{*}(\mathsf{X})+\tr\left(\left(\mathrm{id}-{\mathfrak{T}}^{*}\right)\mathsf{X}\right)\xi $, where $\xi \in \mathcal{L}(A\otimes A')$ is a quantum state, and $\mathsf{X} \in \mathcal{L}(C\otimes C')$.
 \end{lemma}

 \begin{proof}
Since $\mathcal{N}$ is a tele-covariant channel and $\Theta$ is a tele-covariant channel preserving superchannel, we have $S\left[\mathcal{N}\right]=S(\mathsf{C}_{\mathcal{N}})- \log\left(|A|\right)$ and $S\left[\Theta\left(\mathcal{N}\right)\right]=S\left(\mathsf{C}_{\Theta(\mathcal{N})}\right)- \log\left(|C|\right)$, where $\mathsf{C}_{\mathcal{N}}$ and $\mathsf{C}_{\Theta(\mathcal{N})}$ are the Choi states of $\mathcal{N}$ and $\Theta(\mathcal{N})$, respectively. Then the entropy difference becomes $S\left[\Theta\left(\mathcal{N}\right)\right] - S\left[\mathcal{N}\right] = S\left(\mathsf{C}_{\Theta(\mathcal{N})}\right) - S(\mathsf{C}_{\mathcal{N}}) + \log \widetilde{d} $, where $\widetilde{d}:= |A|/|C|$. Also, $\mathsf{C}_{\mathcal{N}}$ and $\mathsf{C}_{\Theta(\mathcal{N})}$ are connected by the representing map $\mathfrak{T}$ which is CP but not trace-preserving, in general. As the superchannel $\Theta$ satisfies $\Theta^{*}(\mathcal{R}_{C \rightarrow C'})=\frac{|C|}{|A|}\mathcal{R}_{A \rightarrow A'}$ and $\Theta(\mathcal{R}_{A \rightarrow A'}) \leq \frac{|C|}{|A|}\mathcal{R}_{C \rightarrow C'}$, using Eq.~(\ref{new_rep_3}) it can be shown that the representing map $\mathfrak{T}$ is a subunital quantum channel. Hence, the entropy difference then is lower bounded by~(see~\cite[Eq.~(37)]{Buscemi_2016}, also \cite{Das2018})
\begin{equation}
     S\left[\Theta\left(\mathcal{N}\right)\right] - S\left[\mathcal{N}\right] \geq  \staD*{\mathsf{C}_{\mathcal{N}}}{\left(\widetilde{\mathcal{P}}^{\rm R}\circ \mathfrak{T}\right)\left(\mathsf{C}_{\mathcal{N}}\right)}+ \log\frac{|A|}{|C|},
\end{equation}
where each term in the above inequality is defined in the statement of the lemma. This completes the proof.
 \end{proof}
 When $|A|=|C|$ holds, Lemma \ref{lem:tele-cov-ent-diff} tells that if the entropy difference before and after applying the superchannel is considerably small, then one can conclude that states $\mathsf{C}_{\mathcal{N}}$ and $\left(\widetilde{\mathcal{P}}^{\rm R}\circ \mathsf{T}'\right)(\mathsf{C}_{\mathcal{N}})$ are very close. That is by performing the recovery operation $\widetilde{\mathcal{P}}^{\rm R}$, one can approximately recover the state $\mathsf{C}_{\mathcal{N}}$, or equivalently it means that  there exists a recovery supermap $\Theta^{\rm R}$, defined as $ \Theta^{\rm R} (\widetilde{\mathcal{N})} (\mathsf{X})= \abs{A}\mathrm{tr}_{A}\left((\mathsf{X}^{t} \otimes \mathbbm{1})  \widetilde{\mathcal{P}}^{\rm R}(\mathsf{C}_{\widetilde{\mathcal{N}}})\right)$ that can approximately recover $\mathcal{N}$ from $\Theta(\mathcal{N})$. On more application side, noting that quantum channels form an indispensable part in any quantum network (recall that states, measurements and all physical transformations can be understood as linear maps in the quantum domain), it is of significant importance to protect quantum channels from unwanted noise~(some superchannel). Our inequalities provide a necessary and sufficient condition for this task in terms of entropy and relative entropy of channels.
 
\section{Generalized divergence and entropy of superchannels } \label{Generalized_divergence_and_entropy_of_superchannels}

In quantum information theory, the information is usually encoded in the state of the quantum system and then any quantum information processing task may be understood simply as a collection of parallel and sequential transformations (linear CPTP maps) acting on the  quantum states. However, one can encode quantum information in quantum channels too and then any quantum information processing task may be understood as a collection of parallel and sequential quantum channels and quantum superchannels acting on quantum states and quantum channels, respectively. In this light, one may think quantum channels as first order transformations and quantum superchannels as second order transformations. In fact, one can recursively do this to obtain arbitrary finite order processes mapping a lower order process to another process of the same lower order. Such a paradigm has attracted considerable attention recently, see e.g.~\cite{Giulio2008,Giulio2009,Bisio_2019, Anna2021, Hoffreumon2022}.

In Section~\ref{Entropy_of_quantum_channels}, we generalized the definition of entropy to quantum channels, where we argued that similar to the entropy of quantum states, the entropy of quantum channels can also be defined in terms of relative entropy between the given quantum channel and the completely depolarising map~(see Definition~\ref{equ:channel_entropy_definition}). Here, we first extend the definition of generalized divergences to superchannels using the generalized divergence between two quantum channels. Recall that a map $\mathbf{D}:\mathcal{L}\left(A\right)_{+} \times \mathcal{L}\left(A\right)_{+} \rightarrow \mathbb{R}$ is called a generalized state divergence if it satisfies data processing inequality: $ \BstD*{\rho}{\sigma} \geq \BstD*{\mathcal{N}(\rho)}{\mathcal{N}(\sigma)}$, where $\rho,\sigma \in \mathcal{L}\left(A\right)_{+}$ and $\mathcal{N}\in\mathcal{L}(A,A')$ is a quantum channel. Given a generalized state divergence $\mathbf{D}$, we define the generalized divergence between a quantum channel $\mathcal{N}\in \mathcal{L}\left(A, A'\right)$ and a CP map $\mathcal{M} \in \mathcal{L}\left(A, A'\right)$ as 
\begin{equation}\label{eq:gd-s-c}
    \BchD*{\mathcal{N}}{\mathcal{M}} = \sup_{\Psi_{RA}\in\mathcal{D}(R\otimes A)}\BstD*{\left(\mathrm{id}_{R} \otimes \mathcal{N} \right)(\Psi_{RA})}{ \left(\mathrm{id}_{R} \otimes \mathcal{M} \right)(\Psi_{RA})},
\end{equation}
where $R$ is a reference system of arbitrary dimension. Now, given a generalized divergence $\mathbf{D}$ between channels (cf.~\cite{Cooney_2016,Felix_2018,Yuan}), we define a divergence functional between a quantum superchannel $\Theta$ and  a supermap $\Gamma$, denoted by $\dt$, as follows. Let $\mathcal{L}( A, B)_{+} $ denotes the set of all quantum channels $\mathcal{N}\in \mathcal{L}( A, B)$.
 
\begin{dfn}
    Given a generalized divergence $\BchD*{\cdot}{\cdot}$ between a quantum channel and a CP map (as in Eq.~\eqref{eq:gd-s-c}), a  divergence between a quantum superchannel $\Theta:\mathcal{L}\left(A,B\right)\rightarrow \mathcal{L}\left(C,D\right)$ and a completely CP-preserving supermap $\Gamma:\mathcal{L}\left(A,B\right)\rightarrow \mathcal{L}\left(C,D\right)$ is defined as 
    \begin{equation}
      \BschD*{\Theta }{\Gamma} = \sup_{\mathcal{N} \in \mathcal{L}(R\otimes A, R\otimes B)_{+}}  \BchD*{\left(\mathrm{id} \otimes \Theta  \right)\left( \mathcal{N}  \right) }{\left( \mathrm{id} \otimes    \Gamma \right)\left(\mathcal{N}\right)},
      \label{equ:divergence_of_superchannels}
    \end{equation}
   where $\mathrm{id}$ is the identity quantum superchannel, and $R$ is an arbitrary Hilbert space.
\end{dfn}

We note that generalized divergence between superchannels was also defined in \cite{Hirche_2023}, but our definition differs notably and can be recursively generalized to higher-order processes. Let us first consider a class $\mathsf{SSChan}$ of super-superchannels, i.e., a class of higher-order processes that map superchannels in $\mathscr{L}\left[\mathcal{L}\left(A ,B\right),\mathcal{L}\left(C,D\right)\right]$ to superchannels in $\mathscr{L}\left[\mathcal{L}\left(A',B'\right),\mathcal{L}\left(C',D'\right)\right]$ such that for any $\spadesuit \in \mathsf{SSChan}$ we can write
\begin{equation}
    \spadesuit\left(\Theta\right)= \Lambda_{\text{post}} \circ (\Theta \otimes \id )\circ \Lambda_{\text{pre}} ~~~\forall~ \Theta \in \mathscr{L}\left[\mathcal{L}\left(A ,B\right),\mathcal{L}\left(C,D\right)\right],
    \label{equ:SUPERCHANNEL_th_order_rep}
\end{equation}
where $\Lambda_{\text{pre}}\in \mathscr{L}\left[\mathcal{L}\left(A' ,B'\right),\mathcal{L}\left(A\otimes R ,B\otimes R\right)\right]$ and $\Lambda_{\text{post}} \in \mathscr{L}\left[\mathcal{L}\left(C\otimes R ,D\otimes R\right),\mathcal{L}\left(C',D'\right)\right]$ are preprocessing and postprocessing superchannels, respectively, similar to the preprocessing and postprocessing quantum channels of  Eq.~(\ref{equ:superchannel_action}). Here, $R$ is the reference system on which $\id$ acts. Looking at the mathematical similarity to quantum channels, we expect that the dilation equality, Eq.~\eqref{equ:SUPERCHANNEL_th_order_rep}, will hold for all quantum super-superchannels. But, for the rest of this section, we restrict ourselves to the class $\mathsf{SSChan}$ of super-superchannels.   For any pair of superchannels $\Theta, \Gamma \in \mathscr{L}\left[\mathcal{L}\left(A ,B\right),\mathcal{L}\left(C,D\right)\right]$, our definition of the divergence $\dt$ for superchannels satisfies the monotonicity relation $\BschD*{\Theta}{ \Gamma} \geq \BschD*{\spadesuit\left(\Theta\right) }{ \spadesuit\left(\Gamma\right)}$ for all $\spadesuit\in\mathsf{SSChan}$. The proof goes as follows. For $\spadesuit\in \mathsf{SSChan}$, the superchannels $\spadesuit(\Theta)$ and $\spadesuit(\Gamma)$ can be written in terms of concatenation of superchannels given by Eq.~\eqref{equ:SUPERCHANNEL_th_order_rep}. Using our definition of divergence between two superchannels, we can write  
\begin{equation}
\label{eq:ss-rel-ent}
   \BschD*{\spadesuit\left(\Theta\right) }{ \spadesuit\left(\Gamma\right)}= \sup_{\mathcal{N} \in \mathcal{L}(R'\otimes A', R'\otimes B')_{+}}\BchD*{\left(\mathrm{id}_{R'} \otimes \spadesuit\left(\Theta\right)  \right)\left( \mathcal{N}  \right) }{ \left( \mathrm{id}_{R'} \otimes   \spadesuit\left(\Gamma\right)\right)\left(\mathcal{N} \right)},
\end{equation}
where $R'$ is the Hilbert space of a reference system. Using Eq.~\eqref{equ:SUPERCHANNEL_th_order_rep}, we can write $\left(\id_{R'} \otimes \spadesuit\left(\Theta\right) \right)\left(\mathcal{N}\right) = \left(\id_{R'} \otimes \Lambda_{\text{post}}\right)\left(\id_{R'}\otimes \left(\Theta \otimes \mathrm{id}_{R}\right)\right) \mathcal{F}^{\mathcal{N}}$, where $\mathcal{F}^{\mathcal{N}}:=\left(\id_{R'} \otimes \Lambda_{\text{pre}}\right) (\mathcal{N})$. Similarly, we have $\left(\id_{R'}\otimes \spadesuit\left(\Gamma\right)\right)\left(\mathcal{N}\right) = \left(\id_{R'} \otimes \Lambda_{\text{post}}\right)\left(\id_{R'}\otimes \left(\Gamma \otimes \mathrm{id}_{R}\right)\right) \mathcal{F}^{\mathcal{N}}$. Using these in the right-hand side of Eq.~\eqref{eq:ss-rel-ent}, we obtain 

\begin{align*}
&\BschD*{\spadesuit\left(\Theta\right) }{\spadesuit\left(\Gamma\right)}\\
&=\sup_{\mathcal{N} \in \mathcal{L}(R'\otimes A', R'\otimes B')_{+}} \BchD*{\left(\id_{R'} \otimes \Lambda_{\text{post}}\right)\left(\id_{R'}\otimes \left(\Theta \otimes \mathrm{id}_{R}\right)\right) \left(\mathcal{F}^{\mathcal{N}}\right) }{\left(\id_{R'} \otimes \Lambda_{\text{post}}\right)\left(\id_{R'}\otimes \left(\Gamma \otimes \mathrm{id}_{R}\right)\right) \left(\mathcal{F}^{\mathcal{N}}\right)}\\
&\leq \sup_{\widetilde{\mathcal{N}} \in \mathcal{L}(R'\otimes A\otimes R,~ R'\otimes B\otimes R)_{+}} \BchD*{\left(\id_{R'} \otimes \Lambda_{\text{post}}\right)\left(\id_{R'}\otimes \left(\Theta \otimes \mathrm{id}_{R}\right)\right) \left(\widetilde{\mathcal{N}}\right) }{\left(\id_{R'} \otimes \Lambda_{\text{post}}\right)\left(\id_{R'}\otimes \left(\Gamma \otimes \mathrm{id}_{R}\right)\right) \left(\widetilde{\mathcal{N}}\right)}\\
&\leq \sup_{\widetilde{\mathcal{N}} \in \mathcal{L}(R'\otimes A\otimes R,~ R'\otimes B\otimes R)_{+}} \BchD*{\left(\id_{R'}\otimes \left(\Theta \otimes \mathrm{id}_{R}\right)\right) \left(\widetilde{\mathcal{N}}\right) }{\left(\id_{R'}\otimes \left(\Gamma \otimes \mathrm{id}_{R}\right)\right) \left(\widetilde{\mathcal{N}}\right)}=\BschD*{\Theta}{\Gamma},
\end{align*}

where the first inequality follows from the fact that $ \{\left(\id_{R'} \otimes \Lambda_{\text{pre}}\right) (\mathcal{N}): \forall \hspace{0.2cm}\mathcal{N}\in \mathcal{L}(R'\otimes A', R'\otimes B')_{+}  \}\subseteq \mathcal{L}(R'\otimes A\otimes R,~ R'\otimes B\otimes R)_{+}$. The second inequality follows by using the fact that $\left(\id_{R'} \otimes \Lambda_{\text{post}}\right)$ is a valid quantum super-superchannel and the generalized channel divergence $\mathbf{D}\left[\cdot \Vert \cdot \right]$ is monotonic under quantum super-superchannels. Thus, we have $\BschD*{\Theta}{ \Gamma} \geq \BschD*{\spadesuit\left(\Theta\right) }{\spadesuit\left(\Gamma\right)}$.

 In the following, we define a notion of generalized entropy for quantum superchannels based on the notion of the generalized divergence. Recall that generalized divergences $\BstD*{\cdot }{ \cdot }$  are functionals that take a pair of operators, say $\mathsf{A}, \mathsf{B} \geq 0$, as input and give a real number as output. One can define generalized entropies by taking $\mathsf{B}$ to be the identity operator. For example, by taking $\BstD*{\cdot }{ \cdot }$  as relative entropy, and $\mathsf{B} = \mathbbm{1}$, we get von Neumann entropy, which we have already discussed in the previous section. Also,  by taking $\BstD*{\cdot }{ \cdot }$  as sandwiched Rényi divergence~\cite{Wilde_2014,MUller_2013}, and again $\mathsf{B} = \mathbbm{1}$, one can get $\alpha$-Rényi entropy. To define the entropy of quantum channels, we started with the relative entropy which is a generalized divergence, and replaced states with channels and identity with the completely depolarising map and optimized overall input space of the channel. Similarly, we define the entropy of superchannels by taking the generalized divergence of channels as the starting point and replacing channels with superchannels and replacing the completely depolarising map with the supermap that serves its generalization in the space of supermaps, and then optimizing over all input space of the superchannel. Recall that the completely depolarizing map was defined as $\mathcal{R}\left(\mathsf{X}\right) = \tr(\mathsf{X}) \mathbbm{1}$, where $\mathsf{X}$ was an arbitrary operator belonging to the input space. We now generalize this definition of the completely depolarizing map for the space of supermaps as follows. Let us denote by $\rt$ a completely depolarising supermap.

\begin{dfn}
The action of a completely depolarizing supermap $\rt:\mathcal{L}\left(A,A'\right)\rightarrow \mathcal{L}\left(C,C'\right) $ is defined as
    \begin{equation}
    \label{eq:gen-comp-depol}
    \rt(\mathsf{X})=\tr\left(\Cwhat_{\mathsf{X}}\right)\mathcal{R},
\end{equation}
where $\mathsf{X}$ is an arbitrary map in the input space of $\rt$, $\mathcal{R}$ is the completely depolarising map, and $\Cwhat_{\mathsf{X}}$ is the Choi operator for $\mathsf{X}$.
\end{dfn}
The map $\rt$ is a completely CP-preserving supermap. With this definition, we are now in a position to define the entropy of higher-order processes.

\begin{dfn}
We define the generalized entropy of a superchannel $\Theta:\mathcal{L}\left(A, A'\right)\rightarrow \mathcal{L}\left(C,C'\right)$ as 
\begin{equation}\label{eq:ent-ssc32}
    \mathbf{S}\left[\Theta\right]:=-\BschD*{\Theta } {\rt},
\end{equation} 
where $\BschD*{\Theta } {\rt}$ is a generalized divergence between the superchannel $\Theta$ and the completely depolarizing supermap $\rt$. If we fix the generalized divergence in Eq.~\eqref{eq:ent-ssc32} to be the quantum relative entropy, then we get the entropy $S[\Theta]$ of the superchannel $\Theta$, i.e.,
\begin{equation}\label{eq:ent-ssc2}
  {S}\left[\Theta\right]:= -\schD*{\Theta }{\rt}.
\end{equation} 
\end{dfn}
In the remaining section, we work only with Eq.~\eqref{eq:ent-ssc2}, i.e., the quantum relative entropy. We show that the entropy of superchannels, given by Eq.~\eqref{eq:ent-ssc2}, satisfies the following properties: (1) $S\left[\Theta\right]$ is nondecreasing under $\rt$-subpreserving super-superchannels with a mild assumption (see Proposition \ref{prop:ssch-ent-dec}), (2) $S\left[\Theta\right]$ is subadditive under tensor product of superchannels (see Proposition \ref{prop:ssch-add}), (3) For a replacer superchannel $\Theta_{0}$, which maps every channel to a fix channel $\mathcal{N}_{0}$, the entropy of $\Theta_{0}$ is upper bounded by the entropy of $\mathcal{N}_{0}$ up to a constant term $\log |A|$ (see Proposition \ref{prop:ssch-nomr-cond}). We first show that the entropy of superchannels is nondecreasing under $\rt$-subpreserving super-superchannels as the following proposition.

\begin{proposition}
\label{prop:ssch-ent-dec}
 Let  $\spadesuit: \mathscr{L}\left[\mathcal{L}\left(A ,B\right),\mathcal{L}\left(C,D\right)\right] \rightarrow \mathscr{L}\left[\mathcal{L}\left(A' ,B'\right),\mathcal{L}\left(C',D'\right)\right] $ be a quantum super-superchannel mapping quantum superchannels to quantum superchannels such that it satisfies decomposability as in Eq.~\eqref{equ:SUPERCHANNEL_th_order_rep}. Let $\rt:\mathcal{L}\left(A ,B\right)\rightarrow \mathcal{L}\left(C,D\right)$ and $\wrt:\mathcal{L}\left(A' ,B'\right)\rightarrow \mathcal{L}\left(C',D'\right)$ be the generalizations of completely depolarising map to the space of superchannels such that  $\spadesuit \left[\rt\right] \leq \wrt$ (i.e., $\wrt-\spadesuit \left[\rt\right]$ is completely CP-preserving). Then for any superchannel $\Theta:\mathcal{L}\left(A ,B\right)\rightarrow \mathcal{L}\left(C,D\right)$, we have $S\left[\spadesuit\left[\Theta\right]\right]\geq S\left[\Theta\right]$.
\end{proposition}
\begin{proof}
   Since the super-superchannel $\spadesuit$ satisfies the decomposability condition given by Eq.~\eqref{equ:SUPERCHANNEL_th_order_rep}, the data processing inequality holds for the pair $\Theta$ and $\rt$ under the super-superchannel $\spadesuit$. Thus, we have
    \begin{align*}
      &\sup_{\mathcal{\mathcal{N}}\in \mathcal{L}\left(R \otimes A,R \otimes B\right)_{+}} \chD*{\left(\mathrm{id}_{R'}\otimes\Theta\right)\left(\mathcal{N}\right)}{ \left(\mathrm{id}_{R'}\otimes\rt\right)\left(\mathcal{N}\right)} \\
      &\geq \sup_{\widetilde{\mathcal{\mathcal{N}}}\in \mathcal{L}\left(R'\otimes A',R'\otimes B'\right)_{+}}\chD*{\left(\mathrm{id}_{R'}\otimes\spadesuit\left[\Theta\right]\right)\left(\widetilde{\mathcal{N}}\right)}{ \left(\mathrm{id}_{R'}\otimes\spadesuit\left[\rt\right]\right)\left(\widetilde{\mathcal{N}}\right) } \\
      &\geq \sup_{\widetilde{\mathcal{\mathcal{N}}}\in \mathcal{L}\left(R'\otimes A',R'\otimes B'\right)_{+}}\chD*{\left(\mathrm{id}_{R'}\otimes\spadesuit\left[\Theta\right]\right)\left(\widetilde{\mathcal{N}}\right)}{\left(\mathrm{id}_{R'}\otimes\wrt\right)\left(\widetilde{\mathcal{N}}\right)},
    \end{align*}
where we have used Proposition~\ref{prop2} in the second line (by identifying $\left(\mathrm{id}_{R'}\otimes\spadesuit\left[\rt\right]\right)\left(\widetilde{\mathcal{N}}\right)$  with $\mathcal{M}$ and $\left(\mathrm{id}_{R'}\otimes\wrt\right)\left(\widetilde{\mathcal{N}}\right)$ with $\widetilde{\mathcal{M}}$ and noticing that $\mathcal{M}\leq \widetilde{\mathcal{M}}$ holds). Now, using the definition of entropy of superchannels, we obtain $S\left[\spadesuit\left[\Theta\right]\right] \geq S\left[\Theta\right]$.
\end{proof}

\begin{remark}
\label{rem:rem-decomp-sschan}
We have proved the Proposition \ref{prop:ssch-ent-dec} only for the class of super-superchannels that are decomposable as in Eq.~\eqref{equ:SUPERCHANNEL_th_order_rep}. However, the proposition only requires the data processing inequality for the pair $\Theta$ and $\rt$ under the super-superchannel $\spadesuit$. It will be interesting to show that the data processing inequality holds true for the pair $\Theta$ and $\rt$ under all the super-superchannels $\spadesuit$ and with that we can remove the condition of decomposability from the Proposition~\ref{prop:ssch-ent-dec}. We believe that all super-superchannels are decomposable in the form of Eq.~\eqref{equ:SUPERCHANNEL_th_order_rep}, however, a proof or a counter-proof is left open for now.
\end{remark}

\begin{proposition}
\label{prop:ssch-add}
    The entropy functional of quantum superchannels is subadditive under the tensor product of superchannels, i.e,
   \begin{align}
   \label{equ:additivity_of_entropy_superchannel}
       S\left[\Theta \otimes \Gamma \right] \leq S\left[\Theta\right] + S\left[\Gamma\right],
   \end{align}
   where $\Theta:\mathcal{L}\left(A ,B\right)\rightarrow \mathcal{L}\left(C ,D\right)$  and $\Gamma:\mathcal{L}\left(A' ,B'\right)\rightarrow \mathcal{L}\left(C' ,D'\right)$ are any two superchannels.
\end{proposition}

\begin{proof}
Let $\rt:\mathcal{L}\left(A ,B\right)\rightarrow \mathcal{L}\left(C ,D\right)$ and $\widetilde{\mathcal{R}}^{(2)}:\mathcal{L}\left(A', B'\right)\rightarrow \mathcal{L}\left(C', D'\right)$ be the generalizations of completely depolarising map to the space of superchannels (see Eq.~\eqref{eq:gen-comp-depol}). From the definition of the generalized divergence of superchannels, we have
\begin{align}
 \schD*{\Theta \otimes \Gamma}{\rt \otimes\wrt} = \sup_{\mathcal{N}\in \mathcal{L}(E\otimes A\otimes A', E\otimes B\otimes B')_{+}} \chD*{\left(\mathrm{id}_{E} \otimes \Theta \otimes \Gamma \right) \left(\mathcal{N}\right) }{ \left(\mathrm{id}_{E} \otimes \rt \otimes \wrt \right) \left(\mathcal{N}\right)}, \nonumber
\end{align}
where we can take the ancilla to be $E=R \otimes R'$. Now, using the facts that $\mathcal{L}(R\otimes R' \otimes A\otimes A', R\otimes R' \otimes B\otimes B')$ is isomorphic to  $\mathcal{L}(R\otimes A, R\otimes B) \otimes \mathcal{L}(R'\otimes A', R'\otimes B')$, and supremum over a set is always greater than or equal to the supremum over its subsets, we obtain
\begin{align*}
     &\sup_{\mathcal{N}\in \mathcal{L}(E\otimes A\otimes A', E\otimes B\otimes B')_{+}} \chD*{\left(\mathrm{id}_{E} \otimes \Theta \otimes \Gamma \right) \left(\mathcal{N}\right) }{ \left(\mathrm{id}_{E} \otimes \rt \otimes \wrt \right) \left(\mathcal{N}\right)}\\
     & \geq \sup_{ \substack{\mathcal{N}_1 \otimes \mathcal{N}_2: \\ \mathcal{N}_1\in \mathcal{L}(R\otimes A, R\otimes B)_{+}\\ \mathcal{N}_2\in \mathcal{L}(R'\otimes A', R'\otimes B')_{+}}} \chD*{\left(\mathrm{id}_{R} \otimes \Theta\right) \left(\mathcal{N}_1\right) \otimes \left(\mathrm{id}_{R'} \otimes \Gamma\right) \left(\mathcal{N}_2\right) }{ \left(\mathrm{id}_{R} \otimes \rt\right) \left(\mathcal{N}_1\right)  \otimes \left(\mathrm{id}_{R'} \otimes \wrt\right) \left(\mathcal{N}_2\right)} \\
   &\geq \sup_{\mathcal{N}_1\in \mathcal{L}(R\otimes A, R\otimes B)_{+}}  \chD*{ \left(\mathrm{id}_{R} \otimes \Theta\right) \left(\mathcal{N}_1\right)}{ \left(\mathrm{id}_{R} \otimes \rt\right) \left(\mathcal{N}_1\right) }\\
   &\hspace{0.5cm}+\sup_{\mathcal{N}_2\in \mathcal{L}(R'\otimes A', R'\otimes B')_{+}}  \chD*{\left(\mathrm{id}_{R'} \otimes \Gamma\right) \left(\mathcal{N}_2\right)}{\left(\mathrm{id}_{R'} \otimes \wrt\right) \left(\mathcal{N}_2\right)} \\
    &=\schD*{\Theta }{\rt} + \schD*{\Gamma}{\wrt},
\end{align*}
where the second inequality follows from Proposition \ref{prop3}. Thus, we have 
\begin{align*}
   -\schD*{\Theta \otimes\Gamma }{\rt\otimes\wrt} \leq -\schD*{\Theta }{\rt} - \schD*{\Gamma}{\wrt},
\end{align*}
which implies that $S\left[\Theta \otimes \Gamma \right] \leq S\left[\Theta\right] + S\left[\Gamma\right]$, completing the proof of the proposition.
\end{proof}

\begin{proposition}
\label{prop:ssch-nomr-cond}
    Let $\Theta_{0}:\mathcal{L}\left(A ,A'\right)\rightarrow \mathcal{L}\left(C ,C'\right)$ be a replacer superchannel defined as $\Theta_0\left(\mathsf{X}\right) := |A|^{-1}\tr\left(\Cwhat_{\mathsf{X}}\right)\mathcal{N}_{0}$, where $\mathsf{X}$ is an arbitrary map in the input space of $\Theta_{0}$, $\mathcal{N}_{0}\in \mathcal{L}(C, C')$ is some fixed channel, $\Cwhat_{\mathsf{X}}$ is the Choi operator for $\mathsf{X}$. Then
    \begin{align}
        S\left[\Theta_{0}\right] \leq S\left[ \mathcal{N}_{0}\right]+\log\abs{A}.
    \end{align}
    
\end{proposition}
\begin{proof}
   Let $\mathcal{N} \in \mathcal{L}(R \otimes A, R \otimes A')_{+}$ be an arbitrary bipartite quantum channel. Note that $\left(\mathrm{id}_{R}\otimes\Theta_{0}\right) \mathcal{N}$ and $\left(\mathrm{id}_{R}\otimes\rt\right) \mathcal{N}$ can be written as $ |A|^{-1} \mathcal{M}^{\mathcal{N}}\otimes \mathcal{N}_{0}$ and $\mathcal{M}^{\mathcal{N}}\otimes \mathcal{R}$, respectively for some CPTP map $|A|^{-1} \mathcal{M}^{\mathcal{N}}_{R \rightarrow R}$. Using the definitions of the replacer superchannel $\Theta_{0}$ and supermap $\rt$, we can write
   \begin{align}
      \sup_{\mathcal{\mathcal{N}}\in\mathcal{L}(R \otimes A, R \otimes A')_{+}} \chD*{\left(\mathrm{id}_{R}\otimes\Theta_{0}\right)\mathcal{N}}{ \left(\mathrm{id}_{R}\otimes\rt\right) \mathcal{N} }  
      &= \sup_{\mathcal{\mathcal{N}}\in\mathcal{L}(R \otimes A, R \otimes A')_{+}} \chD*{|A|^{-1} \mathcal{M}^{\mathcal{N}}\otimes \mathcal{N}_{0}}{ \mathcal{M}^{\mathcal{N}}\otimes \mathcal{R}} \nonumber\\ 
      & \geq \chD*{\mathcal{N}_{0} }{\mathcal{R}}-\log|A|,
    \end{align}
where the inequality follows from Proposition \ref{prop3}. This completes the proof of the proposition.
\end{proof}

\section{Discussion}\label{conclusion}
Entropic functionals such as the von Neumann entropy and the relative entropy have played a pivotal role in the development of quantum technologies as they are the widely used functionals to gauge the performance and efficacy of quantum information processing and computing protocols. With the recent advances in quantum information theory, there is a prevailing consensus that the quantum channels constitute the fundamental objects whose control and manipulation yield the most general quantum information processing and computing protocols. Therefore, it is of utmost importance to associate a notion of entropic functionals to the network of quantum channels.

In this work, we start by discussing the notions of entropies and the relative entropies of quantum channels and generalize them to include the higher order processes. Using the results we derive in this work on the entropy gain of a quantum state under positive maps, which generalize some of the existing results in the literature, we provide a compact formula for the change in entropy of a quantum channel belonging to a wide class of quantum channels under the action of superchannels that preserves this class. Tele-covariant channels belong to this wide class of channels and for them we compute the entropy change under tele-covariance preserving superchannels as an example of our results. Further, we define the notion of the sufficiency of superchannels and obtain the necessary and sufficient conditions for sufficiency in terms of entropic inequalities. We then go on to define the notions of generalized divergence and entropy for quantum super-superchannels and analyze their properties.

Following the similar ideas that we use to define generalized divergence and entropy of quantum superchannels, we can also define the generalized divergence and entropy of higher order processes. Given a generalized divergence between two $n^{\text{th}}$ order processes, we define the generalized divergence between two $(n+1)^{\text{th}}$ order processes as follows. Let $\mathcal{V}^{(n)}$ denote the vector space of $n^{\text{th}}$ order maps. We define $(n+1)^{\text{th}}$ order processes as the $(n+1)^{\text{th}}$ order maps that are completely $n^{\text{th}}$ order process preserving. Here $n=0$ represents quantum states, $n=1$ represents quantum channels, $n=2$ represents quantum superchannels, and so on. We denote the set of $n^{\text{th}}$ order processes by $\mathcal{V}^{(n)}_{+}$. Clearly, we have, $\mathcal{V}^{(n)}_{+} \subseteq \mathcal{V}^{(n)}$.

{\it Possible generalized divergence functional}.--- Given a generalized divergence $\BnchD*{\cdot}{\cdot}$ of $n^{\text{th}}$ order process, let us define the functional $\BnonechD*{\Theta^{(n+1)}}{\Gamma^{(n+1)}}$ between two $(n+1)^{\mathrm{\text{th}}}$ order processes $\Theta^{(n+1)}$ and $\Gamma^{(n+1)}$, where $n \in \mathbb{Z}_+$, as follows: 
    \begin{equation}
     \BnonechD*{\Theta^{(n+1)}}{\Gamma^{(n+1)}} = \sup_{\Omega^{(n)} \in  \left(\mathcal{W}^{(n)} \otimes \mathcal{V}^{(n)}\right)_{+} } \BnchD*{\left(\mathrm{id} \otimes \Theta^{(n+1)} \right)\left(\Omega^{(n)}\right)}{\left(\mathrm{id} \otimes \Gamma^{(n+1)}  \right)\left(\Omega^{(n)}\right)}, \label{equ:divergence_of_nth_order_process}
    \end{equation}
    where $\mathcal{W}^{(n)}$ denotes an ancillary space that contains all $n^{\text{th}}$ order processes of all possible finite (positive integers) dimensions and the supremum is taken with respect to the joint process $\left(\mathcal{W}^{(n)} \otimes \mathcal{V}^{(n)}\right)_{+}$. Here $\mathrm{id}$ is the identity of $(n+1)^{\text{th}}$ order process.
 Note that the generalized divergence for $n^{\text{th}}$ order process is a functional $\mathbf{D}^{(n)}:\mathcal{V}^{(n)}_{+} \times \mathcal{V}^{(n)}_{+} \rightarrow \mathbb{R}$, which satisfies the following monotonicity relation under $(n+1)^{\text{th}}$ order process:
    \begin{equation}
    \BnchD*{\Theta^{(n)}}{\Gamma^{(n)}} \geq \BnchD*{\Phi^{(n+1)}\left(\Theta^{(n)}\right)}{\Phi^{(n+1)}\left(\Gamma^{(n)}\right)}, \label{EEE_1}
\end{equation}
for all $\Theta^{(n)},\Gamma^{(n)} \in \mathcal{V}^{(n)}_{+}$ and $\Phi^{(n+1)} \in \mathcal{V}^{(n+1)}_{+}$. We expect that the following dilation theorem similar to Eq.~(\ref{equ:superchannel_action}) will hold for any $n+1$ order process: $\Phi^{(n+1)}\left(\Theta^{(n)}\right)= \Gamma_{\text{post}}^{(n)} \circ \left(\Theta^{(n)} \otimes \id \right)\circ \Lambda^{(n)}_{\text{pre}}$, where $\Gamma^{(n)}$ and $\Lambda^{(n)}$ are two given $n^{\text{th}}$ order processes similar to the preprocessing and postprocessing quantum channels as in Eq.~(\ref{equ:superchannel_action}). It is clear that the right-hand side is also an $n^{\text{th}}$ order process since it is a concatenation of $n^{\text{th}}$ order processes but we leave the proof of converse statement open. In Appendix~\ref{proof_of_monotonicity_of_dn}, we show that for all $\Phi^{(n+1)} \in \mathcal{V}^{(n+1)}_{+}$ and $\Theta^{(n)},\Gamma^{(n)}\in \mathcal{V}^{(n)}_{+}$ such that the above dilation holds, the functional $\BnonechD*{\cdot}{\cdot}$ (Eq.~\eqref{equ:divergence_of_nth_order_process}) becomes a valid generalized divergence.

Now, to define the generalized entropy, we need the generalization of completely depolarizing maps to higher order processes, which we do as follows. Let ${\mathcal{R}}^{(n)}$ denote the completely depolarising map of $n^{\text{th}}$ order process. Then, the action of $\mathcal{R}^{(n)}:\mathcal{V}^{(n-1)} \rightarrow \mathcal{V}^{(n-1)}$, where $\mathcal{V}^{(n-1)}$ denotes the space of all $(n-1)^{\text{th}}$ order processes, is defined as  $\mathcal{R}^{(n)}(\mathsf{X})= \widetilde{f}\left(\mathsf{X}\right)\mathcal{R}^{(n-1)}$, 
where $\mathsf{X}$ is an arbitrary $(n-1)^{\text{th}}$ order process and $\mathcal{R}^{(n-1)}$ is a completely depolarising map of $(n-1)^{\text{th}}$ order process, and $\widetilde{f}(\mathsf{X})$ is a linear functional that one has to choose suitably such that the map $\mathcal{R}^{(n)}$ becomes a completely $(n-1)^{\text{th}}$ order  process preserving map. With these notions, we define the entropy $\mathbf{S}^{(n)}\left[\Theta^{(n)}\right]$ of a higher-order process ${\Theta}^{(n)}:\mathcal{V}^{(n-1)} \rightarrow \mathcal{V}^{(n-1)}$ as 
\begin{equation}
    \mathbf{S}^{(n)}\left[\Theta^{(n)}\right]:=-\BnchD*{\Theta^{(n)}}{\mathcal{R}^{(n)}},
\end{equation} 
where $\BnchD*{\Theta^{(n)}}{\mathcal{R}^{(n)}}$ is a generalized divergence between $n^{\text{th}}$ order process $\Theta^{(n)}$ and the completely depolarizing map $\mathcal{R}^{(n)}$ of $n^{\text{th}}$ order process. 

A key component in proving various propositions and theorems related to the entropy and relative entropy of quantum channels in our paper is the dilation theorem for quantum superchannels~(i.e., Eq~\eqref{equ:superchannel_action}). If it can be shown that the dilation theorem holds for higher-order processes, our results can be extended to a broader class of maps. We leave it as an open question to study the properties of $\mathbf{S}^{(n)}\left[\Theta^{(n)}\right]$, and to investigate an operational meaning for the definition of the generalized divergence in the discrimination task of higher-order quantum processes. It would be interesting to study possible implications of data processing inequality between channels (cf.~Theorem~\ref{theorem_4}) on correction of gates in different information processing tasks. Another interesting direction would be to discuss the properties of conditional entropy, mutual information, and conditional mutual information of multipartite quantum channels that can be derived based on the definition of the relative entropy between quantum channels. We believe that our main results and meta results will find interesting applications in wide domains of quantum information theory, such as open quantum systems, quantum communication, quantum many-body systems, and other areas of quantum physics where entropy change and data processing inequalities play crucial role.

\section*{Acknowledgments}

We are grateful to Maarten Wegewijs and  Kaiyuan Ji for carefully reading our paper and pointing out an error to a previous version of Theorem 2. We thank Mark M Wilde for discussions. Sohail and VP thank the International Institute of Information Technology Hyderabad for the support and hospitality during their visits in the summer 2023 where part of this work was done. US and SD acknowledge support from the Ministry of Electronics and Information Technology (MeitY), Government of India, under Grant No. 4(3)/2024-ITEA and IIIT Hyderabad under the Faculty Seed Grant. SD acknowledges support from the Science and Engineering Research Board, Department of Science and Technology (SERB-DST), Government of India under Grant No. SRG/2023/000217.

\subsection*{Funding and/or Conflicts of interests/Competing interests}
 The authors declare that they have no conflict of interest.

\appendix
\section{Structure of the Appendix}
In the Appendices, we provide detailed proofs of the results (lemmas, propositions, and theorems) of the main text and offer some additional observations. The structure of the appendix is as follows. In Appendix~\ref{appa}, we provide the proofs of statements (which are about some desirable properties of spaces of channels and superchannels) from Section \ref{prelims&notations}. In Appendix~\ref{proof_of_binod}, we have proof of statements from Section~\ref{Entropy_of_quantum_channels}. This section also contains several bounds on entropy gain under CP, CP unital, and positive maps (see Theorem~\ref{result_entropy_gain_1}, Lemma~\ref{lemma_binods} and Proposition  \ref{prop_binod}). In Appendix~\ref{detailde_proof_of_recovery}, we provide the proof of Theorem \ref{theorem_4}, which is also one of our main results and it provides a refinement of the data processing inequality. In Appendix \ref{proof_of_monotonicity_of_dn}, we provide the proof of monotonicity of $\mathbf{D}^{(n)}$.

   \section{Proofs of statements in Section~\texorpdfstring{\ref{prelims&notations}}{}}\label{appa}
   \subsection{Proof of Lemma \ref{equ:lemma1}}\label{detailed_proof_lemma_1}

      Let $\{b^{A}_{j}\}_{j}$ denote an orthonormal basis for $\mathcal{L}\left({A}\right)$. Any operator $\mathsf{K} \in \mathcal{L}(A)$ can be expressed as a linear combination of the elements of $\{b^{A}_{j}\}_{j}$ as: $\mathsf{K}=\sum_{j} \inner{b_{j}^{A},\mathsf{K}} b_{j}^{A}$.
Using this equation, the inner product between linear maps $\mathcal{N}_{A \rightarrow B}$ and $\mathcal{M}_{A \rightarrow B}$ can be written as 
\begin{align*}
   \sum_{j} \inner{\mathcal{N}(b^{A}_{j}),\mathcal{M}(b^{A}_{j})} &= \sum_{i} \sum_{j} \inner{\mathcal{N}(b^{A}_{j}),\mathcal{M}\left(\inner{a^{A}_{i},b^{A}_{j}} a^{A}_{i}\right)}\\
   &= \sum_{i} \inner{\mathcal{N}\left(\mathsf{B}^{A}_{i}\right),\mathcal{M}( a^{A}_{i})}\\
   &= \sum_{i} \inner{\mathcal{N}(a^{A}_{i}),\mathcal{M}(a^{A}_{i})},
\end{align*}
  where $\{a^{A}_{i}\}$ is some orthonormal basis in $\mathcal{L}\left({A}\right)$ and we have used the fact that  $\mathsf{B}^{A}_{i}:=\sum_{j}\inner{b^{A}_{j},a^{A}_{i}}b^{A}_{j}=a^{A}_{i}$ in the last equality. This concludes the proof.

\subsection{Proof of Lemma \ref{theorem_1}}\label{detailed_proof_thm_1}

We will first show that "if part" holds. Note that $\sum_{ijkl}  \mathcal{E}^{A \rightarrow B}_{ijkl} \otimes  \mathcal{E}^{A \rightarrow B}_{ijkl}$ is a CP map. As we have assumed that the supermap $\Theta$ is completely CP-preserving, we have $\Lambda_{\Theta}=\left(\mathrm{id} \otimes \Theta \right) \sum_{ijkl}  \mathcal{E}^{A \rightarrow B}_{ijkl} \otimes \mathcal{E}^{A \rightarrow B}_{ijkl}$ is CP.

 We will now prove the nontrivial "only if" part. Given that $\Lambda_{\Theta}$ is CP,  we need to prove that if $ \Omega \in \mathcal{L}\left(E,E'\right) \otimes \mathcal{L}(A,B) $ is CP~(where $E$ and $E'$ are reference systems), then $\left({\mathrm{id}} \otimes \Theta\right) \Omega$ is also CP. 
  First, we prove the following lemma
\begin{lemma} \label{lamma:1} A map $ \Omega \in \mathcal{L}\left(E,E'\right) \otimes \mathcal{L}(A,B)$ is CP if and only if it can be written as 
\begin{equation}
    \Omega =\sum_{n}\sum_{\alpha \beta \gamma \delta} \sum_{p q r s} {\kappa}^{n}_{\alpha \beta \gamma \delta} \bar{{\kappa}}^{n}_{pqrs} \mathcal{E}^{{E} \rightarrow {E'}}_{\gamma r \alpha p} \otimes \mathcal{E}^{A \rightarrow B}_{\delta s \beta q}, \label{lamma1}
\end{equation}
where $\kappa^{n}_{\alpha \beta \gamma \delta}$ are complex numbers, $\Bar{\kappa}^{n}_{pqrs}$ denotes the complex conjugate of $\kappa^{n}_{pqrs}$~(these complex numbers depend on the CP map $\Omega$), and $\left\{\mathcal{E}^{{E} \rightarrow {E'}}_{\gamma r \alpha p}\right\}_{\gamma r \alpha p}$, $\left\{\mathcal{E}^{A \rightarrow B}_{\delta s \beta q}\right\}_{\delta s \beta q}$ are orthonormal basis for $\mathcal{L}\left(E,E'\right)$ and $\mathcal{L}\left(A,B\right)$, respectively. 
\end{lemma}
\begin{proof}
Let $\left\{e^{Q}_{ij}\right\}_{ij}$ be a canonical basis for the space $\mathcal{L}\left({Q}\right)$, where $ Q = \{A,B,E,E'\}$. Using the bra-ket notation we can write $ e^{Q}_{ij} = \ketbra{i}{j}_{Q}$, where $\left\{\ket{i}_{Q}\right\}_{i}$ is an orthonormal basis for $Q$. A map $\Omega$ is CP if and only if it has a Choi-Kraus decomposition as follows 
\begin{equation}
    \Omega\left(\mathsf{X}\right)=\sum_{n} K_{n} \mathsf{X} K^{\dagger}_{n}, \label{equ:kraus}
\end{equation}
where $K_{n}: E \otimes A \rightarrow E' \otimes B$ are linear operators, and $\mathsf{X} \in \mathcal{L}(E \otimes A)$. So, we can write $K_n$ in terms of basis vectors of $A, B,E$, and $E'$  as follows
\begin{equation}
    K_{n}= \sum_{ \alpha \beta \gamma \delta } \kappa^{n}_{\alpha \beta \gamma \delta} \ket{\alpha_{{E'}}}\ket{\beta_{{B}}} \bra{\gamma_{E}} \bra {\delta_{A}}, \label{equ:choie_kraus}
\end{equation}
where $\kappa^{n}_{\alpha \beta \gamma \delta}= \bra{\alpha_{{E'}}}\bra{\beta_{{B}}} K_{n} \ket{\gamma_{E}} \ket{\delta_{A}}$. Using the above form of $K_{n}$ in Eq.~\eqref{equ:kraus}, we get

\begin{align}
\nonumber
    \Omega\left(\mathsf{X}\right) & = \sum_{n} \sum_{ \alpha \beta \gamma \delta } \sum_{ pqrs } \kappa^{n}_{\alpha \beta \gamma \delta}  \Bar{\kappa}^{n}_{pqrs} \bra{\gamma_{E}} \bra {\delta_{A}} \mathsf{X} \ket{r_{E}} \ket {s_{A}} \ketbra{\alpha_{{E'}}}{p_{{E'}}} \otimes \ketbra{\beta_{{B}}}{q_{{B}}}\\
\nonumber    
    &= \sum_{n} \sum_{\alpha \beta \gamma \delta} \sum_{pqrs} \kappa^{n}_{\alpha \beta \gamma \delta}  \Bar{\kappa}^{n}_{pqrs} \inner{e^{E}_{\gamma r} \otimes e^{A}_{\delta s}, \mathsf{X} } e^{E'}_{\alpha p} \otimes e^{B}_{\beta q}\\
\nonumber    
    &= \sum_{n} \sum_{\alpha \beta \gamma \delta} \sum_{pqrs} \sum_{ijkl} \kappa^{n}_{\alpha \beta \gamma \delta}  \Bar{\kappa}^{n}_{pqrs} x_{ijkl}\inner{e^{E}_{\gamma r} \otimes e^{A}_{\delta s}, e^{E}_{ij} \otimes e^{A}_{kl} } e^{{E'}}_{\alpha p} \otimes e^{B}_{\beta q}\\
\nonumber    
     &=  \sum_{n} \sum_{\alpha \beta \gamma \delta} \sum_{pqrs} \sum_{ijkl} \kappa^{n}_{\alpha \beta \gamma \delta}  \Bar{\kappa}^{n}_{pqrs} x_{ijkl}\inner{e^{E}_{\gamma r}, e^{E}_{ij}} e^{{E'}}_{\alpha p} \otimes \inner{e^{A}_{\delta s}, e^{A}_{kl}} e^{B}_{\beta q}\\
\nonumber     
      &=  \sum_{n} \sum_{\alpha \beta \gamma \delta} \sum_{pqrs} \sum_{ijkl} \kappa^{n}_{\alpha \beta \gamma \delta}  \Bar{\kappa}^{n}_{pqrs} x_{ijkl} \mathcal{E}^{E \rightarrow {E'}}_{\gamma r \alpha p} \left(e^{E}_{ij}\right) \otimes \mathcal{E}^{{A} \rightarrow B}_{\delta s \beta q} \left(e^{A}_{kl}\right)\\
\nonumber      
       &=  \sum_{n} \sum_{\alpha \beta \gamma \delta} \sum_{pqrs}  \kappa^{n}_{\alpha \beta \gamma \delta}  \Bar{\kappa}^{n}_{pqrs} \mathcal{E}^{E \rightarrow {E'}}_{\gamma r \alpha p} \otimes \mathcal{E}^{A \rightarrow B}_{\delta s \beta q} \left(\sum_{\{ijkl\}} x_{ijkl}  e^{E}_{ij} \otimes e^{A}_{kl}   \right) \\
\nonumber       
        &= \sum_{n} \sum_{\alpha \beta \gamma \delta} \sum_{pqrs} \kappa^{n}_{\alpha \beta \gamma \delta}  \Bar{\kappa}^{n}_{pqrs} \mathcal{E}^{E \rightarrow {E'}}_{\gamma r \alpha p} \otimes \mathcal{E}^{A \rightarrow B}_{\delta s \beta q} \left( \mathsf{X} \right).
\end{align}
The second line simply follows from the definition of the Hilbert-Schmidt inner product and we have written back everything in a canonical basis. In the third line, we have used a decomposition of $ \mathsf{X} = \sum_{\{ijkl\}} x_{ijkl}  e^{E}_{ij} \otimes e^{A}_{kl}$ to rewrite $\mathsf{X}$ in terms of canonical basis, and $\{\mathcal{E}^{E \rightarrow {E'}}_{\gamma r \alpha p}\}_{\gamma r \alpha p}$, $\{\mathcal{E}^{{A} \rightarrow B}_{\delta s \beta q}\}_{\delta s \beta q}$ are basis for spaces  $\mathcal{L}\left(E,E'\right)$ and  $\mathcal{L}\left(A, B\right)$, respectively. The fourth and fifth line follows from the definition of basis in $ \{\mathcal{E}^{E \rightarrow {E'}}_{\gamma r \alpha p}\}$ and $\{\mathcal{E}^{{A} \rightarrow B}_{\delta s \beta q}\}$ (see Eq.\eqref{equ:channel_basis}). This concludes the proof of the lemma.
\end{proof} 
 Now given Eq~(\ref{lamma1}) we define a supermap ${\Xi}_{\mathsf{X}} : \mathcal{L}(A ,B)\rightarrow \mathcal{L}\left(E,E'\right)$ by linear extension which is given by the following equation
\begin{eqnarray}
    \Xi_{\mathsf{X}} \left( \mathcal{E}^{{A} \rightarrow {B}}_{\delta s \beta q}  \right)=\sum_{n} \sum_{\alpha \gamma r p} \kappa^{n}_{\alpha \beta \gamma \delta}  \Bar{\kappa}^{n}_{pqrs} \mathcal{E}^{{E} \rightarrow {E'}}_{\gamma r \alpha p}, \label{equ:linear_extension_supermap}
\end{eqnarray}
where $\kappa^{n}_{\alpha \beta \gamma \delta}$ are complex numbers defined in Eq.~\eqref{equ:choie_kraus} and $\Bar{\kappa}^{n}_{\alpha \beta \gamma \delta}$ is the complex conjugate of $\kappa^{n}_{\alpha \beta \gamma \delta}$. So, we can write $\Omega$ as follows
\begin{equation}
    \Omega = \left(\Xi_{\mathsf{X}} \otimes \mathrm{id}\right) \left( \sum_{\delta s  \beta  q } \mathcal{E}^{{A} \rightarrow B}_{\delta s \beta q} \otimes \mathcal{E}^{{A} \rightarrow B}_{\delta s \beta q} \right).
\end{equation}
Applying $\left(\mathrm{id} \otimes \Theta \right)$ on both sides of the above equation, we obtain 
\begin{align}
    \left(\mathrm{id} \otimes \Theta \right) \Omega &= \left(\Xi_{\mathsf{X}} \otimes \mathrm{id}\right) \left( \sum_{\delta s  \beta  q } \mathcal{E}^{{A} \rightarrow B}_{\delta s \beta q} \otimes \Theta \left( \mathcal{E}^{{A} \rightarrow B}_{\delta s \beta q} \right) \right) = (\Xi_{\mathsf{X}} \otimes \mathrm{id}) \Lambda_{\Theta}. \label{equ:Xi_lambda}
\end{align}
It was shown in Lemma~\ref{lamma:1} that a map $\Lambda_{\Theta} \in \mathcal{L}(A,B) \otimes \mathcal{L}(C,D) $ is CP if and only if it can be written as 
\begin{eqnarray}
    \Lambda_{\Theta}=\sum_{m}\sum_{ijkl} \sum_{abcd} b^{m}_{ijkl} \Bar{b}^{m}_{abcd} \mathcal{E}^{A \rightarrow B}_{kcia} \otimes \mathcal{E}^{C \rightarrow D}_{ldjb}. \label{equ:cp_condition_lambda}
\end{eqnarray}
The map $\Lambda_{\Theta}$ in Eq.~\eqref{equ:Xi_lambda} is CP, thus, the use of the above form for $\Lambda_{\Theta}$ in Eq.~\eqref{equ:Xi_lambda} leads to
\begin{align}
\nonumber
    (\mathrm{id} \otimes \Theta) \Omega &= \sum_{m}\sum_{ijkl} \sum_{abcd} b^{m}_{ijkl} \Bar{b}^{m}_{abcd}~\Xi_{\mathsf{X}}\left( \mathcal{E}^{A \rightarrow B}_{kcia} \right) \otimes \mathcal{E}^{C \rightarrow D}_{ldjb} \\
\nonumber    
    & = \sum_{n} \sum_{m}\sum_{ijkl} \sum_{abcd} \sum_{tuvw} b^{m}_{ijkl} \Bar{b}^{m}_{abcd} {\kappa}^{n}_{vitk} \Bar{\kappa}^{n}_{wauc} \mathcal{E}^{{E} \rightarrow {E'}}_{tuvw} \otimes \mathcal{E}^{{C} \rightarrow {D}}_{ldjb}\\  
    & = \sum_{nm} \sum_{jlbd} \sum_{tuvw} \left( \sum_{ik}  {\kappa}^{n}_{vitk} {b}^{m}_{ijkl} \right) \left( \sum_{ac} \Bar{{\kappa}}^{n}_{wauc} {\Bar{b}}^{m}_{abcd} \right) \mathcal{E}^{{E} \rightarrow {E'}}_{tuvw} \otimes \mathcal{E}^{{C} \rightarrow {D}}_{ldjb}, \label{pandu}
\end{align}
where the second line simply follows from the definition of $\Xi_{\mathsf{X}}$~(see Eq.~\eqref{equ:linear_extension_supermap}). 
In terms of canonical basis, the coefficients $\sum_{ik}  {\kappa}^{n}_{vitk} {b}^{m}_{ijkl}$ and $\sum_{ac}  \Bar{\kappa}^{n}_{wauc} \Bar{b}^{m}_{abcd}$ have the following form 

\begin{align}
\nonumber
    \sum_{ik}  {\kappa}^{n}_{vitk} {b}^{m}_{ijkl} &= \mathrm{tr} \left[ \left( \ketbra{j_{{D}}}{v_{{E'}}}\otimes \mathbbm{1}_{{B}} \right) K_{n} \left( \ketbra{t_{{E}}}{l_{{C}}} \otimes \mathbbm{1}_{A} \right) B_{m}^{t}\right ] \\
\nonumber    
    \sum_{ac}  \Bar{\kappa}^{n}_{wauc} \Bar{b}^{m}_{abcd} &= \mathrm{tr} \left[ \left( \ketbra{b_{{D}}}{w_{{E'}}} \otimes \mathbbm{1}_{B} \right) \bar{K}_{n} \left( \ketbra{u_{{A}}}{d_{{C}}} \otimes \mathbbm{1}_{A} \right) B_{m}^{\dagger} \right],
\end{align}
where $B_{m}:{A} \otimes C \rightarrow {B} \otimes {D}$ are linear operators, and $B_{m}^{t}$ is transpose of $B_{m}$. Now we define operators $F^{mn}: E \otimes {C} \rightarrow {E'} \otimes {D}$ such that its matrix elements are given by the following
\begin{equation}
\nonumber
    F^{mn}_{vjtl} = \mathrm{tr} \left[ \left( \ketbra{j_{{D}}}{v_{{E'}}}\otimes \mathbbm{1}_{{B}} \right) K_{n} \left( \ketbra{t_{{E}}}{l_{{C}}} \otimes \mathbbm{1}_{A} \right) {B^t_{m}} \right ].
\end{equation}

Note that the complex conjugate of the matrix elements of $F^{mn}$ are given by 
\begin{align}
\nonumber
    \Bar{F}^{mn}_{wbud} &= \mathrm{tr} \left[ \left( \ketbra{b_{{D}}}{w_{{E'}}} \otimes \mathbbm{1}_{B} \right) \bar{K}_{n} \left( \ketbra{u_{{A}}}{d_{{C}}} \otimes \mathbbm{1}_{A} \right) B_{m}^{\dagger} \right].
\end{align}
In terms of these matrix elements, Eq.~(\ref{pandu}) can be rewritten as the following
\begin{align}
\nonumber
    (\mathrm{id} \otimes \Theta) \Omega &= \sum_{n,m} \sum_{jlbd} \sum_{tuvw} F^{mn}_{vjtl}  \Bar{F}^{mn}_{wbud}~\mathcal{E}^{{E} \rightarrow {E'}}_{tuvw} \otimes \mathcal{E}^{{C} \rightarrow {D}}_{ldjb}.
\end{align}
Looking at the above form and using Lemma~(\ref{lamma:1}), it can be concluded that  $(\mathrm{id} \otimes \Theta) \Omega $ is CP. This completes the proof of the theorem.

\subsection{Proof of Lemma~\ref{lemma_3}}\label{sec:detailed_proof_of_lemma3}
We already have shown that a supermap $\Theta :\mathcal{L}\left(A ,B\right) \rightarrow \mathcal{L}\left(C,D\right)$ is completely CP-preserving if and only if $\Lambda_{\Theta}=\sum_{ijkl} \mathcal{E}^{A \rightarrow B}_{ijkl} \otimes \Theta \left( \mathcal{E}^{A \rightarrow B}_{ijkl}\right) \in \mathcal{L}\left(A ,B\right) \otimes \mathcal{L}\left(C,D\right)$ is CP, 
where $\left\{\mathcal{E}^{{A} \rightarrow B}_{\delta s \beta q}\right\}_{\delta s \beta q}$ is a basis for $\mathcal{L}\left(A,B\right)$ as defined in Eq.~(\ref{equ:channel_basis}). We can determine the complete positivity of $\Lambda_{\Theta}$ by checking the positivity of its Choi operator $ \Cwhat_{\Lambda_{\Theta}}$ as defined below
  \begin{eqnarray}
    \Cwhat_{\Lambda_{\Theta}}=\sum_{pqrs} \left( e^{{A}}_{pq} \otimes e^{C}_{rs}  \right) \otimes  \Lambda_{\Theta} \left( e^{A}_{pq} \otimes e^{C}_{rs}  \right) \in \mathcal{L}(A \otimes C\otimes B \otimes D),
\end{eqnarray}  
where $\{e^{Q}_{ij}\}_{ij}$ denotes a canonical basis for the space $\mathcal{L}\left(Q\right)$ and $ Q = \{A,B,C,D\}$. Using the form of $\Lambda_{\Theta}$ from Eq.~(\ref{Choi-type-operator}) in the above equation, we have 
\begin{equation}
    \Cwhat_{\Lambda_{\Theta}}=\sum_{rs} \sum_{ijkl}\left( e^{A}_{ij} \otimes e^{C}_{rs}  \right) \otimes \left( e^{B}_{kl} \otimes \Theta( \mathcal{E}^{A \rightarrow B}_{ijkl}) e^{C}_{rs}  \right). \label{CCTR}
\end{equation}
Now, the Choi operator of the representing map $\mathsf{T}$ of $\Theta$ is the following
\begin{align}
\nonumber
    \Cwhat_{\mathsf{T}} &= \sum_{ijkl}\left( e^{A}_{ij} \otimes e^{B}_{kl}  \right) \otimes \mathsf{T} \left( e^{A}_{ij} \otimes e^{B}_{kl} \right) = \sum_{ijkl}  \Cwhat_{\mathcal{E}^{
    A \rightarrow B}_{ijkl}} \otimes \mathsf{T} \left( \Cwhat_{\mathcal{E}^{A \rightarrow B}_{ijkl}} \right)= \sum_{ijkl}  \Cwhat_{\mathcal{E}^{A \rightarrow B}_{ijkl}} \otimes  \Cwhat_{\Theta \left({\mathcal{E}^{A \rightarrow B}_{ijkl}}\right)},
\end{align}
where we have used the fact that $e^{A}_{ij} \otimes e^{{B}}_{kl}$ is the Choi operator $\mathsf{C}_{\mathcal{E}^{{A} \rightarrow {B}}_{ijkl}}$ for $\mathcal{E}^{{A} \rightarrow {B}}_{ijkl}$ in the second equality, and the last equality directly follows from Eq.~\eqref{equ:rep_map2}. Interchanging the spaces $B$ and $C$ in the Eq.~(\ref{CCTR}), we have
\begin{align}
\nonumber
    \Cwhat_{\Lambda_{\Theta}} &:=\sum_{rs} \sum_{ijkl}\left( e^{A}_{ij} \otimes e^{B}_{kl}  \right) \otimes \left( e^{C}_{rs}  \otimes \Theta( \mathcal{E}^{A \rightarrow B}_{ijkl}) e^{C}_{rs}  \right)\\
\nonumber    
    &= \sum_{ijkl}\left( e^{A}_{ij} \otimes e^{B}_{kl}  \right) \otimes  \sum_{rs}  \left( e^{C}_{rs}  \otimes \Theta( \mathcal{E}^{A \rightarrow B}_{ijkl}) e^{C}_{rs}  \right)= \sum_{ijkl}  \Cwhat_{\mathcal{E}^{A \rightarrow B}_{ijkl}} \otimes  \Cwhat_{\Theta \left({\mathcal{E}^{A \rightarrow B}_{ijkl}}\right)}=\Cwhat_{\mathsf{T}}.
\end{align}
Now, the transformation $\Cwhat_{\Lambda_{\Theta}} \leftrightarrow \Cwhat_{\mathsf{T}}$ is an $^{\ast}$-isomorphism. So, $\Cwhat_{\mathsf{T}}$ is positive semidefinite if and only if $\Cwhat_{\Lambda_{\Theta}}$ is positive semidefinite. Now, $\Cwhat_{\Lambda_{\Theta}}$ is positive semidefinite if and only if $\Lambda_{\Theta}$ is CP, and $\Lambda_{\Theta}$ is CP if and only if $\Theta$ is completely CP-preserving. As $\mathsf{T}$ is CP if and only if $\Cwhat_{\mathsf{T}}$ is positive semidefinite, we have $\mathsf{T}$ is CP if and only if $\Theta$ is completely CP-preserving. This completes the proof.

\subsection{Quantum channels from representing maps}\label{quantum_channel_from_TP_map}

Let us define a map $\mathsf{T}':\mathcal{L}\left(A \otimes {B}\right) \rightarrow \mathcal{L}\left({C}\otimes {D}\right)$ as follows 
\begin{equation}
    \mathsf{T}{'}(\mathsf{X})=\mathsf{T}(\mathsf{X})+\mathsf{T}_{0}(\mathsf{X}), \label{b11}
\end{equation}
where $\mathsf{T}_{0}(\mathsf{X}):=\left[\mathrm{tr}(\mathsf{X})-\mathrm{tr}(\mathsf{T}(\mathsf{X}))\right] \sigma_{0}$, $\sigma_{0}\in \mathcal{L}(C\otimes D)$ and $\mathsf{T}:\mathcal{L}\left(A \otimes {B}\right) \rightarrow \mathcal{L}\left({C}\otimes {D}\right)$ is the representing map of the superchannel $\Theta:\mathcal{L}\left(A,B\right)\rightarrow \mathcal{L}\left(C,D\right)$. Note that the representing map $\mathsf{T}$ is CP but it could be trace non-increasing or trace increasing or neither trace increasing nor trace decreasing. However, if we take $\sigma_0$ to be an operator with unit trace then $ \mathsf{T}'$ becomes a  trace-preserving map but then $\mathsf{T}'$ may not be a CP map. In the following, we wish to find out if $\mathsf{T}'$ is  a CP map given that $\mathsf{T}$ is CP and $\tr(\sigma_0) = 1$. In the case where $\mathsf{T}$ is trace non-increasing, the condition $\sigma_0\geq 0$ in Eq.~\eqref{b11} guarantees that the map $\mathsf{T}'$ is CP. In other cases, we derive the necessary and sufficient conditions for $\mathsf{T}'$ to be CP as follows.

{\it Necessary and Sufficient Conditions}.---The Choi operator of the map $\mathsf{T}_{0}$ in Eq. \eqref{b11} is given by 
\begin{align}
    \Cwhat_{\mathsf{T}_{0}} =  \sum_{ijkl} (e_{ij}\otimes f_{kl} )\otimes \mathsf{T}_{0}(e_{ij} \otimes f_{kl}) = \left[\mathbbm{1}_{AB}-\mathsf{T}^{*}(\mathbbm{1}_{CD})\right]^{t} \otimes \sigma_0,
\end{align}
where $t$ denotes the transpose operation in the basis $\{e_{ij} \otimes f_{kl}\}$. Then Eq. \eqref{b11} in terms of Choi operators implies
\begin{align}
    \Cwhat_{\mathsf{T}'} = \Cwhat_{\mathsf{T}} + \left[\mathbbm{1}_{AB}-\mathsf{T}^{*}(\mathbbm{1}_{CD})\right]^{t} \otimes \sigma_0,
\end{align}
where $\Cwhat_{\mathsf{T}'}$ and $\Cwhat_{\mathsf{T}}$ are the Choi operators for the maps $\mathsf{T}'$ and $\mathsf{T}$, respectively.
We have freedom over $\sigma_0$ in the above equation and we need to choose it such that $\Cwhat_{\mathsf{T}'}$ becomes positive semidefinite. Note that $\left[\mathsf{T}^{*}(\mathbbm{1}_{CD})-\mathbbm{1}_{AB}\right]^{t}$ is a Hermitian operator and let its spectral decomposition be given by $\left[\mathsf{T}^{*}(\mathbbm{1}_{CD})-\mathbbm{1}_{AB}\right]^{t}=\sum_{m:\lambda_m>0}\lambda_m \Pi_m - \sum_{m:\lambda_m<0}|\lambda_m| \Pi_m+\sum_{m:\lambda_m=0}\lambda_m \Pi_m$, where $\{\Pi_m\}_{m:\lambda_m>0},~\{\Pi_m\}_{m:\lambda_m<0}$ and $\{\Pi_m\}_{m:\lambda_m=0}$ are orthogonal projectors. Further, define
$\Pi_+:=\sum_{m:\lambda_m>0} \Pi_m$ and $\Pi_-:=\sum_{m:\lambda_m<0}|\lambda_m| \Pi_m$ as the projectors onto the support of positive and negative eigenvalues of the Hermitian operator $\left[\mathsf{T}^{*}(\mathbbm{1}_{CD})-\mathbbm{1}_{AB}\right]^{t}$, respectively. We denote by $\Pi_0$ the projector onto the eigenspace corresponding to the zero eigenvalue of $\left[\mathsf{T}^{*}(\mathbbm{1}_{CD})-\mathbbm{1}_{AB}\right]^{t}$. One can now prove easily that $\Cwhat_{\mathsf{T}'}$ is positive semidefinite  if and only if the operator $\sigma_{0}\in \mathcal{L}(C\otimes D)$ satisfies
\begin{equation}
    \left(\Pi_{+}-\Pi_{-}\right) \otimes \sigma_0 \leq \left(\left[|(\mathsf{T}^{*}(\mathbbm{1}_{CD}))^t-\mathbbm{1}_{AB}\right|^{-\frac{1}{2}}+\Pi_0]\otimes \mathbbm{1}_{CD}\right) \Cwhat_{\mathsf{T}} \left([\left|(\mathsf{T}^{*}(\mathbbm{1}_{CD}))^t-\mathbbm{1}_{AB}\right|^{-\frac{1}{2}}+\Pi_0]\otimes \mathbbm{1}_{CD}\right). \label{smat}
\end{equation}
Now we will show that this condition is also compatible with the trace condition on $\sigma_0$~($\tr \sigma_0 = 1$). By taking the trace with respect to the second subsystem in the above equation, we obtain
\begin{equation}
     \left(\Pi_{+}-\Pi_{-}\right)  (\tr{\sigma_0})\leq \left(\left|(\mathsf{T}^{*}(\mathbbm{1}_{CD}))^t-\mathbbm{1}_{AB}\right|^{-\frac{1}{2}}+\Pi_0\right) (\mathsf{T}^{*}(\mathbbm{1}_{CD}))^t \left(\left|(\mathsf{T}^{*}(\mathbbm{1}_{CD}))^t-\mathbbm{1}_{AB}\right|^{-\frac{1}{2}}+\Pi_0\right). \label{smap}
\end{equation}
Let $\{\mu_k\}_k$ be the spectrum of $(\mathsf{T}^{*}(\mathbbm{1}_{CD}))^t$, then the spectral decomposition of $(\mathsf{T}^{*}(\mathbbm{1}_{CD}))^t$ gives $(\mathsf{T}^{*}(\mathbbm{1}_{CD}))^t = \sum_i \mu_i \Pi_i$. Then we have 
\begin{align*}
    \left|(\mathsf{T}^{*}(\mathbbm{1}_{CD}))^t-\mathbbm{1}_{AB}\right|^{-\frac{1}{2}} = \sum_{\substack{i: \\ \mu_i \neq 1}} \left|\mu_i - 1\right|^{-\frac{1}{2}}\Pi_i,
\end{align*}
which implies
\begin{align*}
    \left(\left[(\mathsf{T}^{*}(\mathbbm{1}_{CD}))^t-\mathbbm{1}_{AB}\right]^{-\frac{1}{2}}+\Pi_0\right) (\mathsf{T}^{*}(\mathbbm{1}_{CD}))^t \left(\left[(\mathsf{T}^{*}(\mathbbm{1}_{CD}))^t-\mathbbm{1}_{AB}\right]^{-\frac{1}{2}}+\Pi_0\right) = \sum_{\substack{i: \\ \mu_i \neq 1}} \frac{\mu_i}{\left|\mu_i - 1\right|} \Pi_i+\Pi_0.
\end{align*}
Note that $\mu_i$ is positive but can be less than one while the operator $\left(\Pi_{+}-\Pi_{-}\right)$ may be nonpositive, then there is no clash of trace condition on $\sigma_0$, i.e., $\tr(\sigma_0) = 1$, with inequality in Eq. \eqref{smap}.
\\
In the special case where $\mathsf{T}$ is a trace increasing map, i.e., $\tr(\mathsf{X})<\mathrm{tr}(\mathsf{T}(\mathsf{X}))$ for all $X\in \mathcal{L}\left(A \otimes {B}\right)$, or equivalently, $\mathsf{T}^*(\mathbbm{1}_{CD}) - \mathbbm{1}_{AB}$ is a strictly positive operator, we have $\Pi_+=\mathbbm{1}_{AB}$ and $\Pi_-=0$. Now $\Cwhat_{\mathsf{T}'}$ is positive semidefinite  if and only if the operator $\sigma_{0}\in \mathcal{L}(C\otimes D)$ satisfies
\begin{equation}
    \mathbbm{1}_{AB} \otimes \sigma_0 \leq \left(\left[(\mathsf{T}^{*}(\mathbbm{1}_{CD}))^t-\mathbbm{1}_{AB}\right]^{-\frac{1}{2}}\otimes \mathbbm{1}_{CD}\right) \Cwhat_{\mathsf{T}} \left(\left[(\mathsf{T}^{*}(\mathbbm{1}_{CD}))^t-\mathbbm{1}_{AB}\right]^{-\frac{1}{2}}\otimes \mathbbm{1}_{CD}\right). \label{smat2}
\end{equation}
Taking partial trace with respect to the subsystems $C$ and $D$ in above equation, we get
\begin{align}
    \mathbbm{1}_{AB}\tr(\sigma_0)\leq \left(\left[(\mathsf{T}^{*}(\mathbbm{1}_{CD}))^t-\mathbbm{1}_{AB}\right]^{-\frac{1}{2}}\right) (\mathsf{T}^{*}(\mathbbm{1}_{CD}))^t \left(\left[(\mathsf{T}^{*}(\mathbbm{1}_{CD}))^t-\mathbbm{1}_{AB}\right]^{-\frac{1}{2}}\right) = \sum_i \frac{\mu_i}{\mu_i - 1} \Pi_i> \mathbbm{1}_{AB}.
\end{align}
Thus, the trace condition on $\sigma_0$, i.e., $\tr(\sigma_0) = 1$, does not create any contradiction with Eq.~\eqref{smat2}. 
\\
We now provide two examples of how to construct $\sigma_0$: one for a given trace increasing map and another for a map that is neither trace increasing nor trace decreasing.
\\
{\bf Example $1$:}--Let us define a map $\mathsf{T}:\mathcal{L}\left(A \otimes {B}\right) \rightarrow \mathcal{L}\left({C}\otimes {D}\right)$ as follows 
\begin{equation}
    \mathsf{T}(\mathsf{X}) = \alpha \tr(\mathsf{X}) \mathbbm{1}_{CD},
\end{equation}
 where $\alpha$ is real. We are not assuming whether $\mathsf{T}$ is trace increasing or decreasing yet. The above equation also implies $ \mathsf{T}^{*}(\mathbbm{1}_{CD}) = \alpha n \mathbbm{1}_{AB}  $, where $n := \tr(\mathbbm{1}_{CD})$ and $\mathsf{T}^{*}$ is adjoint of $\mathsf{T}$. We now have 
 \begin{equation}
     \left(\left[(\mathsf{T}^{*}(\mathbbm{1}_{CD}))^t-\mathbbm{1}_{AB}\right]^{-\frac{1}{2}}\otimes \mathbbm{1}_{CD}\right) \Cwhat_{\mathsf{T}} \left(\left[(\mathsf{T}^{*}(\mathbbm{1}_{CD}))^t-\mathbbm{1}_{AB}\right]^{-\frac{1}{2}}\otimes \mathbbm{1}_{CD}\right) = \left(\frac{\alpha}{n \alpha -1}\right)\mathbbm{1}_{AB} \otimes \mathbbm{1}_{CD}.
 \end{equation}
 Notice that if we choose $\alpha>\frac{1}{n}$, then $\mathsf{T}$ becomes trace increasing map. Using Eq.~\eqref{smat2} we obtain the following 
 inequality
 \begin{equation}
     \mathbbm{1}_{AB} \otimes\sigma_0 \leq \left(\frac{\alpha}{n \alpha -1}\right)\mathbbm{1}_{AB}\otimes \mathbbm{1}_{CD}.
 \end{equation}
By a suitable choice of $\alpha$, we can make the number $\left(\frac{\alpha}{n \alpha -1}\right)$ as large as we want. Specifically, setting $\alpha=\frac{1}{n}+ \frac{1}{n^2 L}$, where $L$ is a very large positive real number, gives $\left(\frac{\alpha}{n \alpha -1}\right)=L+\frac{1}{n}$.
 If we choose $\sigma_0$ to a unit trace and diagonal matrix with the diagonal elements being less than $L+ \frac{1}{n}$ (some of the diagonal entries can be chosen to be negative also), then the above inequality holds. Thus with this choice of $\sigma_0$, the CP and trace increasing map $ \mathsf{T}(\mathsf{X}) = \alpha \tr(\mathsf{X}) \mathbbm{1}_{CD}$ can be transformed into a CP and trace-preserving map, using Eq.~\eqref{b11}.
 \\
 {\bf Example $2$:}--
Let us define a map $\mathsf{T}$, whose input and output systems are qubits, by the Kraus operators $K_1 = \alpha \ketbra{0}{0}$ and $K_2 = \beta \ketbra{1}{1}$, where $\{\ket{0}, \ket{1}\}$ forms the computational basis and $\alpha>1$, $0<\beta <1$ are positive real constants. If we take the input of $\mathsf{T}$ as $\mathsf{X} = \ketbra{0}{0}$, we then have $\tr(\mathsf{X})<\mathrm{tr}(\mathsf{T}(\mathsf{X}))$. On the other hand, if we take the input of $\mathsf{T}$ as $\mathsf{X} = \ketbra{1}{1}$, we then have $\tr(\mathsf{X})>\mathrm{tr}(\mathsf{T}(\mathsf{X}))$. Clearly, the map $\mathsf{T}$ is neither trace increasing nor trace decreasing. The Choi operator of $\mathsf{T}$ is given by 
\begin{equation}
    \Cwhat_{\mathsf{T}}  = \alpha^2 \ketbra{00}{00} + \beta^2 \ketbra{11}{11}.
\end{equation}
Now we calculate the right-hand side of Eq.~\eqref{smat}. First we need the following term
\begin{align}
    \vert \mathsf{T}^{*}(\mathbbm{1})^{t} - \mathbbm{1}\vert ^{-\frac{1}{2}} \otimes \mathbbm{1} = \frac{1}{\sqrt{(\alpha^2-1)}} \left( \ketbra{00}{00} + \ketbra{01}{01} \right)+  \frac{1}{\sqrt{(\beta^2-1)}} \left( \ketbra{10}{10} + \ketbra{11}{11}\right).
\end{align}
We then have
\begin{align}
    \mathsf{A}:= \left(\left|(\mathsf{T}^{*}(\mathbbm{1}))^t-\mathbbm{1}\right|^{-\frac{1}{2}}\otimes \mathbbm{1}\right) \Cwhat_{\mathsf{T}} \left(\left|(\mathsf{T}^{*}(\mathbbm{1}))^t-\mathbbm{1}\right|^{-\frac{1}{2}}\otimes \mathbbm{1}\right) = \alpha^2 (\alpha^2 - 1)^{-1} \ketbra{00}{00} + \beta^2 (\beta^2 - 1)^{-1} \ketbra{11}{11}.
\end{align}
We can write $\mathsf{T}^{*}(\mathbbm{1})^t-\mathbbm{1} = M - N$, where $M = (\alpha^2-1)\ketbra{0}{0}$ and $N = (1-\beta^2)\ketbra{1}{1}$. This also implies $\Pi_{M} = \ketbra{0}{0} $ and $\Pi_{N} = \ketbra{1}{1}$. We seek to find $\sigma_0$ such that $\tr(\sigma_0) = 1$ and $(\Pi_M - \Pi_N)\otimes \sigma_0 \leq \mathsf{A}$. Let us choose a diagonal matrix $\sigma$ with eigenvalues $\lambda_1$ and $\lambda_2$ such that $\lambda_1 < \alpha^2(\alpha^2-1)^{-1}$, $ \left|\lambda_2 \right| < \beta^2(\beta^2-1)^{-1}$, $\lambda_2<0$ and $\lambda_1+\lambda_2 > 1$. Now if we take $\alpha = \sqrt{2}$, $\beta = \frac{1}{\sqrt{2}}$ and choose $\lambda_1 = 1.99$, $\lambda_2 = -\frac{1}{5}$ then all the above-mentioned condition are satisfied. Now we choose $\sigma_0 = \frac{\sigma}{\tr(\sigma)}$. With this choice, we have $\tr(\sigma_0) = 1$ and $(\Pi_M - \Pi_N)\otimes \sigma_0 \leq \mathsf{A}$. This example can trivially be generalized to incorporate bipartite systems suitable in the context of representing maps.

\subsection{Proof of Lemma~\ref{lemma_4} }\label{sec:detailed_proof_lemma_4}

Consider the representing map $\mathsf{T}: \mathcal{L}\left(A\otimes B\right) \rightarrow \mathcal{L}\left(C\otimes D\right)$ of a supermap $\Theta:\mathcal{L}\left(A,B\right)\rightarrow \mathcal{L}\left(C,D\right)$  and let $\Theta^{*}$ be the adjoint of $\Theta$. From the definition of adjoint, for all $\mathsf{X} \in \mathcal{L}\left(A\otimes B\right)$ and $\mathsf{Y} \in \mathcal{L}\left(C\otimes D\right)$, we have
\begin{align}
    \inner{\mathsf{T}^{*}(\mathsf{Y}),\mathsf{X}} &= \inner {\mathsf{Y},T(\mathsf{X})} = \inner{\Cwhat_{\Gamma_{\mathsf{Y}}},\Cwhat_{\Theta(\Gamma_{\mathsf{X}})}} = \mathrm{tr}\left[\Cwhat^{\dagger}_{\Gamma_{\mathsf{Y}}}\Cwhat_{\Theta\left(\Gamma_{\mathsf{X}}\right)}\right].
\end{align}
Let $\{e^{Q}_{ij}\}$ be a canonical basis for the space $\mathcal{L}\left(Q\right)$, where $ Q = \{A,B,C,D\}$. Then, the above equation leads to
\begin{align}
&\mathrm{tr}\left[\Cwhat^{\dagger}_{\Gamma_{\mathsf{Y}}}\Cwhat_{\Theta\left(\Gamma_{\mathsf{X}}\right)}\right]\nonumber\\
&=\sum_{ij}\sum_{kl} \mathrm{tr}\left[\left({e_{ij}^{C}}^{\dagger}\otimes \left\{\Gamma_{\mathsf{Y}}\left(e_{ij}^{C}\right)\right\}^{\dagger}\right)\left({e_{kl}^{C}} \otimes \Theta\left(\Gamma_{\mathsf{X}}\right)\left({e_{kl}^{C}}\right)\right)\right]\nonumber\\
&=\sum_{ij}\sum_{kl} \delta_{ik} \delta_{jl} \mathrm{tr}\left[\left\{\Gamma_{\mathsf{Y}}\left({e_{ij}^{C}}\right)\right\}^{\dagger}\Theta\left(\Gamma_{\mathsf{X}}\right)\left({e_{kl}^{C}}\right)\right] =\sum_{ij}\mathrm{tr}\left[\left\{\Gamma_{Y}\left({e_{ij}^{C}}\right)\right\}^{\dagger}\Theta\left(\Gamma_{\mathsf{X}}\right)\left({e_{ij}^{C}}\right)\right] =\inner{\Gamma_{\mathsf{Y}},\Theta(\Gamma_{\mathsf{X}})}.
\end{align}
Using the definition of the adjoint of a supermap $\Theta$, we obtain
\begin{align}
\nonumber
  \inner{\Gamma_{\mathsf{Y}},\Theta(\Gamma_{\mathsf{X}})}  =\inner{\Theta^{*}(\Gamma_{\mathsf{Y}}),\Gamma_{\mathsf{X}}}
&=\sum_{i,j}\inner{\Theta^{*}\left(\Gamma_{\mathsf{Y}}\right)\left({e_{ij}^{A}}\right),\Gamma_{\mathsf{X}}\left({e_{ij}^{A}}\right)}\\
\nonumber
    &=\sum_{i,j} \sum_{k,l} \delta_{ik} 
    \delta_{jl} \inner{\Theta^{*}\left(\Gamma_{\mathsf{Y}}\right)\left({e_{ij}^{A}}\right),\Gamma_{\mathsf{X}}\left({e_{kl}^{A}}\right)}\\
\nonumber    
    &=\sum_{i,j} \sum_{k,l} \inner{ {e_{ij}^{A}}, {e_{kl}^{A}}} \inner{\Theta^{*}\left(\Gamma_{\mathsf{Y}}\right)\left({e_{ij}^{A}}\right),\Gamma_{\mathsf{X}}\left({e_{kl}^{A}}\right)}\\
\nonumber    
    &= \inner{\sum_{i,j} {e_{ij}^{A}}\otimes \Theta^{*}\left(\Gamma_{\mathsf{Y}}\right)\left({e_{ij}^{A}}\right), \sum_{k,l} {e_{kl}^{A}} \otimes \Gamma_{\mathsf{X}}\left({e_{kl}^{A}}\right)}= \inner{\Cwhat_{\Theta^{*}(\Gamma_{\mathsf{Y}})},\mathsf{X}}.  
\end{align}
As one can choose $\mathsf{X}$ arbitrarily, it can be concluded that $\mathsf{T}^{*}(\mathsf{Y})=\mathsf{C}_{\Theta^{*}(\Gamma_{\mathsf{Y}})}$.

 \section{Proofs of statements in Section~\ref{Entropy_of_quantum_channels}}\label{proof_of_binod}
 
\subsection{Proof of monotonicity of relative entropy under quantum superchannel }\label{monotonicity_of_channeldivergence_proof}

 Let $\mathcal{N}_{A' \rightarrow A}$ and $\mathcal{M}_{A' \rightarrow A}$ be a quantum channel and a CP map, respectively, and let $\Theta:\mathcal{L}\left(A',A\right)\rightarrow \mathcal{L}\left(C',C\right)$ be a superchannel. To prove that the proposed inequality holds, we start with the left-hand side, i.e., $ D\left(\Theta\left(\mathcal{N}\right)\Vert \Theta \left(\mathcal{M}\right)\right)$. Now, considering a general superchannel $\Theta$, whose action on any CP map is defined in Eq.\eqref{equ:superchannel_action}, we can write
  \begin{align}
     & \chD*{\Theta\left(\mathcal{N}\right)}{\Theta \left(\mathcal{M}\right)} \nonumber\\
     & = \chD*{V'_{AE \rightarrow C}\hspace{0.05cm}\circ\left(\mathcal{N}_{A'\rightarrow A} \otimes \mathrm{id}_{E}\right) \hspace{0.05cm} \circ U_{C' \rightarrow A'E}}{V'_{AE \rightarrow C}\hspace{0.05cm}\circ\left(\mathcal{M}_{A'\rightarrow A} \otimes \mathrm{id}_{E}\right) \hspace{0.05cm} \circ U_{C' \rightarrow A'E}} \nonumber\\
     & =  \underset{\Psi_{C'B}}{\mathrm{sup}} \staD*{\left( V'_{AE \rightarrow C}\hspace{0.05cm}\circ\left(\mathcal{N}_{A'\rightarrow A} \otimes \mathrm{id}_{E}\right) \hspace{0.05cm} \circ U_{C' \rightarrow A'E} \otimes \mathrm{id}_{B}\right) \Psi_{C'B} }{\left(V'_{AE \rightarrow C}\hspace{0.05cm}\circ\left(\mathcal{M}_{A'\rightarrow A} \otimes \mathrm{id}_{E}\right) \hspace{0.05cm} \circ U_{C' \rightarrow A'E} \otimes \mathrm{id}_{B}\right) \Psi_{C'B}} \nonumber \\
     & \leq \underset{\Psi_{A'E B}}{\mathrm{sup}} \staD*{\left(V'_{AE \rightarrow C}\hspace{0.05cm}\circ\left(\mathcal{N}_{A'\rightarrow A} \otimes \mathrm{id}_{E}\right)  \otimes \mathrm{id}_{B}\right) \Psi_{A'EB} }{\left(V'_{AE \rightarrow C}\hspace{0.05cm}\circ\left(\mathcal{M}_{A'\rightarrow A} \otimes \mathrm{id}_{E}\right) \otimes \mathrm{id}_{B}\right) \Psi_{A'EB}},\nonumber
  \end{align}
  where the inequality in the third line follows from the fact that the set of pure states obtained by varying $\Psi_{C'B}$ and applying $U_{C' \rightarrow A'E}$ is smaller in comparison to the set of all pure states in the system $A'EB$, i.e, $\{\Psi_{A'E B}\}$. In other words, supremum of a set is always greater than or equal to supremum over its subsets. Let us define $\left(\left(\mathcal{N}_{A'\rightarrow A}\otimes \mathrm{id}_{E}\right)\otimes \mathrm{id}_{B}\right)\Psi_{A'EB} := \rho^{\Psi}_{ABE}$, which is also a density operator since $\mathcal{N}_{A' \rightarrow A}$ is a quantum channel, and $\left(\left(\mathcal{M}_{A'\rightarrow A}\otimes \mathrm{id}_{E}\right)\otimes \mathrm{id}_{B}\right)\Psi_{A'EB} := \sigma^{\Psi}_{ABE}$, which is a positive operator, because map $\mathcal{M}$ is CP but not trace-preserving. Now, we have 
   \begin{align}
       & \underset{\Psi_{A'E B}}{\mathrm{sup}} \staD*{\left(V'_{AE \rightarrow C}\hspace{0.05cm}\circ\left(\mathcal{N}_{A'\rightarrow A} \otimes \mathrm{id}_{E}\right)  \otimes \mathrm{id}_{B}\right) \Psi_{A'EB} }{\left(V'_{AE \rightarrow C}\hspace{0.05cm}\circ\left(\mathcal{M}_{A'\rightarrow A} \otimes \mathrm{id}_{E}\right) \otimes \mathrm{id}_{B}\right) \Psi_{A'EB}}\nonumber \\
       & =\underset{\Psi_{A'E B}}{\mathrm{sup}} \staD*{\left( V'_{AE \rightarrow C} \otimes \mathrm{id}_{B}\right) \rho^{\Psi}_{ABE}}{\left( V'_{AE \rightarrow C}\otimes \mathrm{id}_{B}\right) \sigma^{\Psi}_{AEB}}\nonumber \\
       & \leq \underset{\Psi_{A'E B}}{\mathrm{sup}} \staD*{\rho^{\Psi}_{ABE}}{\sigma^{\Psi}_{AEB}}\nonumber \\
       & = \underset{\Psi_{A'E B}}{\mathrm{sup}} \staD*{\left(\left(\mathcal{N}_{A'\rightarrow A}\otimes \mathrm{id}_{E}\right)\otimes \mathrm{id}_{B}\right)\Psi_{A'EB}}{\left(\left(\mathcal{M}_{A'\rightarrow A}\otimes \mathrm{id}_{E}\right)\otimes \mathrm{id}_{B}\right)\Psi_{A'EB}} = \chD*{\mathcal{N}}{\mathcal{M}},
   \end{align}
where the third line follows from the monotonicity of relative entropy under a quantum channel, and the last line simply follows from the definition of relative entropy between two completely positive maps.  Thus, we have the inequality
\begin{equation}
    D\left[\Theta\left(\mathcal{N}\right) \Vert \Theta \left(\mathcal{M}\right)\right] \leq D\left[\mathcal{N}\Vert \mathcal{M}\right].
\end{equation}

\subsection{Proof of monotonicity of relative entropy under CP and TP preserving supermaps} \label{equ:proof_of_prop2}

Since $\mathcal{N}_{A \rightarrow A'}$ and $\mathcal{M}_{ A \rightarrow A'}$ are quantum channels, and the supermap $\Theta:\mathcal{L}\left(A,A'\right)\rightarrow \mathcal{L}\left(C,C'\right)$ is CP- and TP-preserving, we have that $\Theta(\mathcal{N})$ and $\Theta(\mathcal{M})$ are also quantum channels. Here, the quantum channels $\mathcal{M}, \mathcal{N}$ and supermap $\Theta$, are such that the states $\Psi$ and $\Phi$ that realize the supremum of relative entropy functionals $\chD*{\mathcal{N}}{\mathcal{M}}$ and $\chD*{\Theta(\mathcal{N})}{\Theta(\mathcal{M})}$, are the elements of $\mathsf{FRank}\left(\widetilde{A}\otimes{A}\right)\cap \mathcal{D}\left(\widetilde{A}\otimes{A}\right)$ and $\mathsf{FRank}\left(\widetilde{C}\otimes{C}\right)\cap \mathcal{D}\left(\widetilde{C}\otimes{C}\right)$, respectively. We then have $ \chD*{\mathcal{N}}{\mathcal{M}}= \staD*{{\mathsf{C}}_{\mathcal{N}}^{\Psi}}{{\mathsf{C}}_{\mathcal{M}}^{\Psi}}$ and $ \chD*{\Theta(\mathcal{N})}{\mathcal{M}(\Theta)} = \staD*{{\mathsf{C}}_{\Theta(\mathcal{N})}^{\Phi}}{{\mathsf{C}}_{\Theta(\mathcal{M})}^{\Phi}}$, where $\Psi \in \mathsf{FRank}\left(\widetilde{A}\otimes{A}\right)\cap \mathcal{D}\left(\widetilde{A}\otimes{A}\right)$ and $\Phi \in \mathsf{FRank}\left(\widetilde{C}\otimes{C}\right)\cap \mathcal{D}\left(\widetilde{C}\otimes{C}\right)$~(see Appendix \ref{detailde_proof_of_recovery} for detailed calculations). Further,  $\mathsf{C}_{\Theta\left(\mathcal{N}\right)}^{\Phi}$ and $\mathsf{C}_{\Theta\left(\mathcal{M}\right)}^{\Phi}$ are Choi states of channels $\Theta \left(\mathcal{N}\right)$ and $\Theta \left(\mathcal{N}\right)$, respectively and $\mathsf{C}_{\Theta\left(\mathcal{N}\right)}^{\Phi}$ and $\mathsf{C}_{\Theta\left(\mathcal{M}\right)}^{\Phi}$ are Choi states of channels $\Theta \left(\mathcal{N}\right)$ and $\Theta \left(\mathcal{M}\right)$, respectively. Note that the Choi states $\mathsf{C}_{\mathcal{N}}^{\Psi}$, $\mathsf{C}_{\Theta(\mathcal{N})}^{\Phi}$ and  $\mathsf{C}_{\mathcal{M}}^{\Psi}$, $\mathsf{C}_{\Theta(\mathcal{M})}^{\Phi}$ are connected via the representing map $\mathfrak{T}:\mathcal{L}\left(A \otimes A'\right)\rightarrow \mathcal{L}\left(C \otimes C'\right)$ of the supermap $\Theta$ as $\mathsf{C}_{\Theta \left(\mathcal{N}\right)}^{\Phi} =  \mathfrak{T} \left(\mathsf{C}_{\mathcal{N}}^{\Psi}\right)$ and $\mathsf{C}_{\Theta \left(\mathcal{M}\right)}^{\Phi} =  \mathfrak{T} \left(\mathsf{C}_{\mathcal{M}}^{\Psi}\right)$.  As the supermap $\Theta$ is CP-preserving, the representing map $\mathfrak{T}$ is a positive map. However,  $\mathfrak{T}$ may not be a trace-preserving map in general. In the cases where it is trace non-increasing, we can always replace it with the following positive and trace-preserving map: $\mathfrak{T}'(\mathsf{X}):= \mathfrak{T}(\mathsf{X})+\left[\mathrm{tr}(\mathsf{X})-\mathrm{tr}(\mathfrak{T}(\mathsf{X}))\right] \sigma_{0}$, where $\sigma_{0} \in \mathcal{D}(C \otimes C')$ is some fixed state and $\mathsf{X}\in \mathcal{L}(A \otimes A')$. Following Remark~(\ref{rem:choi-state-choi-operator-relation_1}), we have $\mathfrak{T}'(\mathsf{C}^{\Psi}_{\mathcal{N}})=\mathfrak{T}(\mathsf{C}^{\Psi}_{\mathcal{N}})$ and $\mathfrak{T}'(\mathsf{C}^{\Psi}_{\mathcal{M}})=\mathfrak{T}(\mathsf{C}^{\Psi}_{\mathcal{M}})$. Then \cite[Theorem 1]{Hermes2017} leads to 
\begin{align}
 \staD*{{\mathsf{C}}_{\mathcal{N}}^{\Psi}}{{\mathsf{C}}_{\mathcal{M}}^{\Psi}} \geq  \staD*{\mathfrak{T}'\left({\mathsf{C}}_{\mathcal{N}}^{\Psi}\right)}{\mathfrak{T}'\left( \mathsf{C}_{\mathcal{M}}^{\Psi}\right)},
\end{align}
which implies
\begin{equation}
   \chD*{\mathcal{N}}{\mathcal{M}} \geq \chD*{\Theta\left(\mathcal{N}\right)}{\Theta \left(\mathcal{M}\right)}.
\end{equation}
If we relax the TP-preserving condition on the supermap $\Theta$ and take it  only to be CP-preserving, then the representing map $\mathfrak{T}$ will be a positive map, but it may not always be possible to define a trace-preserving map $\mathfrak{T}'$ such that $\mathfrak{T}'(\mathsf{C}^{\Psi}_{\mathcal{N}})=\mathfrak{T}(\mathsf{C}^{\Psi}_{\mathcal{N}})$ holds for all quantum channels $\mathcal{N}$. In this situation, we provide a condition on the supermap $\Theta$ such that its representing map $\mathfrak{T}$ becomes trace-preserving. The representing map $\mathsf{T}$ is trace-preserving if and only if $\Theta^{*}$ is completely depolarising map preserving~(see Lemma~\ref{lemma5}). A similar condition for the representing map  $\mathfrak{T}$ is given as follows: the map $\mathfrak{T}$ is trace-preserving if and only if its adjoint is unital, i.e.,  $\mathfrak{T}^{*}\left(\mathbbm{1}\right) = \mathbbm{1}$, which can also be written as $\mathfrak{T}^{*}\left( \mathsf{C}^{\Phi}_{\Theta^{-1}_{\Phi}(\mathcal{R}_{C \rightarrow C'})} \right) = \mathsf{C}^{\Psi}_{\Theta^{-1}_{\Psi}(\mathcal{R}_{A \rightarrow A'})} $. Using Eq.~(\ref{equ:theta_caps}), we obtain the following condition on the supermap $\Theta$, which implies that $\mathfrak{T}$ is trace-preserving:
\begin{equation}
    \Theta^{*} \circ \Theta^{*}_{\Phi}\left(\mathcal{R}_{C \rightarrow C'}\right)= \Theta^{*}_{\Psi} \left(\mathcal{R}_{A \rightarrow A'}\right). \label{equ:tp_condition}
\end{equation}
If $\mathcal{N}, \mathcal{M}$  and $\Theta$ are such that $\Psi$ and $\Phi$ are maximally entangled states, then $\Theta^{*}_{\Psi}$ and $\Theta^{*}_{\Phi}$ become the identity supermap~(see the definition of $\Theta_{\Psi}, \Theta_{\Phi}$ in  Eq.~(\ref{equ: HRI})), and $\Theta^{*}$ becomes $\mathcal{R}$-preserving supermap.

\subsection{Proof of Proposition~\ref{prop3}}\label{Detailed_proof_of_Proposition4}

    From our definition of relative entropy between a quantum channel and a positive map, we have 
    \begin{align}
        & \chD*{\mathcal{N}\otimes \widetilde{\mathcal{N}}}{\mathcal{M}\otimes \widetilde{\mathcal{M}}}\nonumber\\
        &= \underset{\Psi_{RAB}}{\mathrm{sup}}~\staD*{\left(\mathrm{id}_{R}\otimes\mathcal{N}\otimes \widetilde{N}\right) (\Psi_{RAB}) }{\left( \mathrm{id}_{R}\otimes\mathcal{M}\otimes\widetilde{\mathcal{M}}\right)(\Psi_{RAB})}\nonumber\\
        & \geq \underset{\Psi_{R_{1}A}\otimes \Psi_{R_{2}B}}{\mathrm{sup}}~ 
 \staD*{\left(\mathrm{id}_{R_{1}}\otimes\mathcal{N}\right)\Psi_{R_{1}A}\otimes\left(\mathrm{id}_{R_{2}} \otimes\widetilde{N}\right)\Psi_{R_{2}B}}{\left(\mathrm{id}_{R_{1}}\otimes\mathcal{M}\right)\Psi_{R_{1}A}\otimes\left(\mathrm{id}_{R_{2}} \otimes\widetilde{M}\right)\Psi_{R_{2}B}}\nonumber\\
         &= \underset{\Psi_{R_{1}A}\otimes \Psi_{R_{2}B}}{\mathrm{sup}}~\left(\staD*{\left(\mathrm{id}_{R_{1}}\otimes\mathcal{N}\right)\Psi_{R_{1}A}}{\left(\mathrm{id}_{R}\otimes\mathcal{M}\right)\Psi_{R_{1}A}}+ \staD*{\left(\mathrm{id}_{R_{2}}\otimes\widetilde{\mathcal{N}}\right)\Psi_{R_{2}B}}{\left(\mathrm{id}_{R_{2}}\otimes\widetilde{\mathcal{M}}\right)\Psi_{R_{2}B}}\right)\nonumber\\
         & = \chD*{\mathcal{N}}{\mathcal{M}}+\chD{\widetilde{\mathcal{N}}}{\widetilde{\mathcal{M}}},
    \end{align}
    where in the second line, we have used the  fact that  supremum of a function over a larger set is always greater than its  supremum over a smaller subset, and in the third line we have used the additivity of quantum relative entropy, i.e., for positive semidefinite operators $\rho_{1},\rho_{2},\sigma_{1}$, and $\sigma_{2}$, we have $D\left(\rho_{1}\otimes\rho_{2}\Vert \sigma_{1}\otimes\sigma_{2}\right) = \mathrm{tr}\left(\rho_{2}\right)D\left(\rho_{1}\Vert \sigma_{1}\right)+\mathrm{tr}\left(\rho_{1}\right)D\left(\rho_{2}\Vert \sigma_{2}\right)$, and the last line follows simply from the definition of relative entropy between a quantum channel and a CP map. 

\subsection{Axioms for state and channel entropy}\label{axioms}
We can write von Neumann entropy in terms of relative entropy as follows
\begin{equation}
     S(A)_{\rho} = -D\left(\rho_{A} \Vert \mathbbm{1}_A\right).
\end{equation}
The above definition of entropy of a quantum state satisfies some physically motivated axioms:
\begin{enumerate}
   \item Under the action of random unitary channels, it increases: let $\mathcal{U}(\rho)= \sum_{i} p_{i} U_i \rho U_{i}^{\dagger}$ be a random unitary channel with $U_i$'s being unitary operators. Then $S(\mathcal{U}(\rho)) \geq S(\rho)$ holds. 
   \item It is additive under tensor product of quantum states: for $\rho \in \mathcal{D}(A)$ and $\sigma \in \mathcal{D}(A')$, we have $S(\rho \otimes \sigma) =S(\rho)+ S(\sigma)$.
   \item It is maximum for the maximally mixed state and zero for all pure states.
\end{enumerate}
It was advocated in \cite{Gour_entropy} that the entropy functional for a quantum channel should satisfy the following physically motivated axioms:
\begin{enumerate}
    \item It should increase under the action of random unitary superchannels.
    \item It should be additive under tensor product of quantum channels.
    \item It should be maximum for the completely depolarizing channel and should be zero for a replacement channel that outputs a pure state.
\end{enumerate}

\subsection{Proof of additivity of channel entropy }\label{prop2_channel_entropy_addition_proof}
 Let $\mathcal{N}_{A_{1} \rightarrow B_1}$ and $\mathcal{M}_{A_{2} \rightarrow B_2}$ be two quantum channels, and $\mathcal{R}_{A_{1} \rightarrow B_1}$, $\mathcal{R}_{A_{2} \rightarrow B_2}$  be the corresponding completely depolarising maps. It directly follows that additivity holds if the following equality holds.
   \begin{align*}
     \chD*{\mathcal{N}_{A_1\rightarrow B_1}\otimes {\mathcal{M}}_{A_2 \rightarrow B_2}}{\mathcal{R}_{A_1\rightarrow B_1}\otimes {\mathcal{R}}_{A_2 \rightarrow B_2}}  = \chD*{\mathcal{N}_{A_1 \rightarrow B_1}}{\mathcal{R}_{A_1\rightarrow B_1}} + \chD*{\mathcal{M}_{A_2 \rightarrow B_2}}{\mathcal{R}_{A_2\rightarrow B_2}}.
   \end{align*}
   To show that the above equality holds, it is sufficient to prove that the quantity on the left-hand side is less than or equal to the quantity on the right-hand side and then the other way around is also true. If both inequalities are true, then, both sides must be equal. First, we will prove the $``\geq"$ part, i.e.,  
    \begin{align*}
       \chD*{\mathcal{N}_{A_1\rightarrow B_1}\otimes {\mathcal{M}}_{A_2 \rightarrow B_2}}{\mathcal{R}_{A_1\rightarrow B_1}\otimes {\mathcal{R}}_{A_2 \rightarrow B_2}} \geq \chD*{\mathcal{N}_{A_1 \rightarrow B_1}}{\mathcal{R}_{A_1\rightarrow B_1}} + \chD*{\mathcal{M}_{A_2 \rightarrow B_2}}{\mathcal{R}_{A_2\rightarrow B_2}}.
    \end{align*}
Let us define the action of channel $\mathcal{M}_{A_2 \rightarrow B_2}$ and the CP map $\mathcal{R}_{A_2 \rightarrow B_2}$ on an arbitrary pure state $\Psi_{R A_{1} A_{2}}$ as $\mathcal{M}_{A_2 \rightarrow B_2} \left( \Psi_{R A_{1} A_{2}} \right) := \rho_{RA_1B_2}$, and $ \mathcal{R}_{A_2 \rightarrow B_2} \left( \Psi_{R A_{1} A_{2}} \right) := \sigma_{RA_1B_2}$,
  where $\rho_{RA_1B_2}$ is a density operator and $\sigma_{RA_1B_2}$ is a positive operator. We then have 

   \begin{align}
    & \chD*{\mathcal{N}_{A_1\rightarrow B_1}\otimes {\mathcal{M}}_{A_2 \rightarrow B_2}}{\mathcal{R}_{A_1\rightarrow B_1}\otimes {\mathcal{R}}_{A_2 \rightarrow B_2}} \nonumber\\
    & = \sup_{\Psi_{R A_1 A_2}} \staD*{\left(\mathrm{id}_{R} \otimes \mathcal{N}_{A_1 \rightarrow B_1}\otimes\mathcal{M}_{A_2\rightarrow B_2}\right)\Psi_{R A_1 A_2}}{\left(\mathrm{id}_{R}\otimes \mathcal{R}_{A_1 \rightarrow B_1}\otimes\mathcal{R}_{A_2\rightarrow B_2}\right)\Psi_{RA_1A_2}} \nonumber\\
    & = \sup_{\Psi_{R A_1 A_2}} \staD*{\left(\mathrm{id}_{R} \otimes \mathcal{N}_{A_1 \rightarrow B_1}\otimes\mathrm{id}_{B_2}\right)\rho_{R A_1 B_2}}{\left(\mathrm{id}_{R}\otimes \mathcal{R}_{A_1 \rightarrow B_1}\otimes\mathrm{id}_{ B_2}\right)\sigma_{RA_1B_2}}\nonumber\\
    & = \sup_{\Psi_{R' A_1}} \staD*{\left(\mathrm{id}_{R'} \otimes \mathcal{N}_{A_1 \rightarrow B_1}\right)\rho_{R' A_1 }}{\left(\mathrm{id}_{R'}\otimes \mathcal{R}_{A_1 \rightarrow B_1}\right)\sigma_{R'A_1}},
 \end{align} 
   where $R' = RB_2$ is relabeling of the system  $RB_{2}$. Now, using the definition of the completely depolarising map, the operator $\left(\mathrm{id}_{R'}\otimes \mathcal{R}_{A_{1}\rightarrow B_{1}}\right)\sigma_{R'A_{1}}$ can be simplified to  $\sigma_{R'}\otimes\mathbbm{1}_{B_{1}}$. Also, let us define $A_{R'B_{1}} := \left(\mathrm{id}_{R'}\otimes \mathcal{N}_{A_{1}\rightarrow B_{1}}\right)\rho_{R'A_{1}}$, then we observe $A_{R'} = \mathrm{tr}_{B_{1}}\left(A_{R'B_{1}}\right) = \left(\mathrm{id}_{R'}\otimes\mathrm{tr_{B_{1}}}\right)\left(\mathrm{id_{R'}}\otimes \mathcal{N}_{A_{1}\rightarrow B_{1}}\right) \rho_{R'A_{1}} = \left(\mathrm{id}_{R'}\otimes\mathrm{tr_{B_{1}}} \circ \mathcal{N}_{A_{1}\rightarrow B_{1}}\right) \rho_{R'A_{1}} = \rho_{R'} $. Thus, we can write
   \begin{align}
      \staD*{\left(\mathrm{id}_{R'} \otimes \mathcal{N}_{A_1 \rightarrow B_1}\right)\rho_{R' A_1 } }{\left(\mathrm{id}_{R'}\otimes \mathcal{R}_{A_1 \rightarrow B_1}\right)\sigma_{R'A_1}}  &=  -S(A_{R'B_{1}})-\mathrm{tr}\left(A_{R'}\log \sigma_{R'}\right),\\
       \staD*{\left(\mathrm{id}_{R'}\otimes \mathcal{N}_{A_{1}\rightarrow B_{1}}\right)\rho_{R'A_{1}}}{\left(\mathrm{id}_{R'}\otimes \mathcal{R}_{A_{1}\rightarrow B_{1}}\right)\rho_{R'A_{1}}} & = - S\left(A_{R'B_{1}}\right)-\mathrm{tr}\left(A_{R'}\log \rho_{R'}\right),
   \end{align}
   where $S\left(\rho \right)=-\tr \left( \rho \log \rho \right)$ is the von Neumann entropy of $\rho$.  
 Using the above equations, we see that
 \begin{align}
   & \staD*{\left(\mathrm{id}_{R'}\otimes \mathcal{N}_{A_{1}\rightarrow B_{1}}\right)\rho_{R'A_{1}}}{\left(\mathrm{id}_{R'}\otimes \mathcal{R}_{A_{1}\rightarrow B_{1}}\right)\sigma_{R'A}}\nonumber\\
   &= \staD*{\left(\mathrm{id}_{R'}\otimes \mathcal{N}_{A_{1}\rightarrow B_{1}}\right)\rho_{R'A_{1}}}{\left(\mathrm{id}_{R'}\otimes \mathcal{R}_{A_{1}\rightarrow B_{1}}\right)\rho_{R'A}} + \staD*{\rho_{R'}}{\sigma_{R'}}.\nonumber
 \end{align}
Defining
\begin{align}
    X&:= \staD*{\left(\mathrm{id}_{R'}\otimes \mathcal{N}_{A_{1}\rightarrow B_{1}}\right)\rho_{R'A_{1}}}{\left(\mathrm{id}_{R'}\otimes \mathcal{R}_{A_{1}\rightarrow B_{1}}\right)\sigma_{R'A}},\\
    Y&:= \staD*{\left(\mathrm{id}_{R'}\otimes \mathcal{N}_{A_{1}\rightarrow B_{1}}\right)\rho_{R'A_{1}}}{\left(\mathrm{id}_{R'}\otimes \mathcal{R}_{A_{1}\rightarrow B_{1}}\right)\rho_{R'A}},
\end{align}
we have $X = Y + \staD*{\rho_{R'}}{\sigma_{R'}}=Y + \staD*{\rho_{RB_2}}{\sigma_{RB_2}}$. Now exploiting the monotonicity of the relative entropy under partial trace, we arrive at 
\begin{align}
    X 
    & \leq Y + \staD*{\rho_{RB_{2}A_{1}}}{\sigma_{RB_{2}A_{1}}} \nonumber\\
    & = Y + \staD*{\left(\mathrm{id}_{R}\otimes \mathcal{M}_{A_2 \rightarrow B_2}\otimes \mathrm{id}_{A_1}\right)\Psi_{R A_{2}A_{1}}}{\left(\mathrm{id}_{R}\otimes \mathcal{R}_{A_2 \rightarrow B_2}\otimes \mathrm{id}_{A_1}\right)\Psi_{RA_{2}A_{1}}} \nonumber\\
    & = Y + \staD*{\left(\mathrm{id}_{R}\otimes \mathrm{id}_{A_1}\otimes \mathcal{M}_{A_2 \rightarrow B_2}\right)\Psi_{R A_{1}A_{2}}}{\left(\mathrm{id}_{R}\otimes \mathrm{id}_{A_1} \otimes \mathcal{R}_{A_2 \rightarrow B_2}\right)\Psi_{R A_{1}A_{2}}}.
\end{align}
 Taking supremum on both sides of the above inequality with respect to state $\Psi_{R A_1A_2}$, we obtain 
 \begin{align}
     \sup_{\Psi_{R A_{1}A_{2}}} X \leq \sup_{\Psi_{R A_{1}A_{2}}} Y + \chD*{\mathcal{M}_{A_2 \rightarrow B_2}}{\mathcal{R}_{A_2 \rightarrow B_2}},
 \end{align}
 which implies 
 \begin{align}
   \chD*{\mathcal{N}_{A_{1}\rightarrow B_{1}}\otimes \mathcal{M}_{A_{2}\rightarrow B_{2}}}{\mathcal{R}_{A_{1}\rightarrow B_{1}}\otimes \mathcal{R}_{A_{2}\rightarrow B_{2}}}  &\leq
     \chD*{\mathcal{N}_{A_1 \rightarrow B_1}}{\mathcal{R}_{A_1\rightarrow B_1}} + \chD*{\mathcal{M}_{A_2 \rightarrow B_2}}{\mathcal{R}_{A_2\rightarrow B_2}}.
 \end{align}
By multiplying with the minus sign on both sides of the above inequality and using the definition of the entropy of a quantum channel, we obtain the desired inequality 
 \begin{equation}
     S\left[\mathcal{N}_{A_{1}\rightarrow B_{1}}\otimes \mathcal{M}_{A_{2}\rightarrow B_{2}}\right] \geq S\left[\mathcal{N}_{A_{1}\rightarrow B_{1}}\right] + S\left[\mathcal{M}_{A_{2}\rightarrow B_{2}}\right]. \label{equ_additivity_of_entropy_1}
 \end{equation}
Now, we will prove the $``\leq"$ part, i.e., $S\left(\mathcal{N}_{A_{1}\rightarrow B_{1}}\otimes \mathcal{M}_{A_{2}\rightarrow B_{2}}\right) \leq S\left(\mathcal{N}_{A_{1}\rightarrow B_{1}}\right) + S\left(\mathcal{M}_{A_{2}\rightarrow B_{2}}\right)$. From the definition of relative entropy between a quantum channel and a CP map, we have
\begin{align}
    & \chD*{\mathcal{N}_{A_{1}\rightarrow B_{1}}\otimes \mathcal{M}_{A_{2}\rightarrow B_{2}}}{\mathcal{R}_{A_{1}\rightarrow B_{1}}\otimes \mathcal{R}_{A_{2}\rightarrow B_{2}}} \nonumber\\
     &= \sup_{\Psi_{EA_1 A_2}} \staD*{(\mathrm{id}_{E} \otimes {\mathcal{N}}_{A_1 \rightarrow B_1} \otimes {\mathcal{M}}_{A_2 \rightarrow B_2} ) \Psi_{EA_1 A_2}}{(\mathrm{id}_{E} \otimes {\mathcal{R}}_{A_1 \rightarrow B_1} \otimes {\mathcal{R}}_{A_2 \rightarrow B_2} ) \Psi_{{EA_1 A_2}}},
\end{align}
where we can take the ancilla $E$ to be isomorphic to $A_1 \otimes A_2$. Equivalently, we can take $E=R_1 \otimes R_2$, such that, reference systems $R_1$ and $R_2$ are isomorphic to $A_1$ and $A_2$, respectively. Now, using the facts that $R_1 \otimes R_2 \otimes A_1 \otimes A_2$ is isomorphic to $R_1 \otimes A_1 \otimes R_2 \otimes A_2$, and the supremum of a function over a set is always greater than or equal to the supremum of the same over its subsets, we obtain
\begin{align}
     &\sup_{\Psi_{R_1 R_2 A_1 A_2}} \staD*{(\mathrm{id}_{R_1 R_2} \otimes {\mathcal{N}}_{A_1 \rightarrow B_1} \otimes {\mathcal{M}}_{A_2 \rightarrow B_2} ) \Psi_{EA_1 A_2}}{(\mathrm{id}_{R_1 R_2} \otimes {\mathcal{R}}_{A_1 \rightarrow B_1} \otimes {\mathcal{R}}_{A_2 \rightarrow B_2} ) \Psi_{{EA_1 A_2}}}\nonumber\\
     & \geq \sup_{\widetilde{\Psi}} \staD*{(\mathrm{id}_{R_1} \otimes {\mathcal{N}}_{A_1 \rightarrow B_1}) \Psi_{R_1 A_1} \otimes (\mathrm{id}_{R_2} \otimes {\mathcal{M}}_{A_2 \rightarrow B_2} ) \Psi_{ R_2 A_2}}{(\mathrm{id}_{R_1} \otimes {\mathcal{R}}_{A_1 \rightarrow B_1}) \Psi_{R_1 A_1} \otimes (\mathrm{id}_{R_2} \otimes {\mathcal{R}}_{A_2 \rightarrow B_2} ) \Psi_{ R_2 A_2}} \nonumber\\
    &= \sup_{\widetilde{\Psi}}\left\{ \right. \mathrm{tr}\left[(\mathrm{id}_{R_2} \otimes {\mathcal{M}}_{A_2 \rightarrow B_2} ) \Psi_{ R_2 A_2} \right] ~\staD*{(\mathrm{id}_{R_1} \otimes {\mathcal{N}}_{A_1 \rightarrow B_1}) \Psi_{R_1 A_1}}{(\mathrm{id}_{R_1} \otimes {\mathcal{R}}_{A_1 \rightarrow B_1}) \Psi_{R_1 A_1}} \nonumber \\
    &\hspace{1cm} + \mathrm{tr}\left[(\mathrm{id}_{R_1} \otimes {\mathcal{N}}_{A_1 \rightarrow B_1} ) \Psi_{ R_1 A_1} \right]~ \staD*{(\mathrm{id}_{R_2} \otimes {\mathcal{M}}_{A_2 \rightarrow B_2}) \Psi_{R_2 A_2}}{(\mathrm{id}_{R_2} \otimes {\mathcal{R}}_{A_2 \rightarrow B_2}) \Psi_{R_2 A_2}} \left.  \right\}  \nonumber  \\ 
   &= \sup_{\widetilde{\Psi}_{1}} \staD*{(\mathrm{id}_{R_1} \otimes {\mathcal{N}}_{A_1 \rightarrow B_1}) \Psi_{R_1 A_1}}{(\mathrm{id}_{R_1} \otimes {\mathcal{R}}_{A_1 \rightarrow B_1}) \Psi_{R_1 A_1}}\nonumber\\
   &\hspace{1cm}+\sup_{\widetilde{\Psi}_{2}} \staD*{(\mathrm{id}_{R_2} \otimes {\mathcal{M}}_{A_2 \rightarrow B_2}) \Psi_{R_2 A_2}}{(\mathrm{id}_{R_2} \otimes {\mathcal{R}}_{A_2 \rightarrow B_2}) \Psi_{R_2 A_2}} \nonumber\\
    &=\chD*{\mathcal{N}_{A_1 \rightarrow B_1}}{\mathcal{R}_{A_1\rightarrow B_1}} + \chD*{\mathcal{M}_{A_2 \rightarrow B_2}}{\mathcal{R}_{A_2\rightarrow B_2}},
\end{align}
where, $\widetilde{\Psi}:=\Psi_{R_1 A_1} \otimes \Psi_{ R_2 A_2}$, $\widetilde{\Psi}_{1}:=\Psi_{R_1 A_1}$ and $\widetilde{\Psi}_{2}:=\Psi_{R_2 A_2}$.  In the second line we have used the additivity of quantum relative entropy, i.e., for positive semidefinite operators $\rho_{1},\rho_{2},\sigma_{1}$, and $\sigma_{2}$, we have, $D\left(\rho_{1}\otimes\rho_{2}\Vert \sigma_{1}\otimes\sigma_{2}\right) = \mathrm{tr}\left(\rho_{2}\right)D\left(\rho_{1}\Vert \sigma_{1}\right)+\mathrm{tr}\left(\rho_{1}\right)D\left(\rho_{2}\Vert \sigma_{2}\right)$, and the last line simply follows from the definition of relative entropy between a  quantum channel and a CP map. Thus, we have the following inequality,
\begin{align}
   - \chD*{\mathcal{N}_{A_1\rightarrow B_1}\otimes {\mathcal{M}}_{A_2 \rightarrow B_2}}{\mathcal{R}_{A_1\rightarrow B_1}\otimes {\mathcal{R}}_{A_2 \rightarrow B_2}}  \leq -\chD*{\mathcal{N}_{A_1 \rightarrow B_1}}{\mathcal{R}_{A_1\rightarrow B_1}} - \chD*{\mathcal{M}_{A_2 \rightarrow B_2}}{\mathcal{R}_{A_2\rightarrow B_2}},
\end{align}
which implies that 
\begin{align}
   S\left[{\mathcal{N}}_{A_1 \rightarrow B_1} \otimes {\mathcal{M}}_{A_2 \rightarrow B_2}\right] \leq S\left[{\mathcal{N}}_{A_1 \rightarrow B_1}] + S[{\mathcal{M}}_{A_2 \rightarrow B_2}\right].\label{equ:additivity_of_entropy_ge}
\end{align}
Combining  inequalities~\eqref{equ_additivity_of_entropy_1} and \eqref{equ:additivity_of_entropy_ge}, we  conclude that the entropy functional of a quantum channel is additive under tensor product of channels, i.e,
\begin{equation}
     S\left[{\mathcal{N}}_{A_1 \rightarrow B_1} \otimes {\mathcal{M}}_{A_2 \rightarrow B_2}\right] = S\left[{\mathcal{N}}_{A_1 \rightarrow B_1}\right] + S\left[{\mathcal{M}}_{A_2 \rightarrow B_2}\right].
\end{equation}

\subsection{Entropy gain under positive maps}\label{sec_entropy_gain_under_posive_maps}
\label{sec:lemma1}
    
   \begin{lemma}
 Let us consider a pair of positive semidefinite operators $\{\sigma \in \mathcal{L}\left(C\right)_{+},\rho\in \mathcal{L}\left(A\right)_{+}\}$ such that whenever the finite sum $\sum_{n}c_{n} \sigma^{n} = 0$, we have $\sum_{n}c_{n} \rho^{n} = 0$, where $c_n$ are complex numbers. Here we adopt the convention that $\sigma^{0}=\mathbbm{1}_C$ and $\rho^{0}=\mathbbm{1}_A$. Then there exists a unital CP map $\mathcal{F}: \mathcal{L}\left(C\right) \rightarrow \mathcal{L}\left(A\right)$, such that $\rho= \mathcal{F}(\sigma)$.
   \end{lemma}  
   \begin{proof}
      Consider a set $\mathtt{K}_C =\{ \sigma, \mathbbm{1} \}\subset \mathcal{L}\left(C\right)$, and let $\mathcal{A}_{\mathtt{K}_C}$ denote the sub-algebra generated by $\mathtt{K}_{C}$. Clearly, $\mathcal{A}_{\mathtt{K}_C}$ is a commutative unital $C^{*}$ sub-algebra of $\mathcal{L}(C)$. Now, we define a map $\mathcal{F}: \mathcal{A}_{\mathtt{K}_C} \rightarrow \mathcal{L}\left(A\right)$ by $\mathcal{F}(\sigma^{n})=\rho^{n}$, for all $n \in \mathbb{N} \cup \{0\}$, and extend it by linearity to the whole space. As $\sigma$ is positive and hence, Hermitian, it is clear that the map $\mathcal{F}$ is a $^{\ast}-\text{homomorphism}$ which also implies that $\mathcal{F}$ is a positive map. As $\mathcal{A}_{\mathtt{K}_C}$ is a commutative unital $C^{*}$ algebra and $\mathcal{L}\left(A\right)$ is a unital $C^{*}$ algebra, we have that $\mathcal{F}$ is completely positive~\cite{paulsen_2003}. We now employ Arveson's extension theorem to extend it to $\mathcal{L}\left(C\right)$.
      
    \emph{Arveson’s extension theorem}~\cite{paulsen_2003}: Let $\mathfrak{B}$ be a $C^{*}$-algebra and $\mathfrak{D} \subseteq \mathfrak{B}$ an operator system (a Hermitian subspace containing the identity of the $C^{*}$-algebra). Let $\mathcal{F}:\mathfrak{D} \rightarrow \mathcal{L}\left(A\right)$ be a completely positive map. Then there exists a completely positive map, $\mathcal{F'}: \mathfrak{B} \rightarrow \mathcal{L}\left(A\right)$, which extends $\mathcal{F}$ to the whole algebra $\mathfrak{B}$. Clearly, $\mathcal{A}_{\mathtt{K}_C}$ is an operator system. Hence, the map $\mathcal{F}$ extends to the full algebra $\mathcal{L}\left(C\right)$. 
   \end{proof}
Next, we derive several entropic inequalities for unital CP maps, CP maps, and positive maps. Our results also reduce to well-known previous results for an appropriate choice of a parameter involved in the remainder term of these bounds. Our main result in this section is Theorem \ref{result_entropy_gain_1}, which is on the entropy gain of positive maps.

\begin{proposition} \label{prop_binod}
    Let $\rho \in \mathcal{D}\left(A\right)$, $\sigma \in \mathcal{D}\left(C\right)$ and $\mathcal{F}:\mathcal{L}\left(C\right) \rightarrow \mathcal{L}\left(A\right)$ be a  completely positive unital map that connects $\rho$ and $\sigma$. Then, the von Neumann entropy of $\rho$ and $\sigma$ are related by the following equation.
    \begin{equation}
        S\left(\rho\right) \geq \norm{\mathcal{F}^{*}(\mathbbm{1}_{A})}_{\infty}S\left(\sigma\right),
    \end{equation}
    where $\mathcal{F}^{*}: \mathcal{L}\left(A\right) \rightarrow \mathcal{L}\left(C\right)$ is the adjoint of $\mathcal{F}$.
\end{proposition}
\begin{proof}
    Let $\mathcal{F}:\mathcal{L}\left(C\right) \rightarrow \mathcal{L}\left(A\right)$ denote a unital and completely positive map, then we can define a completely positive trace nonincreasing map $\widetilde{\mathcal{F}}:\mathcal{L}\left(C\right) \rightarrow \mathcal{L}\left(A\right)$ as follows: $\widetilde{\mathcal{F}}(\sigma):= \frac{\mathcal{F}(\sigma)}{\norm{\mathcal{F}^{*}(\mathbbm{1}_A)}_{\infty}} $. Moreover, for any unital and completely positive map $\mathcal{F}:\mathcal{L}\left(C\right) \rightarrow \mathcal{L}\left(A\right)$, and for any positive operator $\sigma\in \mathcal{L}\left(C\right)$, the following operator inequality holds (see \cite{Wolf2012}): $\mathcal{F}\left(\sigma\right)\log \mathcal{F}\left(\sigma\right)\leq \mathcal{F}\left(\sigma \log \sigma\right)=\norm{\mathcal{F}^{*}\left(\mathbbm{1}_A\right)}_{\infty}\widetilde{\mathcal{F}}(\sigma \log \sigma)$.
    By taking trace on both sides and using the fact that $\widetilde{F}$ is trace nonincreasing, we then have $-S(\left(\mathcal{F}\left(\sigma\right)\right) \leq - \norm{\mathcal{F}^{*}\left(\mathbbm{1}_A\right)}_{\infty} S\left(\sigma\right)$, which implies that $S\left(\rho\right) \geq \norm{\mathcal{F}^{*}\left(\mathbbm{1}_A\right)}_{\infty} S\left(\sigma\right)$.

 Note that, in principle, there may exist more than one unital CP map connecting $\rho$ to $\sigma$ ~(which we also proved earlier in the beginning of this section). Thus, this lower-bound can be further tightened as $S(\rho) \geq \sup\{\norm{\mathcal{F}^{*}\left(\mathbbm{1}_A\right)}_{\infty}\} S\left(\sigma\right)$, where the supremum is taken over all possible unital CP maps $\mathcal{F}$ connecting states $\rho$ and $\sigma$. Further, notice that the value of $\sup\{\norm{\mathcal{F}^{*}\left(\mathbbm{1}_A\right)}_{\infty}\}$ decides whether the entropy gain is positive or negative. Thus, we can make the following remark: For a completely positive unital map $\mathcal{F}$, entropy gain is non-negative if  $\sup\{\norm{\mathcal{F}^{*}\left(\mathbbm{1}_A\right)}_{\infty}\} \geq 1 $. Once we assume that the map $\mathcal{F}$ is also trace-preserving, which implies that $\mathcal{F}^{*}$ is unital, we then have $\sup\{\norm{\mathcal{F}^{*}\left(\mathbbm{1}_A\right)}_{\infty}\} = 1$, thereby retrieving the well-known result of~\cite{Alberti:1977wc}.
\end{proof}

 \begin{lemma}\label{lemma_binods}
     Let $\mathcal{F}:\mathcal{L}\left(C\right) \rightarrow \mathcal{L}\left(A\right)$ be a completely positive map, and let the entropy gain of $\sigma \in \mathcal{D}\left({C}\right)$  under the map $\mathcal{F}$ be defined as $\Delta S\left(\sigma, \mathcal{F}\right):=S\left(\mathcal{F}\left(\sigma\right)\right) - S\left(\sigma\right)$. Then, for all $\sigma \in \mathcal{D}\left({C}\right)$ such that $\mathcal{F}\left(\sigma\right)>0$, we have
    \begin{eqnarray}
        \Delta S\left(\sigma, \mathcal{F}\right) \geq \max \left\{\staD*{\sigma}{\sigma_{\alpha}}, \staD*{\sigma}{ \hat{\sigma}_{\alpha} }\right \},
    \end{eqnarray}
    where $\alpha := \norm{\mathcal{F}^{*}(\mathbbm{1}_{A})}_{\infty}$, $\sigma_{\alpha} :=\left(\frac{1}{\alpha}\right)^{\alpha}\left(\mathcal{F}^{*}\circ \mathcal{F}(\sigma)\right)^{\alpha}$, and $\hat{\sigma}_{\alpha}:=\frac{1}{\alpha}\mathcal{F}^{*} \left(\mathcal{F}(\sigma)\right)^{\alpha}$.
 \end{lemma} 
 
    \begin{proof}
        We first define a completely positive, trace nonincreasing map $\widetilde{\mathcal{F}}:\mathcal{L}\left(C\right) \rightarrow \mathcal{L}\left(A\right)$ from the map $\mathcal{F}:\mathcal{L}\left(C\right) \rightarrow \mathcal{L}\left(A\right)$ as $\widetilde{F}(\sigma):= \frac{1}{\alpha}\mathcal{F}(\sigma)$, where $\alpha := \norm{\mathcal{F}^{*}(\mathbbm{1}_{A})}_{\infty}$. Now, using the definition of adjoint, the entropy gain can be written as $\Delta S\left(\sigma, \mathcal{F}\right) = \tr\left(\sigma \log \sigma\right)- \tr\left(\alpha \sigma \widetilde{\mathcal{F}}^{*}\left(\log \mathcal{F}\left(\sigma\right)\right)\right)$. As $\mathcal{F}\left(\sigma\right)$ is a full rank operator by assumption and $\widetilde{\mathcal{F}}^{*}$ is subunital, we use Eq.~(\ref{entropy_gain_subunital}) to arrive at $\Delta S\left(\sigma, \mathcal{F}\right) \geq  \tr\left(\sigma \log \sigma\right)- \tr\left( \sigma \alpha(\log \widetilde{\mathcal{F}}^{*} \circ \mathcal{F}\left(\sigma\right))\right)$, which implies $\Delta S\left(\sigma, \mathcal{F}\right) \geq \staD*{\sigma }{ \sigma_{\alpha}}$, where $\sigma_{\alpha} =\left(\frac{1}{\alpha}\right)^{\alpha}\left(\mathcal{F}^{*}\circ \mathcal{F}(\sigma)\right)^\alpha$. Note that we can also write the entropy gain as $\Delta S\left(\sigma, \mathcal{F}\right) = \tr\left(\sigma \log \sigma\right)- \tr\left(\widetilde{\mathcal{F}}\left(\sigma\right) \log \left(\mathcal{F}\left(\sigma\right)\right)^{\alpha}\right)$. Then, by following the similar arguments as above, we obtain another bound on the entropy gain as  $\Delta S\left(\sigma, \mathcal{F}\right) \geq \staD*{\sigma }{\hat{\sigma}_{\alpha}}$, where $\hat{\sigma}_{\alpha}:=\frac{1}{\alpha}\mathcal{F}^{*} \left(\mathcal{F}(\sigma)\right)^{\alpha}$. Since both the remainder terms, $\staD*{\sigma}{\sigma_{\alpha}}$ and $\staD*{\sigma}{ \hat{\sigma}_{\alpha} }$, are valid lower bounds on the entropy gain, we have 
        \begin{eqnarray}
        \Delta S\left(\sigma, \mathcal{F}\right) \geq \max \left\{\staD*{\sigma}{\sigma_{\alpha}}, \staD*{\sigma}{ \hat{\sigma}_{\alpha} }\right \}.
    \end{eqnarray}
    This completes the proof of the lemma.
    \end{proof}
    
If we take $\mathcal{F}^{*}$ to be unital~(or equivalently $\mathcal{F}$ to be a quantum channel) in the above lemma, then we have $\alpha = 1$,  and we obtain a well-known result~\cite{Buscemi_2016}: $\Delta S\left(\sigma, \mathcal{F}\right) \geq D\left(\sigma \Vert \mathcal{F}^{*}\circ \mathcal{F}(\sigma)\right)$. We now state some known facts which we will be using to prove our next results.

\begin{enumerate}
    \item Let $A$ and $B$ be two finite-dimensional Hilbert spaces, and let $\mathcal{F}: \mathcal{L}(A) \rightarrow \mathcal{L}(B)$ be a positive map such that $\mathcal{F}(\mathbbm{1}_A) \leq \mathbbm{1}_B$. Then for all $\mathsf{X} > 0$, we have the following inequality~\cite{Wolf2012}:
    \begin{eqnarray}
        \mathcal{F}\left(\log (\mathsf{X})\right) \leq \log (\mathcal{F}(\mathsf{X})). \label{entropy_gain_subunital}
    \end{eqnarray} 
    \item   Let $\rho \in \mathcal{D}(A)$ be a quantum state, and $\sigma \in \mathcal{L}(A)_{+}$ be a positive semidefinite operator. Then the following limit holds~\cite{wilde2016}:
    \begin{eqnarray}
        \staD*{\rho}{ \sigma}= \lim_{\epsilon \rightarrow 0+} \staD*{\rho }{ \sigma + \epsilon \mathbbm{1}_A}. \label{wilde 1}
    \end{eqnarray}
    \item Let $\rho \in \mathcal{D}(A)$ be a quantum state, and $\sigma, \sigma' \in \mathcal{L}(A)_{+}$ be positive semidefinite operators such that $\sigma \leq \sigma'$. Then we have the following inequality~\cite{wilde2016}  :
    \begin{eqnarray}
        \staD*{\rho}{ \sigma} \geq \staD*{\rho}{ \sigma'}. \label{wilde 2}
    \end{eqnarray}
\end{enumerate}

In the following, we provide a twofold generalization of Lemma \ref{lemma_binods}: (1) we consider positive maps instead of completely positive maps, (2) we drop the requirement that $\mathcal{F}(\sigma)>0$.

\begin{theorem} \label{result_entropy_gain_1}
    Let $\mathcal{F}:\mathcal{L}\left(A\right) \rightarrow \mathcal{L}\left(B\right)$ be a positive map. Then, for all $\rho \in \mathcal{D}\left(A\right)$, the change in entropy after and before applying the map is given by
    \begin{eqnarray}
        S\left(\mathcal{F}\left(\rho\right)\right) - S\left(\rho\right) \geq \staD*{\rho}{\rho_{\alpha}},
    \end{eqnarray}
    where $\alpha := \norm{\mathcal{F}^{*}(\mathbbm{1}_B)}_{\infty}$ and  $\rho_{\alpha} :=\frac{1}{\alpha^\alpha}\left(\mathcal{F}^{*}\circ \mathcal{F}(\rho)\right)^{\alpha}$.
 \end{theorem}   
    \begin{proof}
        For the given positive map $\mathcal{F}$, we define a positive and trace nonincreasing map $\widetilde{\mathcal{F}}$ as $\widetilde{\mathcal{F}}(\rho):= \frac{1}{\alpha}\mathcal{F}(\rho)$, where $\alpha := \norm{\mathcal{F}^{*}(\mathbbm{1}_B)}_{\infty}$.
        Now we define an operator $\sigma := \mathcal{F}(\rho)+ \epsilon \mathbbm{1}_B$. Clearly for all $\epsilon > 0$, we have $\sigma >0$. Let us consider the following:
        \begin{eqnarray}
            S(\sigma)-S(\rho) &&= \tr (\rho \log \rho)-\tr (\sigma \log \sigma) \nonumber \\
            &&= \tr (\rho \log \rho)-\tr ( \alpha \widetilde{\mathcal{F}}(\rho) \log \sigma)- \epsilon \tr (\log \sigma) \nonumber \\
            &&= \tr (\rho \log \rho)- \alpha \tr (\rho \widetilde{\mathcal{F}}^{*}\log \sigma)- \epsilon \tr (\log \sigma) \nonumber\\
            && \geq \tr (\rho \log \rho)- \alpha \tr (\rho \log \widetilde{\mathcal{F}}^{*}(\sigma))- \epsilon \tr (\log \sigma) \nonumber\\
            &&= (1-\alpha)\tr (\rho \log \rho)+ \alpha \staD*{\rho}{\widetilde{\mathcal{F}}^{*}(\sigma)}- \epsilon \tr (\log \sigma).
        \end{eqnarray}
         As $\widetilde{\mathcal{F}}$ is trace nonincreasing, we have $\widetilde{\mathcal{F}}^{*}(\mathbbm{1}_B) \leq \mathbbm{1}_A$, then the inequality in the fourth line of above relations follows from the inequality~(\ref{entropy_gain_subunital}). Now $\widetilde{\mathcal{F}}^{*}(\sigma)=\widetilde{\mathcal{F}}^{*} \circ \mathcal{F} (\rho)+ \epsilon \widetilde{\mathcal{F}}^{*} (\mathbbm{1}_B) \leq \widetilde{\mathcal{F}}^{*} \circ \mathcal{F} (\rho)+ \epsilon \mathbbm{1}_A$. Thus, using inequality~(\ref{wilde 2}), we have $ \staD*{\rho}{\widetilde{\mathcal{F}}^{*}(\sigma)} \geq \staD*{\rho}{\widetilde{\mathcal{F}}^{*} \circ \mathcal{F} (\rho)+ \epsilon \mathbbm{1}_A}$. Using this, we get 
        \begin{eqnarray}
         S(\sigma)-S(\rho) && \geq  (1-\alpha)\tr (\rho \log \rho)+  \alpha \staD*{\rho}{\widetilde{\mathcal{F}}^{*} \circ \mathcal{F} (\rho)+ \epsilon \mathbbm{1}_A}- \epsilon \tr (\log \sigma).
        \end{eqnarray}
        Now we use Eq.~(\ref{wilde 1}) and the continuity of the von Neumann entropy to arrive at the following:
        \begin{eqnarray}
            S(\mathcal{F}(\rho)) -S(\rho) &&\geq (1-\alpha)\tr (\rho \log \rho)+  \alpha 
 \staD*{\rho}{\widetilde{\mathcal{F}}^{*} \circ \mathcal{F} (\rho)} \nonumber \\
            &&= \tr (\rho \log \rho) - \tr\left(\rho \alpha \log \widetilde{\mathcal{F}}^{*} \circ \mathcal{F} (\rho)\right) \nonumber\\
            &&= \tr \left(\rho \log \rho\right) - \tr\left(\rho \log \left\{\widetilde{\mathcal{F}}^{*} \circ \mathcal{F} (\rho)\right\}^{\alpha}\right) \nonumber\\
            &&= \staD*{\rho}{\left\{\widetilde{\mathcal{F}}^{*} \circ \mathcal{F} (\rho)\right\}^{\alpha}} = \staD*{\rho}{\left(\frac{1}{\alpha}\right)^{\alpha} \left(\mathcal{F}^{*}\circ \mathcal{F}(\rho)\right)^{\alpha}}.
        \end{eqnarray}
    This completes the proof of the theorem.
    
   \end{proof}
If we take $\mathcal{F}^{*}$ to be unital~(or equivalently $\mathcal{F}$ to be a trace-preserving) in Theorem \ref{result_entropy_gain_1}, then we have $\alpha = 1$,  and we obtain a well-known result~\cite{Buscemi_2016}: $S\left(\mathcal{F}(\rho)\right) -S(\rho) \geq \staD*{\rho}{\mathcal{F}^{*}\circ \mathcal{F}(\sigma)}$.

\subsection{Proof of Entropy gain under superchannels } \label{sec_proof_of_channel_entropy_gain}

Let us consider a channel $\mathcal
{N}:\mathcal{L}\left({A}\right) \rightarrow \mathcal{L}\left(B\right)$ such that the state that realizes the supremum in the entropy functional $S\left[\mathcal{N}\right]$ is an element of $ \mathsf{FRank}\left(\widetilde{A}\otimes{A}\right)\cap \mathcal{D}\left(\widetilde{A}\otimes{A}\right)$. 
 Then we can write
\begin{align}
\chD*{\mathcal{N}}{\mathcal{R}} &= \sup_{\rho_{\widetilde{A}A}} \staD*{\mathrm{id}_{\widetilde{A}}\otimes\mathcal{N}_{A \rightarrow B} \left(\rho_{\widetilde{A}A}\right)}{\mathrm{id}_{\widetilde{A}}\otimes\mathcal{R}_{A \rightarrow B}\left(\rho_{\widetilde{A}A}\right)}\nonumber\\
    & =  \staD*{\left(\mathrm{id}_{\widetilde{A}}\otimes \mathcal{N}_{A \rightarrow B}\right)\ketbra{\Psi}{\Psi}}{\left(\mathrm{id}_{\widetilde{A}}\otimes \mathcal{R}_{A \rightarrow B}\right)\ketbra{\Psi}{\Psi}},
\end{align}
where $\ketbra{\Psi}{\Psi} \in  \mathsf{FRank} \left(\widetilde{A}\otimes{A}\right)\cap \mathcal{D}\left(\widetilde{A}\otimes{A}\right)$. A simple calculation shows that $\left(\mathrm{id}_{\widetilde{A}}\otimes{R}_{A \rightarrow B}\right)\ketbra{\Psi}{\Psi}= \Psi_{\widetilde{A}}\otimes\mathbbm{1}_{B}$, where $\Psi_{\widetilde{A}} = \mathrm{tr}_{A}\left(\ketbra{\Psi}{\Psi}\right)$. Since $\Psi_{\widetilde{A}}$ is full rank operator, we have  $\left(\mathrm{id}_{\widetilde{A}}\otimes \mathcal{N}_{A \rightarrow B} \right)\ketbra{\Psi}{\Psi} = {\mathsf{C}}_{\mathcal{N}}^{\Psi}$, where ${\mathsf{C}}_{\mathcal{N}}^{\Psi}$ is the Choi state of the quantum channel $\mathcal{N}_{A \rightarrow B}$ with respect to $\ketbra{\Psi}{\Psi}$. Thus, we can write the entropy of the quantum channel $\mathcal{N}_{A \rightarrow B}$ as 
\begin{align}
\label{eq:s-before-sup-chan}
    S\left[\mathcal{N}_{A \rightarrow B}\right] &= - \staD*{{\mathsf{C}}_{\mathcal{N}}^{\Psi}}{\Psi_{\widetilde{A}}\otimes\mathbbm{1}_{B}}= S\left({\mathsf{C}}_{\mathcal{N}}^{\Psi}\right) - S\left(\Psi_{\widetilde{A}}\right). 
\end{align}

Now, consider a superchannel $\Theta:\mathcal{L}\left(A ,B\right)\rightarrow \mathcal{L}\left(C,D\right)$ such that the state that realizes the supremum of the entropy functional $S\left[\Theta\left(\mathcal{N}_{A \rightarrow B}\right)\right]$ is an element $\ketbra{\Phi}{\Phi} \in \mathsf{FRank}\left(\widetilde{C}\otimes{C}\right)\cap \mathcal{D}\left(\widetilde{C}\otimes{C}\right)$. Again via a simple calculation, we can show that $\left(\mathrm{id}_{\widetilde{C}}\otimes {R}_{C \rightarrow D}\right)\ketbra{\Phi}{\Phi} = \Phi_{\widetilde{C}}\otimes\mathbbm{1}_{D}$, where $\Phi_{\widetilde{C}} = \mathrm{tr}_{C}\left(\ketbra{\Phi}{\Phi}\right)$. Since $\Phi_{\widetilde{C}}$ is a full rank operator, we have  $\left(\mathrm{id}_{\widetilde{C}}\otimes \Theta\left(\mathcal{N}_{A \rightarrow B}\right)\right)\ketbra{\Phi}{\Phi}  = \mathsf{C}_{\Theta\left(\mathcal{N}\right)}^{\Phi}$, where $\mathsf{C}_{\Theta\left(\mathcal{N}\right)}^{\Phi}$ is the Choi state for the quantum channel $\Theta \left(\mathcal{N}_{A \rightarrow B}\right)$. Thus, we can write the entropy of quantum channel $\Theta \left(\mathcal{N}_{A \rightarrow B}\right)$ as 
\begin{align}
\label{eq:s-after-sup-chan}
    S\left[\Theta\left(\mathcal{N}_{A \rightarrow B}\right)\right] &= - \staD*{{\mathsf{C}}_{\Theta \left(\mathcal{N}\right)}^{\Phi} }{\Phi_{\widetilde{C}}\otimes\mathbbm{1}_{D}} = S\left({\mathsf{C}}_{\Theta \left(\mathcal{N}\right)}^{\Phi}\right) - S\left(\Phi_{\widetilde{C}}\right).
\end{align}
Now, using Eqs. \eqref{eq:s-before-sup-chan} and \eqref{eq:s-after-sup-chan}, the difference in the entropy of the quantum channel $\mathcal{N}_{A \rightarrow B}$ after and before applying the superchannel $\Theta$ becomes
\begin{align}
     S\left[\Theta\left(\mathcal{N}_{A \rightarrow B}\right)\right] - S\left[\mathcal{N}_{A \rightarrow B}\right]
     =  \Delta \left[\mathsf{C}_{\Theta(\mathcal{N})}^{\Phi},\mathsf{C}_{\mathcal{N}}^{\Psi}\right]  + \Delta' \left[\Psi_{\widetilde{A}},\Phi_{\widetilde{C}}\right],
\end{align}
where $\Delta \left[\mathsf{C}_{\Theta(\mathcal{N})}^{\Phi},\mathsf{C}_{\mathcal{N}}^{\Psi}\right]:=S\left(\mathsf{C}_{\Theta \left(\mathcal{N}\right)}^{\Phi}\right) - S \left(\mathsf{C}_{\mathcal{N}}^{\Psi}\right)$ and $\Delta' \left[\Psi_{\widetilde{A}},\Phi_{\widetilde{C}}\right]:=S\left(\Psi_{\widetilde{A}}\right) - S\left(\Phi_{\widetilde{C}}\right)$.

In the following, we aim to find a lower bound for the entropy gain. One may calculate a lower bound on the functions $\Delta$ and $\Delta'$ separately or together. Let $\mathfrak{T}$ be the representing map for the superchannel $\Theta$. Since the Choi states $\mathsf{C}_{\Theta \left(\mathcal{N}\right)}^{\Phi}$ and $\mathsf{C}_{\mathcal{N}}^{\Psi}$ are connected via the representing map $\mathfrak{T}$ (see Eq. \eqref{equ:new_rep_1} and Remark \ref{rem:choi-state-choi-operator-relation_1}), we can find a lower bound on $\Delta \left[\mathsf{C}^{\Phi}_{\Theta\left(\mathcal{N}\right)},\mathsf{C}_{\mathcal{N}}^{\Psi}\right]$ in terms of the relative entropy between the Choi state $\mathsf{C}_{\mathcal{N}}^{\Psi}$ and the positive operator $\left(\mathsf{C}_{\mathcal{N}}^{\Psi}\right)_{\alpha}$ using Theorem \ref{result_entropy_gain_1}, as follows:
\begin{align}
    \Delta \left[\mathsf{C}_{\Theta(\mathcal{N})}^{\Phi},\mathsf{C}_{\mathcal{N}}^{\Psi}\right] & =   S\left(\mathfrak{T}\left(\mathsf{C}_{\mathcal{N}}^{\Psi}\right)\right) - S \left(\mathsf{C}_{\mathcal{N}}^{\Psi}\right)\geq \staD*{\mathsf{C}_{\mathcal{N}}^{\Psi}}{\left(\mathsf{C}_{\mathcal{N}}^{\Psi}\right)_{\alpha}} ,
\end{align}
where  $\left(\mathsf{C}_{\mathcal{N}}^{\Psi}\right)_{\alpha}= \frac{1}{\alpha^{\alpha}} \left(\mathfrak{T}^{*} \circ \mathfrak{T}\left( \mathsf{C}_{\mathcal{N}}^{\Psi} \right)\right)^{\alpha}= \frac{1}{\alpha^{\alpha}} \left(\mathsf{C}^{\Psi}_{{\Hat{\mathrm{\Theta}}_{\Psi}}^{-1} \circ \mathrm{\Theta}^{*} \circ \Hat{\mathrm{\Theta}}_{\Phi}\circ \mathrm{\Theta}(\mathcal{N})}\right)^{\alpha}$ and $\alpha= \norm{\mathfrak{T}^{*}(\mathbbm{1})}_{\infty}=\norm{\mathsf{C}^{\Psi}_{{\Hat{\mathrm{\Theta}}_{\Psi}}^{-1} \circ \mathrm{\Theta}^{*} \circ \Hat{\mathrm{\Theta}}_{\Phi}\circ \mathrm{\Theta}^{-1}_{\Phi}(\mathcal{R})}}_{\infty}$ with $\mathcal{R}$ being the completely depolarizing map in $\mathcal{L}\left(C,D\right)$. The resulting inequality provides a lower bound of $ \Delta \left[\mathsf{C}_{\Theta(\mathcal{N})}^{\Phi},\mathsf{C}_{\mathcal{N}}^{\Psi}\right]$. With this, we conclude that the following inequality holds
\begin{equation}
    S\left[\Theta\left(\mathcal{N}_{A \rightarrow B}\right)\right] - S\left[\mathcal{N}_{A \rightarrow B}\right]
     \geq  \staD*{\mathsf{C}_{\mathcal{N}}^{\Psi}}{\left(\mathsf{C}_{\mathcal{N}}^{\Psi}\right)_{\alpha}} + \Delta' \left[\Psi_{\widetilde{A}},\Phi_{\widetilde{C}}\right],
\end{equation}
which completes our proof.

\section{Refinement of data processing inequality} \label{detailde_proof_of_recovery}

 Let us consider a pair of quantum channels  $\mathcal
{N}:\mathcal{L}\left({A}\right) \rightarrow \mathcal{L}\left(A'\right)$ and $\mathcal
{M}:\mathcal{L}\left({A}\right) \rightarrow \mathcal{L}\left(A'\right)$ such that the state that realizes the supremum in the relative entropy functional $D\left[\mathcal{N}\Vert \mathcal{M}\right]$ is an element of $\mathsf{FRank}\left(\widetilde{A}\otimes{A}\right)\cap \mathcal{D}\left(\widetilde{A}\otimes{A}\right)$. Then we can write
\begin{align}
  \chD*{\mathcal{N}}{\mathcal{M}} &= \sup_{\rho_{\widetilde{A}A}} \staD*{\mathrm{id}_{\widetilde{A}}\otimes\mathcal{N} \left(\rho_{\widetilde{A}A}\right)}{\mathrm{id}_{\widetilde{A}}\otimes\mathcal{M}\left(\rho_{\widetilde{A}A}\right)}\nonumber\\
    & = \staD*{\left(\mathrm{id}_{\widetilde{A}}\otimes \mathcal{N}\right)\ketbra{\Psi}{\Psi}}{\left(\mathrm{id}_{\widetilde{A}}\otimes \mathcal{M}\right)\ketbra{\Psi}{\Psi}},
\end{align}
where $\ketbra{\Psi}{\Psi} \in \mathsf{FRank}\left(\widetilde{A}\otimes{A}\right)\cap \mathcal{D}\left(\widetilde{A}\otimes{A}\right)$, i.e., $ \rm{rank}( \mathrm{tr}_{A}\left(\ketbra{\Psi}{\Psi}\right)) = \vert \widetilde{A}\vert$. We then have  $\left(\mathrm{id}_{\widetilde{A}}\otimes \mathcal{N} \right)\ketbra{\Psi}{\Psi} = {\mathsf{C}}_{\mathcal{N}}^{\Psi}$ and $\left(\mathrm{id}_{\widetilde{A}}\otimes \mathcal{N} \right)\ketbra{\Psi}{\Psi} = {\mathsf{C}}_{\mathcal{M}}^{\Psi}$, where ${\mathsf{C}}_{\mathcal{N}}^{\Psi}$ and ${\mathsf{C}}_{\mathcal{M}}^{\Psi}$ are Choi states of quantum channels $\mathcal{N}$ and $\mathcal{M}$, respectively, with respect to $\Psi$. Therefore, in this case, we can write the relative entropy between quantum channels $\mathcal{N}$ and $\mathcal{M}$ as 
\begin{align}
   \chD*{\mathcal{N}}{\mathcal{M}} = \staD*{{\mathsf{C}}_{\mathcal{N}}^{\Psi}}{{\mathsf{C}}_{\mathcal{M}}^{\Psi}}. 
\end{align}
Now, consider a superchannel $\Theta:\mathcal{L}\left(A ,A'\right)\rightarrow \mathcal{L}\left(C,C'\right)$ such that the state that realizes the supremum in the  relative entropy functional $\chD*{\Theta(\mathcal{N})}{\Theta(\mathcal{M})}$ is an element of $\mathsf{FRank}\left(\widetilde{C}\otimes{C}\right)\cap \mathcal{D}\left(\widetilde{C}\otimes{C}\right)$. In particular, let $\ket{\Phi}$ be the state that realizes the supremum of the relative entropy $\chD*{\Theta(\mathcal{N})}{\Theta(\mathcal{M})}$. Since $\rm{rank}\left(\mathrm{tr}_{C}\left(\ketbra{\Phi}{\Phi}\right)\right) = \vert \widetilde{C}\vert$, we  have  $\left(\mathrm{id}_{\widetilde{C}}\otimes \Theta\left(\mathcal{N}\right)\right)\ketbra{\Phi}{\Phi}  = \mathsf{C}_{\Theta\left(\mathcal{N}\right)}^{\Phi}$ and $\left(\mathrm{id}_{\widetilde{C}}\otimes \Theta\left(\mathcal{M}\right)\right)\ketbra{\Phi}{\Phi}  = \mathsf{C}_{\Theta\left(\mathcal{M}\right)}^{\Phi}$, where $\mathsf{C}_{\Theta\left(\mathcal{N}\right)}^{\Phi}$ and $\mathsf{C}_{\Theta\left(\mathcal{M}\right)}^{\Phi}$ are Choi states of channels $\Theta \left(\mathcal{N}\right)$ and $\Theta \left(\mathcal{M}\right)$, respectively. Thus, again,  we can write the relative entropy between quantum channels $\Theta(\mathcal{N})$ and $\Theta(\mathcal{M})$ as 
\begin{align}
    \chD*{\Theta(\mathcal{N})}{\Theta(\mathcal{M})} = \staD*{{\mathsf{C}}_{\Theta(\mathcal{N})}^{\Phi}}{{\mathsf{C}}_{\Theta(\mathcal{M})}^{\Phi}}. 
\end{align}
Also, from the data processing inequality, we have $ D\left[\mathcal{N}\Vert \mathcal{M}\right] \geq  D\left[\Theta\left(\mathcal{N}\right)\Vert \Theta \left(\mathcal{M}\right)\right]$, which implies the following inequality
\begin{align}
   \staD*{{\mathsf{C}}_{\mathcal{N}}^{\Psi}}{{\mathsf{C}}_{\mathcal{M}}^{\Psi}} \geq  \staD*{{\mathsf{C}}_{\Theta(\mathcal{N})}^{\Phi}}{{\mathsf{C}}_{\Theta(\mathcal{M})}^{\Phi}}.
\end{align}
Since the Choi states $\mathsf{C}_{\mathcal{N}}^{\Psi}$ and $\mathsf{C}_{\Theta(\mathcal{N})}^{\Phi}$, as well as $\mathsf{C}_{\mathcal{M}}^{\Psi}$  and $\mathsf{C}_{\Theta(\mathcal{N})}^{\Phi}$ are connected via the representing map $\mathfrak{T}$ of the superchannel $\Theta$ as $\mathsf{C}_{\Theta \left(\mathcal{N}\right)}^{\Phi} =  \mathfrak{T} \left(\mathsf{C}_{\mathcal{N}}^{\Psi}\right)$ and $\mathsf{C}_{\Theta \left(\mathcal{M}\right)}^{\Phi} =  \mathfrak{T} \left(\mathsf{C}_{\mathcal{M}}^{\Psi}\right)$, we can rewrite the above entropic inequality as 
\begin{align}
 \staD*{{\mathsf{C}}_{\mathcal{N}}^{\Psi}}{{\mathsf{C}}_{\mathcal{M}}^{\Psi}} \geq  \staD*{\mathfrak{T}\left({\mathsf{C}}_{\mathcal{N}}^{\Psi}\right)}{\mathfrak{T}\left( \mathsf{C}_{\mathcal{M}}^{\Psi}\right)}.
\end{align}
Recall that the representing map $\mathfrak{T}$ of superchannel $\Theta$ is, in general, a CP map that becomes trace-preserving if the adjoint of $\Theta$ is completely depolarising map preserving~(see Lemma~\ref{lemma5}). However, we can replace $\mathfrak{T}$ with the following trace-preserving map: $\mathfrak{T}'(\mathsf{X}):= \mathfrak{T}(\mathsf{X})+\left[\mathrm{tr}(\mathsf{X})-\mathrm{tr}(\mathfrak{T}(\mathsf{X}))\right] \sigma_{0}$, where $\mathsf{X} \in \mathcal{L}(A\otimes A')$ and $\sigma_{0}\in \mathcal{L}(C \otimes C')$ with $\tr(\sigma_0)=1$. For a quantum channel $\mathcal{N}$, the map $\mathfrak{T}'$ satisfies $\mathfrak{T}'(\mathsf{C}^{\Psi}_{\mathcal{N}})=\mathfrak{T}(\mathsf{C}^{\Psi}_{\mathcal{N}})$ (see Remark \ref{rem:choi-state-choi-operator-relation_1}). Since  $\Theta \in \mathsf{SCT}$ by assumption, it is always possible to choose $\sigma_0 \in \mathcal{L}(C \otimes C')$ such that $\mathfrak{T}'$ becomes a CPTP map (see Appendix \ref{quantum_channel_from_TP_map}). Thus, we have 
\begin{align}
    \staD*{{\mathsf{C}}_{\mathcal{N}}^{\Psi}}{{\mathsf{C}}_{\mathcal{M}}^{\Psi}} \geq  \staD*{\mathfrak{T}'\left({\mathsf{C}}_{\mathcal{N}}^{\Psi}\right)}{\mathfrak{T}'\left( \mathsf{C}_{\mathcal{M}}^{\Psi}\right)}.
\end{align}
Since $\mathfrak{T}'$ is a CPTP map and Choi states are quantum states, we can write the refined version of the above inequality using the result~\cite[Theorem 2.1]{Junge_2018} as 
\begin{equation}
    \staD*{{\mathsf{C}}_{\mathcal{N}}^{\Psi}}{{\mathsf{C}}_{\mathcal{M}}^{\Psi}} -  \staD*{\mathfrak{T}\left({\mathsf{C}}_{\mathcal{N}}^{\Psi}\right)}{\mathfrak{T}\left( \mathsf{C}_{\mathcal{M}}^{\Psi}\right)}
      \geq - \log F \left(\mathsf{C}^{\Psi}_{\mathcal{N}},\left(\mathcal{P}^{\rm{R}}\circ \mathfrak{T}'\right)\mathsf{C}^{\Psi}_{\mathcal{N}}\right).
    \end{equation}
Therefore, we have
    \begin{equation}
     \chD*{\mathcal{N}}{\mathcal{M}} -\chD*{\Theta\left(\mathcal{N}\right)}{\Theta \left(\mathcal{M}\right)}  \geq - \log F \left(\mathsf{C}^{\Psi}_{\mathcal{N}},\left(\mathcal{P}^{\rm{R}}\circ \mathfrak{T}'\right)\mathsf{C}^{\Psi}_{\mathcal{N}}\right),
    \end{equation}
    which concludes our proof.   
    
\section{Proof of monotonicity of \texorpdfstring{$\mathbf{D}^{(n)}$}{}}\label{proof_of_monotonicity_of_dn}
We restrict ourselves to the class of $(n+1)^{\text{th}}$ order processes which satisfy the following equation.
\begin{equation}
    \Phi^{(n+1)}\left(\Theta^{(n)}\right)= \Gamma^{(n)}_{\text{post}} \circ \left(\id \otimes ~\Theta^{(n)}  \right)\circ \Lambda^{(n)}_{\text{pre}}.\label{equ:nth_order_rep}
\end{equation}
For this class of $(n+1)^{\text{th}}$ order processes, our definition of generalized divergence between two $n^{\text{th}}$ order processes satisfies the monotonicity relation, Eq.\eqref{EEE_1}. The proof goes as follows. Consider two $n^{\text{th}}$ order processes $\Theta^{(n)}$ and $\Gamma^{(n)}$ such that action of any $(n+1)^{\text{th}}$ order process, say $\Phi^{(n+1)}$, can be written in terms of concatenation of $n^{\text{th}}$ order processes, given by Eq.~\eqref{equ:nth_order_rep}. Now, we show that in this scenario, data processing inequality holds, i.e., $\BnchD*{\Theta^{(n)}}{\Gamma^{(n)}} \geq \BnchD*{\Phi^{(n+1)}\left(\Theta^{(n)}\right)}{\Phi^{(n+1)}\left(\Gamma^{(n)}\right)}$. Using Eq.~\eqref{equ:divergence_of_nth_order_process}, we can write  

\begin{align}
  &\BnchD*{\Phi^{(n+1)}\left(\Theta^{(n)}\right)}{\Phi^{(n+1)}\left(\Gamma^{(n)}\right)}\nonumber \\
  &=\sup_{\Omega^{(n-1)} \in  \left(\mathcal{W}^{(n-1)} \otimes \mathcal{V}^{(n-1)}\right)_{+}  }\BnachD*{\left(\mathrm{id} \otimes\Phi^{(n+1)}\left(\Theta^{(n)}\right) \right)\left(\Omega^{(n-1)}\right)}{\left(\mathrm{id} \otimes\Phi^{(n+1)}\left(\Gamma^{(n)}\right) \right)\left(\Omega^{(n-1)}\right)}. \label{Binod_1}
\end{align}
Using Eq.~\eqref{equ:nth_order_rep}, we can write $\left(\mathrm{id} \otimes \Phi^{(n+1)}\left(\Theta^{(n)}\right) \right)\left(\Omega^{(n-1)}\right) = \left(\mathrm{id} \otimes\Gamma_{\text{post}}^{(n)}\right)\left(\mathrm{id} \otimes\left(\mathrm{id} \otimes\Theta^{(n)}\right)\right) \left(\widetilde{\Omega}^{(n-1)}\right)$, where we have defined $\widetilde{\Omega}^{(n-1)}:=\left(\mathrm{id} \otimes\Lambda_{\text{pre}}^{(n)}\right) \Omega^{(n-1)}$, and similarly we also have $\left(\mathrm{id} \otimes\Phi^{(n+1)}\left(\Gamma^{(n)}\right)\right)\left(\Omega^{(n-1)}\right) = \left(\mathrm{id} \otimes\Gamma_{\text{post}}^{(n)}\right)\left(\mathrm{id}\otimes\left(\mathrm{id} \otimes\Gamma^{(n)}\right)\right) \left(\widetilde{\Omega}^{(n-1)}\right)$.
Now, using these on the right-hand side of Eq.~(\ref{Binod_1}), we obtain

\begin{align}
   &\BnchD*{\Phi^{(n+1)}\left(\Theta^{(n)}\right)}{\Phi^{(n+1)}\left( \Gamma^{(n)}\right)}\nonumber\\
   &\leq \sup_{\widetilde{\Omega}^{(n-1)} } \BnachD*{\left(\mathrm{id} \otimes\Gamma_{\text{post}}^{(n)}\right)\left(\mathrm{id} \otimes\left(\mathrm{id} \otimes\Theta^{(n)}\right)\right) \left(\widetilde{\Omega}^{(n-1)}\right)}{\left(\mathrm{id} \otimes\Gamma_{\text{post}}^{(n)}\right)\left(\mathrm{id}\otimes\left(\mathrm{id} \otimes\Gamma^{(n)}\right)\right) \left(\widetilde{\Omega}^{(n-1)}\right)}. \label{Binod_4}
\end{align}
In the above inequality, $\BnachD*{\cdot}{ \cdot}$ is a generalized divergence by our definition of the divergence functional of $n^{\text{th}}$ order process, thus using the fact that if $\Gamma_{\text{post}}^{(n)}$ is an $n^{\text{th}}$ order physical process then $\mathrm{id} \otimes \Gamma_{\text{post}}^{(n)} $ is also a valid $n^{\text{th}}$  order physical process, we have:
\begin{align}
   & \BnachD*{\left(\mathrm{id} \otimes\Gamma_{\text{post}}^{(n)}\right)\left(\mathrm{id} \otimes\left(\mathrm{id} \otimes\Theta^{(n)}\right)\right) \left(\widetilde{\Omega}^{(n-1)}\right)}{\left(\mathrm{id} \otimes\Gamma_{\text{post}}^{(n)}\right)\left(\mathrm{id}\otimes\left(\mathrm{id} \otimes\Gamma^{(n)}\right)\right) \left(\widetilde{\Omega}^{n-1}\right)} \nonumber\\
    &\leq \BnachD*{\left(\mathrm{id}\otimes\left(\mathrm{id}\otimes\Theta^{(n)} \right)\right) \left(\widetilde{\Omega}^{(n-1)}\right)}{\left(\mathrm{id}\otimes\left(\mathrm{id}\otimes\Gamma^{(n)} \right)\right) \left(\widetilde{\Omega}^{(n-1)}\right)}.
\end{align}
Using this  inequality  in (\ref{Binod_4}), we finally have
\begin{align}
\nonumber
    \BnchD*{\Theta^{(n)}}{\Gamma^{(n)}} \geq \BnchD*{\Phi^{(n+1)}\left(\Theta^{(n)}\right)}{\Phi^{(n+1)}\left(\Gamma^{(n)}\right)}.
\end{align}

\bibliography{ref}

\end{document}